\documentclass{article}

\usepackage{kr}
\usepackage{natbib}

\usepackage{times}

\usepackage[utf8]{inputenc} 
\usepackage[T1]{fontenc}    
\usepackage[colorlinks=true,
linkcolor=purple,
citecolor=purple]{hyperref} 

\usepackage{graphicx}
\usepackage[small]{caption}
\usepackage{subcaption}
\usepackage{todonotes}
\usepackage{adjustbox}
\usepackage{bm}
\usepackage{listings}
\usepackage{url}            
\usepackage{booktabs}       
\usepackage{amsfonts}       
\usepackage{nicefrac}       
\usepackage{microtype}      
\usepackage{xcolor}         
\usepackage{amsmath,amsthm,amssymb}
\usepackage{tikz} 
\usepackage{times}
\usepackage{soul}
\usepackage{algorithm}
\usepackage{algorithmic}
\usepackage{enumitem}
\usepackage{cleveref}
\usepackage{float}

\usetikzlibrary{arrows,positioning,backgrounds} 
\usetikzlibrary{plotmarks}
\usepackage{pgfplots}
\pgfplotsset{compat=1.18}

\pgfkeys{/pgf/number format/.cd,
    set thousands separator={}}

\usepackage[altpo]{backnaur}
\newcommand{\bh}{{\rm BH}}

\newcommand{\ptime}{{\rm P}}
\newcommand{\fptime}{{\rm FP}}
\newcommand{\np}{{\rm NP}}
\newcommand{\stp}{\Sigma_2^p}
\newcommand{\conp}{{\rm coNP}}
\newcommand{\fpnp}{\fptime$^{\np}$}
\newcommand{\es}{\mathbf{e}}
\newcommand{\M}{\mathcal{M}}
\newcommand{\T}{\mathcal{T}}
\newcommand{\dm}{{\rm dim}}
\newcommand\false{\mathbf{false}}
\newcommand\true{\mathbf{true}}
\newcommand{\dt}{\mathsf{DTree}}

\newcommand{\pos}{\textsc{Pos}}
\newcommand{\npos}{\textsc{Neg}}

\newcommand{\lel}{\preceq}
\newcommand{\lnel}{\prec}

\newcommand{\led}{\textsc{LEH}}
\newcommand{\full}{\textsc{Full}}
\newcommand{\allpos}{\textsc{AllPos}}
\newcommand{\allneg}{\textsc{AllNeg}}
\newcommand{\cons}{\textsc{Cons}}
\newcommand{\meet}{\textsc{GLB}}
\newcommand{\node}{\textsc{Node}}

\newcommand{\posl}{\textsc{PosLeaf}}
\newcommand{\negl}{\textsc{NegLeaf}}
\newcommand{\leaf}{\textsc{Leaf}}

\newcommand{\comp}{\textsc{Comp}}

\newcommand{\nf}{\textsc{NF}}
\newcommand{\pred}{\textsc{Pr}}
\newcommand{\minf}{\text{\rm min}}

\newcommand{\astruct}{\mathfrak{A}}

\newcommand{\EF}{{Ehrenfeucht-Fra\"\i ss\'e\ }}
\newcommand{\FO}{{\rm FO}} 	
\newcommand{\qr}{{\rm qr}}
\newcommand{\dom}{{\rm dom}}  

\newcommand{\foil}{\textsc{FOIL}}

\newcommand{\dtfoil}{\textsc{DT-FOIL}}

\newcommand{\qdtfoil}{\textsc{Q-DT-FOIL}}
\newcommand{\optdtfoil}{\textsc{Opt-DT-FOIL}}

\newcommand{\sem}{\textsc{DFS}}
\newcommand{\minsem}{\textsc{MinimalDFS}}

\newcommand{\opp}{\textsc{Opp}}
\newcommand{\suf}{\textsc{SUF}}
\newcommand{\sr}{\textsc{SR}}
\newcommand{\minsr}{\textsc{MinimalSR}}
\newcommand{\msr}{\textsc{MinimumSR}}
\newcommand{\mcr}{\textsc{MinimumCR}}
\newcommand{\mca}{\textsc{MaximumCA}}

\newcommand{\bound}{\textsc{Bound}}

\newcommand{\csr}{\textsc{CSR}}
\newcommand{\nsr}{\textsc{NSR}}

\newcommand{\sat}{\textsc{SAT}}
\newcommand{\unsat}{\textsc{UNSAT}}
\newcommand{\sate}{\textsc{SATExp}}

\newcommand{\dnf}{\textsc{DNF}}

\newcommand{\zerolabel}{{\footnotesize 0}}
\newcommand{\onelabel}{{\footnotesize 1}}

\newcommand{\logicEvaluation}{\textsc{Eval}}
\newcommand{\logicComputation}{\textsc{Comp}}

\newtheorem{theorem}{Theorem}
\newtheorem{proposition}[theorem]{Proposition}

\newtheorem{corollary}[theorem]{Corollary}

\newtheorem{lemma}[theorem]{Lemma}

\definecolor{codegreen}{rgb}{0,0.6,0}
\definecolor{codegray}{rgb}{0.5,0.5,0.5}
\definecolor{codepurple}{rgb}{0.58,0,0.82}
\definecolor{backcolour}{rgb}{0.95,0.95,0.92}

\lstdefinestyle{mystyle}{
    backgroundcolor=\color{backcolour},   
    commentstyle=\color{codegreen},
    keywordstyle=\color{magenta},
    numberstyle=\tiny\color{codegray},
    stringstyle=\color{codepurple},
    basicstyle=\ttfamily\footnotesize,
    breakatwhitespace=false,         
    breaklines=true,                 
    captionpos=b,                    
    keepspaces=true,                 
    showspaces=false,                
    showstringspaces=false,
    showtabs=false,                  
    tabsize=2
}

\tikzset{
    rt/.style={
		rectangle,
		fill = white,
		draw=black, 
		text centered,
		inner sep=0.5ex
		},
    rtt/.style={ 
    	rt,
    	inner sep=0.1ex
    	},
    ert/.style={ 
     	rt,
     	dashed
     	}, 
    ertt/.style={ 
        rtt,
        dashed
        }, 
    rect/.style={ 
        rectangle,
        fill = white,
        rounded corners,
        draw=black, 
        text centered,
        inner sep=0.8ex
        },
    rectw/.style={
        rect,
        draw=white
        },
    erect/.style={ 
    	rect,
    	dashed
    	},
    erectw/.style={ 
     	rectw,
     	dashed
     	},
    arrout/.style={
           ->,
           -latex,
           },
    arrin/.style={
           <-,
           latex-,
           },
    arrb/.style={
           <->,
           >=latex,
           }
}

\title{A Uniform Language to Explain Decision Trees}

\author{%
Marcelo Arenas$^{1}$\and
Pablo Barcel\'o$^{1}$\and
Diego Bustamante$^{1}$\and
Jose Caraball$^{1}$ \and
Bernardo Subercaseaux$^2$
\affiliations
$^1$PUC Chile\\
$^2$Carnegie Mellon University
}

\begin{document}
	
	\maketitle
	
 \begin{abstract}
	The formal XAI community has studied a plethora of interpretability queries aiming to understand the classifications made by decision trees.
	 However, a more uniform understanding of what questions we can hope to answer about these models, traditionally deemed to be easily interpretable, has remained elusive.
	 In an initial attempt to understand uniform languages for interpretability, \citet{DBLP:conf/nips/ArenasBBPS21} proposed~\foil, a logic for explaining black-box ML models, and showed that it can express a variety of interpretability queries. However, we show that~\foil~is limited in two important senses: (i) it is not expressive enough to capture some crucial queries, 
	 and (ii) its model agnostic nature results in a high computational complexity for decision trees. 
    In this paper, we carefully craft two fragments of first-order logic that allow for efficiently interpreting decision trees:~\qdtfoil~and its optimization variant~\optdtfoil.
	 We show that our proposed logics can express not only a variety of interpretability queries considered by previous literature, but also elegantly allows users to specify different objectives the sought explanations should optimize for. 
	Using finite model-theoretic techniques, we show that the different ingredients of~\qdtfoil~are necessary for its expressiveness, and yet that queries in~\qdtfoil~can be evaluated with a polynomial number of queries to a SAT solver, as well as their optimization versions in~\optdtfoil.  
	Besides our theoretical results, we provide a SAT-based implementation of the evaluation for~\optdtfoil~that is performant on industry-size decision trees.
 \end{abstract}

\section{Introduction}\label{sec:intro}
\paragraph{Formal XAI.} The increasing need to comprehend the decisions made by machine learning (ML) models has fostered a large body of research in {\em explainable AI} (XAI) methods~\citep{molnar2022}, leading to the introduction of numerous queries and scores that aim to explain the predictions produced by such models. 
Within the wide variety of methods and subareas in XAI, our work is part of the \emph{formal XAI} approach~\citep{formal-xai,darwiche2023logic, Marques-Silva_2023}, which aims to ground the study of explainability in a mathematical framework.
In this line, our work leverages ideas from finite model theory \citep{fmt-book} to study the complexity and expressiveness of a fragment of first order logic tailored to explain decision trees, as well as ideas from automated reasoning to produce efficient CNF encodings for evaluating these queries through SAT solvers.


\paragraph{Decision Trees and Explanations.}
Decision trees 
are a very popular choice of ML models for tabular data, and one of the standard arguments in favor of their use is their supposed \emph{interpretability}~\citep{gunningDARPAExplainableArtificial2019, molnar2022, DBLP:journals/corr/Lipton16a}. 
However, the formal XAI community has shown that the interpretability of decision trees is nuanced, and that even for these apparently simple models, some kinds of explanations are easy to produce while others are computationally challenging~\citep{audemard2021explanatory, NEURIPS2020_b1adda14, DBLP:journals/corr/abs-2207-12213, DBLP:journals/corr/abs-2010-11034, DBLP:journals/jair/IzzaIM22}.
 Let us immediately present some examples of queries (illustrated in~\Cref{fig:query-examples}) that have been considered in the literature~\citep{Darwiche_Hirth_2020,NEURIPS2020_b1adda14,DBLP:journals/jair/IzzaIM22}.
\begin{itemize}
    \item \emph{Minimal/Minimum Sufficient Reasons}:  Given an input instance $\es$ and a decision tree $\T$, what is the \emph{smallest} subset $S$ of features in $\es$ such that the classification $\T(\es)$ is preserved regardless of the values of the features outside $S$? The notion of \emph{smallest} can be defined either in terms of set containment (\emph{minimal}) or cardinality (\emph{minimum}).
    \item \emph{Minimum Change Required/Maximum Change Allowed}: Given an input instance $\es$ and a decision tree $\T$, what is the smallest set of features that must be changed in $\es$ to change the classification $\T(\es)$? Conversely, what is the largest set of features that can be changed in $\es$ without changing the classification $\T(\es)$?
    \item \emph{Determinant Feature Set}: Given a decision tree $\T$, what is the smallest set of features that, when fixed, determines the classification of any input instance?
\end{itemize}

\begin{figure*}[ht]
\begin{subfigure}{0.24\textwidth}
    \centering
\includegraphics[scale=0.1]{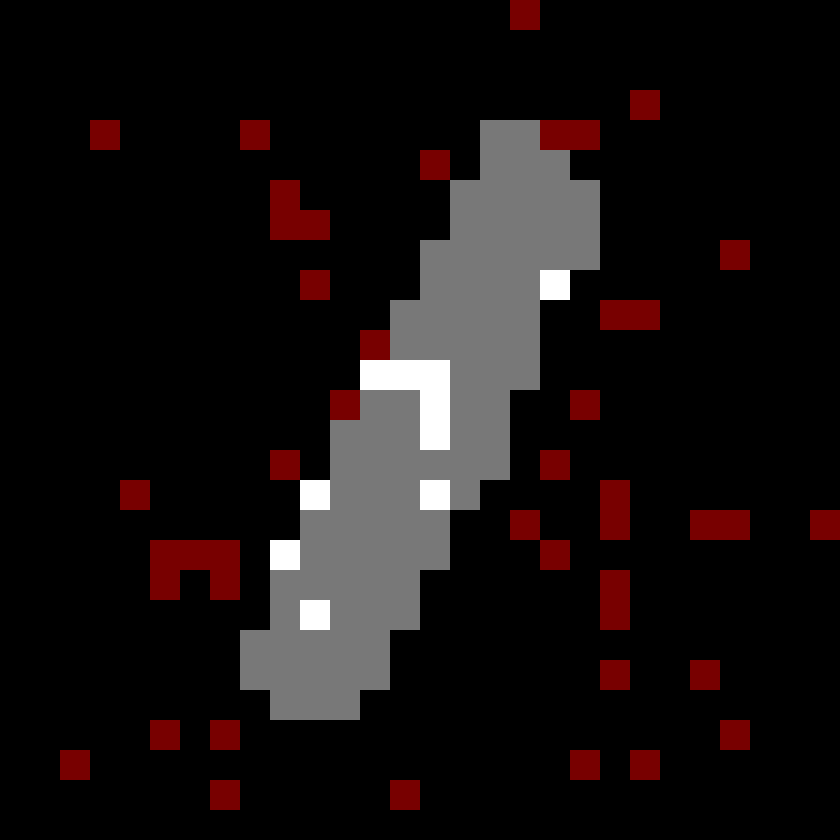}
\caption{Minimum Sufficient Reason}\label{subfig:msr}
\end{subfigure}
\begin{subfigure}{0.24\textwidth}
    \centering
    \includegraphics[scale=0.1]{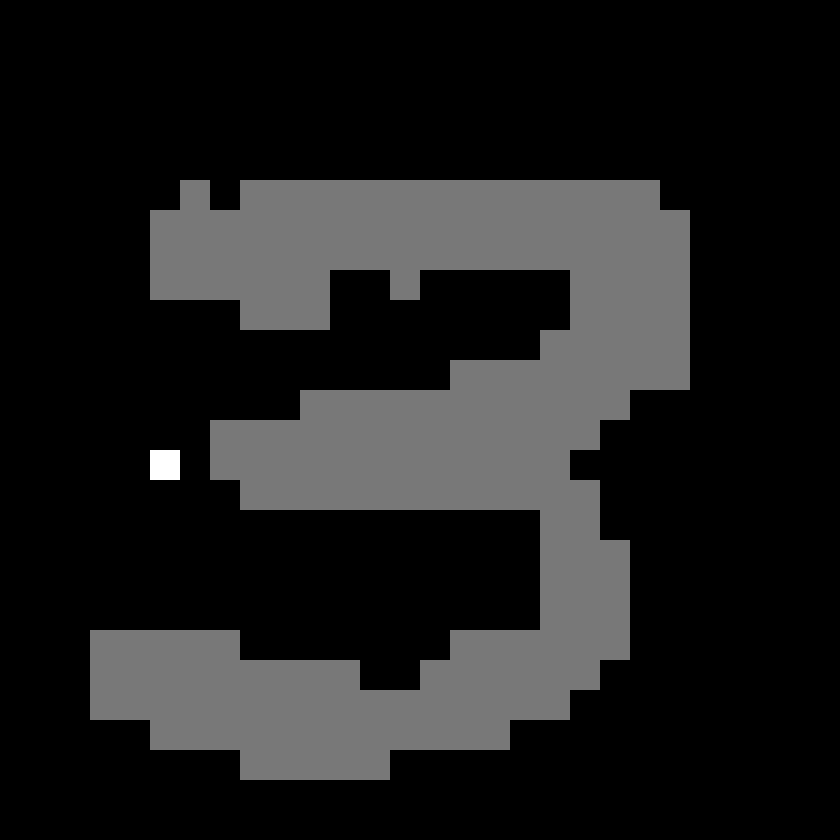}
\caption{Minimum Change Required}\label{subfig:mcr}
\end{subfigure}
\begin{subfigure}{0.24\textwidth}
    \centering
    \includegraphics[scale=0.1]{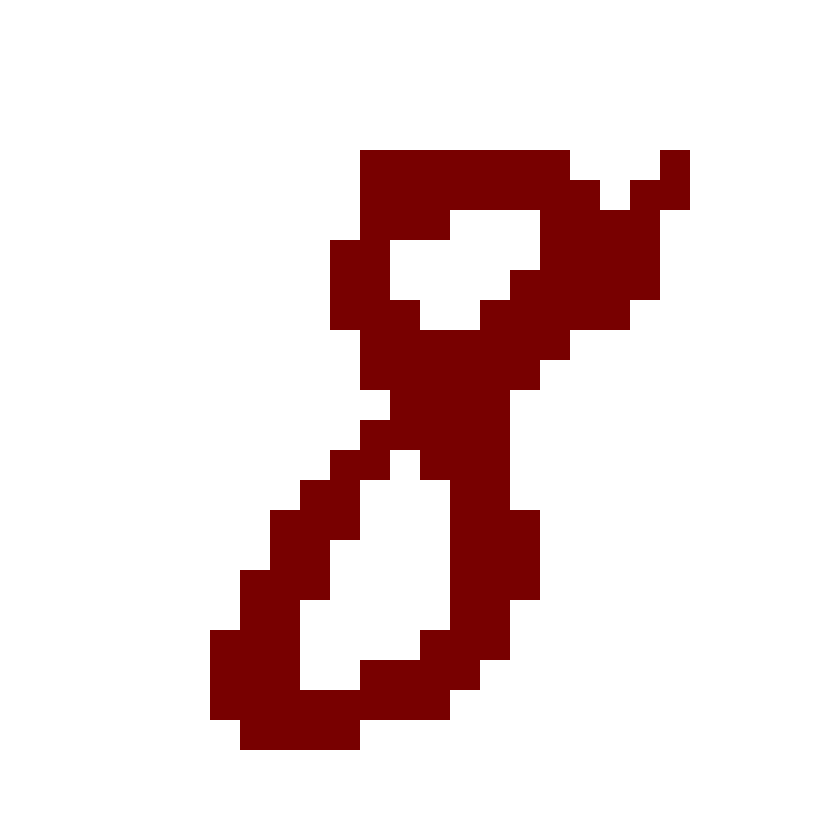}
\caption{Maximum Change Allowed}\label{subfig:mca}
\end{subfigure}
\begin{subfigure}{0.24\textwidth}
    \centering
    \includegraphics[scale=0.1]{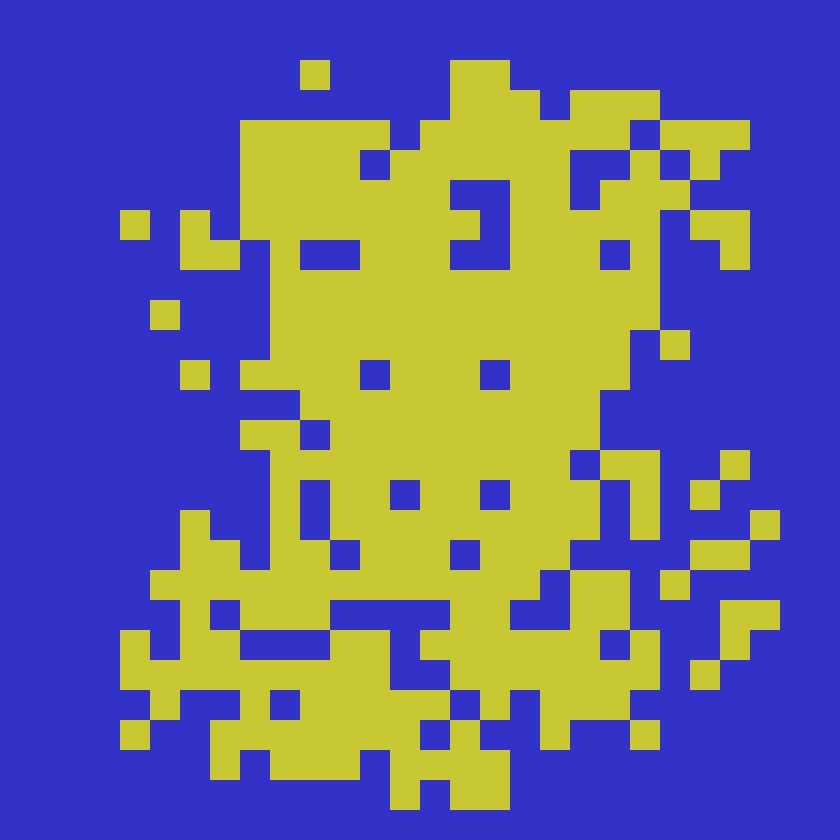}
\caption{Determinant Feature Set}\label{subfig:dfs}
\end{subfigure}
\caption{Illustration of different explanations for decision trees of $500$ leaves over the binarized MNIST dataset~\citep{deng2012mnist}. The explanations are obtained using our implementation over the~\optdtfoil~queries described in~\Cref{sec:opt-dt-foil}. \Cref{sub@subfig:msr} displays a minimum sufficient reason for an image classified as \texttt{1}, where the white pixels of the original image that are part of the explanation are highlighted in perfect white, and the black pixels that are part of the explanation are highlighted in red. Arguably, this explanation reveals that the model, trained to recognize digit~\texttt{1}, has learned to detect a slanted vertical stripe of white pixels, sorrounded by black pixels. \Cref{sub@subfig:mcr} shows that adding a single white pixel to the image of a~\texttt{3} is enough to change its classification (cf.~\emph{one-pixel attacks}~\citep{Su_2019}). Interestingly,~\Cref{sub@subfig:mca} shows that one can simultaneously flip all pixels on an image of digit~\texttt{8} while retaining its classification, showing the model is somewhat invariant to the roles of white and black pixels in the original image. Finally,~\Cref{sub@subfig:dfs} shows that a subset of the pixels, around the center of MNIST images (colored in yellow), is enough to determine the veredicts of a model trained to detect digit \texttt{1}. As an application, one could leverage this knowledge to reduce the dimensionality of the dataset by cropping the borders. }\label{fig:query-examples}
\end{figure*}
\paragraph{Motivation for Interpretability Languages.}
The variety of interpretability queries and scores that have been proposed in the literature can be seen as a call for \emph{interpretability languages} in which such queries could be expressed in a uniform manner. 
We highlight two reasons for the development of interpretability languages:

\begin{itemize}

\item \emph{No Silver-Bullet Principle:} The variety of interpretability queries seems to reflect the fact that no single kind of explanation is always the best.
Moreover, it is often not a single query or score, but a combination of them, that provides the best explanation~\citep{doshivelez2017rigorous, Marques-Silva_Ignatiev_2023}. In the same line, it has been shown that 
some widely used explainability scores, believed to be theoretically mature and robust, may behave counterintuitively in certain situations~\citep{inms-corr19,DBLP:conf/ijcai/Ignatiev20,DBLP:journals/corr/abs-1910-02065,DBLP:conf/aies/SlackHJSL20,DBLP:conf/icml/KumarVSF20, huang2023inadequacy}. 

\item \emph{A Uniform Understanding of Interpretability:} As posed by~\citep{NEURIPS2020_b1adda14}, the computational complexity of interpretability queries on a class of models (e.g., decision trees, neural networks) can be seen as a measure of their interpretability. However, existing analyses (see also~\citep{alfano2024evenif, lin2024complexity}) rely on the particular queries being chosen. In contrast, by analyzing the complexity of evaluating interpretability queries in a uniform language, we can obtain a more general understanding of the interpretability of a class of models. 
\end{itemize}

A first step toward ML interpretability languages was carried out by~\citet{DBLP:conf/nips/ArenasBBPS21}, who designed a 
simple explainability language based on first-order logic, called $\foil$ ({\em first-order interpretability logic}), 
that was able to express some basic explainability queries. 	
However, as noted by~\citet{DBLP:conf/nips/ArenasBBPS21}, the primary purpose of $\foil$ was not to serve as a practical explainability language but as a  foundation upon which such languages could be constructed. To date, nevertheless, we have no complete understanding of why $\foil$ is not a good practical language for explainability, nor what needs to be added to it in order to make it a more effective tool for performing such tasks. To gain a deeper understanding of this issue, we introduce two desiderata that any language used for explainability queries should meet: 
\begin{itemize}
\item {\em Rich expressive power:} The language should be able to express a broad range of explainability queries used in practice. Some desirable characteristics in terms of expresiveness are the combination of queries (e.g., is there a \emph{minimum sufficient reason} common to two input instances $\es$ and $\es'$?), and the possibility of expressing \emph{preferred} explanations that contain specific features of interest to the user \citep{ijcai2022p91, alfano2024evenif}. 
\item {\em Efficiency:} The complexity of the language used to express explainability queries must be manageable. Note that this does not necessarily imply that the evaluation should take polynomial time; SAT solvers are a mature technology that allows to solve many $\np$-hard problems in practice, 
and has been effective in computing explanations for various ML models~\citep{DBLP:conf/ijcai/Izza021,DBLP:conf/cp/YuISB20,DBLP:conf/sat/IgnatievS21}. In this sense, a language whose evaluation requires a small number of calls to a SAT solver can still be practical, which is the case in our work. 
\end{itemize}

\paragraph{Theoretical Contributions.}
We start by assessing the suitability of $\foil$ as an explainability language. 
Regarding its expressive power, we show that there are crucial explainability queries that cannot be expressed in this language, e.g., the query {\em minimum sufficient reason}~\citep{shih2018symbolic,NEURIPS2020_b1adda14}, as described earlier, cannot be expressed in $\foil$. 
Regarding computational complexity, we show that, under some widely believed complexity assumptions, queries expressed in $\foil$ cannot be evaluated with a polynomial number of calls to an $\np$ oracle. Specifically, we show that the $\foil$ evaluation problem over decision trees is hard for each level of the polynomial hierarchy, which goes well beyond the problems that can be solved with a polynomial number of calls to an~$\np$ oracle. 

Considering these limitations, we pursue progress along two key avenues: First, we define an extension of $\foil$ that can capture many of the explainability notions found in practice. Then, we seek a meaningful restriction of it that maintains expressive richness while remaining compatible with evaluation via SAT solver technology.
\begin{itemize} 
\item 
Regarding the extension of $\foil$, we enhance it with a simple predicate $\preceq$ that enables reasoning about the cardinalities of sets of features. This addition allows us to express minimum sufficient reason  and other cardinality-based inquiries like \emph{minimum change required}.
\item 
The high complexity of $\foil + \{\preceq\}$ motivates the design of~\dtfoil, a similar logic to $\foil  + \{\preceq\}$ that is tailored specifically for decision trees, which makes it evaluation tractable. In particular, by \emph{guarding quantification}, formulas in~\dtfoil~can be evaluated in polynomial time.  Unfortunately, this guarded quantification prevents~\dtfoil~from expressing queries relating to minimization, or more in general, that require unbounded quantification. From this observation, we study two possible ways forward: (i) \qdtfoil, a simple extension of~\dtfoil~that allows for unbounded quantification without alternations, and (ii) \optdtfoil, an \emph{optimization} version of~\dtfoil. 
We show that the evaluation problem for~\qdtfoil~lies in the \emph{boolean hierarchy} (\bh), thus requiring a constant number of calls to an~\np~oracle (e.g., a SAT solver), and that the evaluation problem for~\optdtfoil~is in $\mathrm{P}^\np$, thus requiring a polynomial number of calls to an~\np~oracle.
\end{itemize}

\paragraph{Implementation.}
We provide a partial implementation of the evaluation of $\qdtfoil$ and
$\optdtfoil$ queries over decision trees, leveraging modern SAT-solvers and
automated reasoning techniques for obtaining efficient CNF encodings. The
fragment of $\qdtfoil$ and $\optdtfoil$ queries that we can evaluate includes
all examples of queries presented in this paper. The reason for our
implementation being limited is that, part of the theoretical evaluation
algorithm in $\mathrm{P}^\np$ for certain queries relies on a polynomial-time
subroutine stemming from finite model theory whose constant factor is
prohibitively large. Nonetheless, for the subset of queries that we can
evaluate, we show that our implementation runs in the order of magnitude of
seconds over decision trees with thousands of nodes and hundreds of features,
thus making it suitable for practical use
(cf.~\citep{DBLP:journals/corr/abs-2010-11034,
gomesmantovaniBetterTreesEmpirical2024}).

\begin{figure}
		\begin{minipage}{0.9\columnwidth}
		    \centering
        \def\vstep{1.2}
        \begin{tikzpicture}[scale=0.9, transform shape]
            \node[draw, circle, thick, fill=cyan!60!white] (R) at (0, 0) {$2$};
            \node[draw, circle, thick, fill=red!60!white] (Rl) at (-1.5, {-0.8*\vstep}) {$1$};
            \node[] (R1node) at (-2.1, {-0.8*\vstep}) {$u$:};
            \node[draw, circle, thick, fill=orange!60!white] (Rr) at (1.5, {-0.8*\vstep}) {$4$};
            \node[draw, circle, thick, fill=green!60!white] (Rll) at (-2.2, {-2*\vstep}) {$3$};
            
             \node[draw, circle, thick,fill=orange!60!white] (Rlr) at (-0.8, {-2*\vstep}) {$4$};
             
             \node[draw, circle, thick, fill=green!60!white] (Rrl) at (0.8, {-2*\vstep}) {$3$};
            
             \node[] (Rrr) at (2.2, {-2*\vstep}) {$\true$};
             
             \node[] (Rlll) at (-2.9, {-3*\vstep}) {$\true$};
             
              \node[draw, circle, thick, fill=orange!60!white] (Rllr) at (-1.9, {-3*\vstep}) {$4$};
              
              \node[] (Rlrl) at (-1, {-3*\vstep}) {$\false$};
              
               \node[] (Rlrr) at (-0.1, {-3*\vstep}) {$\true$};
               
               \node[draw, circle, thick, fill=red!60!white] (Rrll) at (0.9, {-3*\vstep}) {$1$};
              
               \node[] (Rrlr) at (2.0, {-3*\vstep}) {$\false$};
               
                  \node[] (Rllrl) at (-2.3, {-4*\vstep}) {$\false$};
              
               \node[] (Rllrr) at (-1.3, {-4*\vstep}) {$\true$};
            
            \node[] (Rrlll) at (0.4, {-4*\vstep}) {$\true$};
            \node[] (Rrlllnode) at (-0.2, {-4*\vstep}) {$v$:};
            
             \node[] (Rrllr) at (1.5, {-4*\vstep}) {$\false$};
            
             \draw[->, thick] (R) -- (Rl) node[midway, above] {\zerolabel};
             \draw[->, thick] (R) -- (Rr) node[midway, above] {\onelabel};
             \draw[->, thick] (Rl) -- (Rll) node[midway, left, yshift=0.1cm] {\zerolabel};
             \draw[->, thick] (Rl) -- (Rlr) node[midway, right, yshift=0.1cm] {\onelabel};
             
             \draw[->, thick] (Rll) -- (Rlll) node[midway, left, yshift=0.1cm] {\zerolabel};
              \draw[->, thick] (Rll) -- (Rllr) node[midway, right, yshift=0.1cm] {\onelabel};
              
               \draw[->, thick] (Rr) -- (Rrl) node[midway, left, yshift=0.1cm] {\zerolabel};
               \draw[->, thick] (Rr) -- (Rrr) node[midway, right, yshift=0.1cm] {\onelabel};
             
              \draw[->, thick] (Rrl) -- (Rrll) node[midway, left, ] {\zerolabel};
               \draw[->, thick] (Rrl) -- (Rrlr) node[midway, right, yshift=0.1cm] {\onelabel};
               
               \draw[->, thick] (Rlr) -- (Rlrl) node[midway, left, yshift=0.1cm] {\zerolabel};
               \draw[->, thick] (Rlr) -- (Rlrr) node[midway, right, yshift=0.1cm] {\onelabel};
               
                \draw[->, thick] (Rllr) -- (Rllrl) node[midway, left, yshift=0.1cm, xshift=0.05cm] {\zerolabel};
               \draw[->, thick] (Rllr) -- (Rllrr) node[midway, right, yshift=0.1cm, xshift=-0.05cm] {\onelabel};
            
            	\draw[->, thick] (Rrll) -- (Rrlll) node[midway, left, yshift=0.1cm, xshift=0.05cm] {\zerolabel};
            	
            	\draw[->, thick] (Rrll) -- (Rrllr) node[midway, right, yshift=0.1cm, xshift=-0.05cm] {\onelabel};

        \end{tikzpicture}
        \caption{Example of a decision tree of dimension $4$.}
        \label{fig:tree-example}
		\end{minipage}
\end{figure}

\section{Background}\label{sec:background}
	
	
	
	\paragraph{Models and instances.} We use an abstract notion of a model of dimension $n$,
	and define it as a Boolean function $\M :  \{0,1\}^n \to \{0,
	1\}$.\footnote{We focus on Boolean  models, which is common in 
	formal XAI research \citep{DBLP:journals/jair/WaldchenMHK21,audemard2021explanatory, 10497107}.} 
    We write $\dm(\M)$ for the dimension of a model
	$\M$. A {\em partial instance} of dimension $n$ is a tuple $\es \in
	\{0,1,\bot\}^n$, where $\bot$ is used to represent undefined features.
	We define $\es_\bot =  \{i \in \{1,\dots,n\} \mid \es[i] = \bot\}$.  
	An {\em instance} of dimension $n$ is a tuple $\es \in \{0,1\}^n$, that is, a partial instance without undefined features. 
	
	Given partial instances $\es_1$, $\es_2$ of
	dimension $n$, we say that $\es_1$ is {\em subsumed} by $\es_2$ if, and only if, $\es_1[i] = \es_2[i]$, for every $i \in \{1, \ldots, n\}$ with $\es_1[i] \neq \bot$.  	
	That is, it is possible to obtain $\es_2$ from $\es_1$ by replacing
	some unknown values.  For example, $(1,\bot)$ is subsumed by $(1,0)$, but it is not subsumed by $(0,0)$. 
	A partial instance $\es$ can be
	seen as a compact representation of the set of instances $\es'$
	such that $\es$ is subsumed by $\es'$, where such instances $\es'$ are
	called the {\em completions} of $\es$.
	
	
	\paragraph{Decision trees.}  A \emph{decision tree} over instances of dimension
	$n$ is a rooted directed tree~$\T$ with labels on edges and
	nodes such that: (i) each leaf is labeled with~$\true$ or
	$\false$; (ii) each internal node (a node that is not a leaf)
	is labeled with a feature $i \in \{1,\dots,n\}$; (iii) each
	internal node has two outgoing edges, one labeled~$0$ and the
	another one labeled~$1$; and (iv) in every path from the root
	to a leaf, no two nodes on that path have the same label.
	Every instance $\es \in \{0,1\}^n$ defines a unique
	path $\pi_\es = u_1 \cdots u_k$ from the root $u_1$ to a leaf
	$u_k$ of $\T$ such that: if the label of $u_i$ is
	$j \in \{1,\dots,n\}$, where $i \in \{1, \ldots, k-1\}$, then
	the edge from $u_i$ to $u_{i+1}$ is labeled with
	$\es[j]$. Further, the instance $\es$ is positive, denoted by
	$\T(\es) = 1$, if the label of $u_k$ is~$\true$; otherwise the
	instance $\es$ is negative, which is denoted by $\T(\es) =
	0$. For example, for the decision tree $\T$ in
	Figure~\ref{fig:tree-example} and instances $\es_1 =
	(0,0,1,1)$ and $\es_2 = (0,1,1,0)$, it holds that $\T(\es_1) =1$
	and $\T(\es_2) = 0$.
	
	\subsection{First Order Interpretability Logic (\texorpdfstring{$\foil$}{FOIL})}
Our work is inspired by the {\em first-order interpretability logic} ($\foil$) 
	\citep{DBLP:conf/nips/ArenasBBPS21}, which is a simple explainability language rooted in first-order logic.  
	In particular, $\foil$ is nothing else than first-order
logic over two relations on the set of partial instances of a given
dimension: A unary relation $\pos$ which indicates the value
of an instance in a model, and a binary relation $\subseteq$ that represents 
the subsumption relation among partial instances.
	
	Given a vocabulary $\sigma$ consisting of relations $R_1$, $\ldots$, $R_\ell$, recall that a structure $\astruct$ over $\sigma$ consists of a domain, where quantifiers are instantiated, and an interpretation for each relation $R_i$. Moreover, given a first-order formula $\varphi$ defined over the vocabulary $\sigma$, we write $\varphi(x_1, \ldots, x_k)$ to indicate that $\{x_1, \ldots, x_k\}$ is the set of free variables of $\varphi$. Finally, given a structure $\astruct$ over the vocabulary $\sigma$ and elements $a_1$, $\ldots$, $a_k$ in the domain of $\astruct$, we use $\astruct \models \varphi(a_1, \ldots, a_k)$ to indicate that formula $\varphi$ is satisfied by $\astruct$ when each variable $x_i$ is replaced by element $a_i$ ($1 \leq i \leq k$).
	
	
	Consider a model $\M$ with $\dm(\M) = n$. The structure $\astruct_\M$ representing 
	$\M$ over the vocabulary formed by $\pos$ and $\subseteq$ is defined as follows.
        The domain of $\astruct_\M$ is the set $\{0,1, \bot\}^n$ of all
		partial instances of dimension $n$.
        An instance $\es \in \{0,1\}^n$ is in the interpretation of
		 $\pos$ in $\astruct_\M$ if and only if $\M(\es) = 1$, and no partial instance including undefined features is contained in the interpretation of $\pos$. 
		 Moreover, a pair $(\es_1,\es_2)$ is in the interpretation of
		relation $\subseteq$ in $\astruct_\M$ if and only if $\es_1$ is
		subsumed by $\es_2$.
%
	Finally, given a formula $\varphi(x_1, \ldots, x_k)$ in $\foil$ and
	partial instances $\es_1$, $\ldots$, $\es_k$ of dimension $n$, model
	$\M$ is said to {\em satisfy} $\varphi(\es_1, \ldots, \es_k)$, denoted
	by $\M \models \varphi(\es_1, \ldots, \es_k)$, if 
	$\astruct_\M \models \varphi(\es_1, \ldots, \es_k)$.
	
	Notice that for a decision tree $\T$, the structure $\astruct_\T$ can be exponentially larger than $\T$. Hence, $\astruct_\T$ is a theoretical construction needed to formally define the semantics of $\foil$, but that should not be built when verifying in practice if a formula $\varphi$ is satisfied by $\T$. 
	
	\subsection{Expressing interpretability queries in \foil}\label{sec:expressing-in-foil}
	
	It will be instructive for the rest of our presentation, to see a few examples of how $\foil$ can be used to express some natural explainability queries on models.  
	In these examples 
	we make use of the following $\foil$ formula:  
	$$\full(x) \ = \ \forall y \, (x \subseteq y \, \rightarrow \, y \subseteq x).$$%
	Notice that if $\M$ is a model and $\es$ is a partial instance, then 
	$\M \models \full(\es)$ if and only if $\es$ is also an instance (i.e., it has no undefined features). 
	We also use the formula 
	$$\allpos(x) \ = \ \forall y \, \big((x \subseteq y \wedge \full(y)) \, \rightarrow \, \pos(y)\big),$$
	such that $\M \models \allpos(\es)$ if and only if every completion $\es'$ of $\es$ is a positive instance of $\M$. 
	Analogously, we define a formula $\allneg(x)$. 

A {\em sufficient reason} (SR) for an instance $\es$ over a model $\M$ is a partial instance $\es'$ such that 
	$\es' \subseteq \es$ and each completion of $\es'$ takes the same value over $\M$ as $\es$. 
	We can define SRs in $\foil$ as follows:  
\vspace{-0.2cm}
\begin{multline*} 
        \sr(x,y) = \full(x) \wedge y \subseteq x \ \land \\  (\pos(x) \to \allpos(y)) \land (\neg \pos(x) \to \allneg(y)).
\end{multline*}
	In fact, it is easy to see that $\M \models \sr(\es,\es')$ if and only if $\es'$ is a SR for $\es$ over $\M$. 
	Notice that $\es$ is always a SR for itself. However, we are typically interested in SRs that satisfy some optimality criterion. A common such criterion is that of being {\em minimal} \citep{shih2018symbolic,DBLP:journals/corr/abs-2010-11034,NEURIPS2020_b1adda14}. Formally, $\es'$ is a {\em minimal SR} for $\es$ over $\M$, if $\es'$ is a SR for $\es$ over $\M$ and there is no partial instance $\es''$ that is properly subsumed by $\es'$ that is also a SR for $\es$. Let us write $x \subset y$ for $x \subseteq y \wedge \neg (y \subseteq x)$. Then for 
	\[
		\minsr(x,y) =  \sr(x,y) \land
		\forall z \,(z \subset y \, \rightarrow \neg \sr(x,z)),
	\]%
	we have that $\M \models \minsr(\es,\es')$ if and only if $\es'$ is a minimal SR for $\es$ over $\M$. Minimal SRs
	have also been called {\em prime implicant} or {\em abductive} explanations in the literature \citep{DBLP:conf/aaai/IgnatievNM19,DBLP:journals/corr/abs-2211-00541}. 
	
	A usual global interpretability question about an ML model is to decide which features are sufficient/relevant for the
	prediction~\citep{Huang_Cooper_Morgado_Planes_Marques-Silva_2023, Darwiche_Hirth_2020}. In other words, which features determine the decisions made by the model. Such a notion can be defined in $\foil$ as follows:
\begin{equation*}	
\resizebox{\columnwidth}{!}{$
		\displaystyle
		\sem(x) = \forall y  \big(\suf(x,y) \rightarrow
		(\allpos(y) \vee \allneg(y))\big)
	$}
\end{equation*}%
        where $\suf(x,y)$ is a $\foil$ formula ($\suf$ stands
        for {\em same undefined features}) such that
        $\M \models \suf(\es,\es')$ if and only if $\es_\bot
        = \es'_\bot$, i.e., the sets of undefined features in $\es$
        and $\es'$ are the same (see the supplementary material for the definition of $\suf$). Then we have that $\M \models \sem(\es)$ if and only if for every $\es'$ with $\es_\bot = \es'_\bot$, all
        completions of $\es'$ receive the same classification over
        $\M$. That is, the output of the model on each instance
        is invariant to the features that are undefined in $\es$.  We
        call this a {\em Determinant Feature Set} (\sem).  
        
        
        As before, we can also express that $\es$
        is {\em minimal} with respect to feature determinacy using the
         formula:
	$$\minsem(x)  = \sem(x) \, \wedge \, \forall y \,(y \subset x \, \rightarrow \neg \sem(y)).$$

\section{Limitations of~\foil}\label{sec:foil-limitations}
As we show in this section, $\foil$ fails to meet either of the two criteria we are looking for in a practical language that provides explanations about decision trees. 
The first issue is its limited expressivity: there are important notions of explanations that cannot be expressed in this language. 
The second issue is its high computational complexity: There are queries in $\foil$ that cannot be evaluated with a polynomial number of calls
to an NP oracle. 

\subsection{Limited expressiveness}
\label{sec-lim-foil}
In some scenarios we want to express a stronger condition for SRs and \sem s: not only that they are minimal, but also that they are {\em minimum}. Formally, 
	a SR $\es'$ for $\es$ over $\M$ is {\em minimum}, if there is no SR $\es''$ for $\es$ over $\M$ with $|\es''_\bot| < |\es'_\bot|$, i.e., $\es''$ has more undefined features than 
	$\es'$. Analogously, we can define the notion of {\em minimum} \sem. 
	One can observe that an \sem\ is minimum if and only if it is minimal \citep{local-vs-global}. 
	Therefore, the $\foil$ formula $\minsem(x)$ presented earlier indicates that $\es$ is both the minimum and minimal \sem. This is however not the case for SRs; a sufficient reason can be minimal without being minimum. 
	The following theorem shows that $\foil$ cannot express the query that verifies if a partial instance $\es'$ is a minimum SR for a given instance $\es$ over decision trees. Due to space constrains, see supplementary material for all proofs.

	
	\begin{theorem}\label{thm:ne-foil}
		There is no formula \msr$(x,y)$ in \foil~such that, for every decision tree $\T$, instance $\es$ and partial instance $\es'$, we have that 
		$\T \models$ \msr$(\es,\es') \Leftrightarrow \text{$\es'$ is a minimum SR for $\es$ over $\T$.}$  
	\end{theorem}
	

\subsection{High complexity}
\label{sec-hc-foil}
For each query $\varphi(x_1, \ldots, x_k)$ in $\foil$, we define its associated problem \logicEvaluation$(\varphi)$ as follows. 
%
	%
	\begin{center}
		\fbox{\begin{tabular}{rl}
				{\sc Problem:} & \logicEvaluation$(\varphi)$\\
				{\sc Input:} & A decision tree $\T$ and partial instances \\ & $\es_1, \ldots, \es_k$ of dimension $n$\\
				{\sc Output:} & \textsc{Yes}, if $\T \models \varphi(\es_1, \ldots, \es_k)$, \\ 
				& and \textsc{No} otherwise
		\end{tabular}}
	\end{center}
	
	It is known that there exists a formula $\phi(x)$ in $\foil$ for which its evaluation problem over the class of decision trees is $\np$-hard~\citep{DBLP:conf/nips/ArenasBBPS21}. We want to determine whether the language $\foil$ is appropriate for implementation using SAT encodings.
	Thus, it is natural to ask whether the evaluation problem for formulas in this logic can always be decided in polynomial time by using a $\np$ oracle.
	However, we prove that this is not always the case. Although the evaluation of $\foil$ formulas is always in the polynomial hierarchy (PH), there exist formulas in $\foil$ for which their corresponding evaluation problems are hard for every level of PH. Based on widely held complexity assumptions, we can conclude that $\foil$ contains formulas whose evaluations cannot be decided in polynomial time by using a $\np$ oracle.
	
	
	\begin{theorem} 
		\label{thm:eval-folistar} 
		The following statements hold:
		\begin{enumerate} 
		\item 
		Let $\phi$ be a $\foil$ formula. Then there exists $k \geq 0$ such that 
			\logicEvaluation$(\phi)$ is in the $\Sigma_k^{\rm{P}}$ complexity class. 
			\item 
			For every $k \geq 0$, there is an \foil-formula $\phi_k$ such that \logicEvaluation$(\phi_k)$ 
			is $\Sigma_k^{\rm{P}}$-hard. 
			\end{enumerate} 
	\end{theorem}

\section{A Better Logic to Explain Decision Trees}\label{sec:dt-foil}

According to the previous section, we have two limitations regarding
$\foil$ that we need to address in order to build a practical language
to provide explanations about decision trees. This imposes two needs
on us: on one hand, we must extend $\foil$ to increase its expressive
power, and on the other hand, we must constrain the resulting language
to ensure that its evaluation complexity is appropriate. In this
section, we define the language \dtfoil\ that takes into consideration
both criteria and is specifically tailored for decision trees. We show
that \dtfoil\ is a natural language to express explainability notions
for decision trees, which, however, lacks a general mechanism to
express minimality conditions. Based on these
observations, we present in the following section two extensions
of \dtfoil\ that are capable of expressing rich notions of explanation
over decision trees, and for which the evaluation problem can be
solved with a polynomial number of calls to an NP oracle.



\subsection{The definition of \dtfoil}
\label{sec-dtfoil-c-v-1}
$\foil$ cannot express properties such as minimum sufficient reason
that involve comparing cardinalities of sets of features. As a first
step, we solve this issue extending the vocabulary of $\foil$ with a
simple binary relation $\lel$ defined as:
$$\M \models \es \lel \es' \ \ \Longleftrightarrow \ \
|\es_\bot| \geq |\es'_\bot|.$$ As we will show later, the use of this
predicate indeed allows us to express many notions of
explanations. However, the inclusion of this predicate in $\foil$ can
only add extra complexity. Therefore, our second step is to define the
logic $\dtfoil$, which is tailored for decision trees and can
efficiently make use of this new extra power.

\paragraph{Atomic formulas.}
Predicates $\subseteq$ and $\lel$, as well as predicates $\full$ and
$\suf$ used in Section \ref{sec:background}, can be called {\em
syntactic} in the sense that they refer to the values of the features
of partial instances, and they do not make reference to classification
models. It turns out that all the syntactic predicates needed in our
logical formalism can be expressed as first-order queries over the
predicates $\subseteq$ and $\lel$. Moreover, such formulas can be
evaluated in polynomial time:

\begin{theorem}\label{theo:ptime-atomic}
  Let $\phi$ be a first-order formula defined over 
  the vocabulary
  $\{\subseteq, \lel\}$. Then \logicEvaluation$(\phi)$ is in $\ptime$.
\end{theorem}%
The  {\em atomic formulas} of $\dtfoil$ are defined
as the set of $\foil$ formulas  over the vocabulary
$\{\subseteq, \lel\}$. Note that we could not have simply taken
one of these predicates when defining atomic formulas, as we show in the supplementary material that they cannot
be defined in terms of each other. The following is an
example of a new atomic formula in $\dtfoil$:	
\(
\cons(x,y) \ = \  \exists z \, (x \subseteq z \wedge y \subseteq z).
\)
This relation checks whether two partial instances $\es$ and $\es'$ are {\em consistent}, in the sense 
that features that are defined in both $\es$ and $\es'$ have the same value. We use this formula in the rest of this work.

	\vspace{-0.75em}
	\paragraph{The logic $\dtfoil$.}
At this point, we depart from the model-agnostic approach of
$\foil$ and introduce the concept of {\em guarded} quantification,
which specifically applies to decision trees. This involves
quantifying over the elements that define a decision tree, namely its nodes and leaves. 

Given a decision tree $\T$ and a node $u$ of $\T$,
the instance $\es_u$ {\em represented} by $u$ is defined as follows.
If $\pi = u_1 \cdots u_k$ is the unique path that leads from the
root of $\T$ to $u_k = u$, then: (i) for
every $i \in \{1, \ldots, k-1\}$, if the label of node $u_i$ is $j \in
\{1,\dots,n\}$, then $\es_u[j]$ is equal to the label of the edge in
$\T$ from $u_i$ to $u_{i+1}$; and (ii) for each $j \in \{1, \ldots,
n\}$, $\es_u[j] = \bot$ if the label of $u_i$ is different from $j$ for
every $i \in \{1,\dots,k-1\}$.
For example, for the decision tree $\T$ in Figure~\ref{fig:tree-example} and the nodes $u$, $v$ shown in this figure, it holds that
$\es_u = (\bot, 0, \bot, \bot)$ and $\es_v = (0, 1, 0, 0)$.
Then we define
predicates $\node(x)$ and $\posl(x)$ as follows, given a decision tree $\T$ and a partial instance $\es$: 
(i) $\T \models
\node(\es)$ if and only if $\es = \es_u$ for some node $u \in \T$;
(ii) $\T \models \posl(\es)$ if and
only if $\es = \es_u$ for some leaf $u$ of $\T$ with label $\true$.
Then considering the vocabulary $\{\subseteq, \lel, \node, \posl\}$, the logic $\dtfoil$ is recursively defined as follows: (i)~Atomic formulas are $\dtfoil$ formulas. (ii)~$\dtfoil$ formulas are closed under Boolean combinations. (iii)~If $\phi$ is a $\dtfoil$ formula, then
$\exists x (\node(x) \wedge \phi)$, $\forall x
(\node(x) \to \phi)$, $\exists x (\posl(x) \wedge
\phi)$ and $\forall x (\posl(x) \to \phi)$ are $\dtfoil$ formulas.
      
The logic $\dtfoil$ is termed {\em guarded} due to the fact that every quantification is protected by a collection of nodes or leaves in the decision tree. As the decision tree comprises a linear number of nodes (in the size of the tree), and hence a linear number of leaves, it follows from Theorem \ref{theo:ptime-atomic} that every guarded formula can be evaluated within polynomial time. 

\begin{proposition}\label{prop:ptime-guarded}
  Let $\phi$ be a~\dtfoil~formula. Then \logicEvaluation$(\phi)$ is in $\ptime$.
\end{proposition}%

\subsection{On the expressiveness of $\dtfoil$}
The logic $\dtfoil$ allows to express in a simple way the basic
notions of explanation that we study in this paper. In
particular, the basic predicates that are needed to express such
notions can be easily expressed using guarded quantification,
considering the predicate $\cons$ defined in Section \ref{sec-dtfoil-c-v-1}:

\begin{equation*}
	\resizebox{1.0\columnwidth}{!}{$
    \begin{split}
	\leaf&(x) \ =  \ \node(x) \land
         \forall y \, \big(\node(y) \to (x \subseteq y \to y \subseteq x)\big)\\
	\allpos&(x) \ = \ \forall y \, \big(\node(y) \to 
	\\ & ((\leaf(y) \wedge \cons(x,y)) \to \posl(y))\big)\\
	\allneg&(x) \ = \ \forall y \, \big(\node(y) \to
	\\ & ((\leaf(y) \wedge \cons(x,y)) \to \negl(y))\big)\\
	\pos&(x) \ = \ \full(x) \wedge \allpos(x)
	\end{split}$
	}
\end{equation*}

With these definitions, both predicates $\sr$ and $\sem$ can be
expressed as \dtfoil~formulas. In fact, with the new definitions of
the predicates $\pos$, $\allpos$ and $\allneg$, we have that the
formula~in~\Cref{sec:expressing-in-foil} defining $\sr$ is a $\dtfoil$
formula.
For $\sem$ the situation is a bit more complex. Observe first that we cannot
use the definition of $\sem(x)$ provided in~\Cref{sec:expressing-in-foil}, as such a formula
involves an
unrestricted quantifier. Instead, we can use the
following $\dtfoil$ formula $\sem(x)$:
\begin{align}
        \label{eq:sem-2}
 &         \forall y \, \big[\node(y) \to (\allpos(y) \ \to\ \\
\notag &        \hspace{10pt}        \forall z \, (\node(z) \to (\allneg(z) \ \to \\
\notag & \hspace{20pt}        \neg \exists w \, (\suf(x,w) \wedge \cons(w,y) \wedge
                \cons(w,z)))))\big].
\end{align}
Notice that this is a \dtfoil~formula since the formula $\neg \exists
w \, (\suf(x,w) \wedge \cons(w,y) \wedge \cons(w,z))$ is atomic (that
is, it is defined by using only the predicates $\subseteq$ and
$\lel$).

The logic $\dtfoil$ does not have an explicit mechanism to represent
minimal notions of explanation. This can be proved
by showing that the notion of minimum sufficient reason cannot be
expressed in the logic, which is a simple corollary of
Proposition \ref{prop:ptime-guarded} and the fact that the problem of
verifying, given an instance $\es$, a partial instance $\es'$, and a
decision tree $\T$, whether $\es'$ is a minimum sufficient reason for
$\es$ over $\T$ is \conp-complete \citep{NEURIPS2020_b1adda14}.
\begin{corollary}\label{prop:non-dtfoil}
Assuming $\ptime \neq \np$, there is no formula \msr$(x,y)$
		in \dtfoil\ such that, for every decision tree $\T$,
		instance $\es$ and partial instance $\es'$, it holds that
		$\T \models$ \msr$(\es,\es') \Leftrightarrow \text{$\es'$
		is a minimum \sr~for $\es$ over $\T$.}$
\end{corollary}
This result motivates two extensions of $\dtfoil$ that will finally
meet our desiderata for an interpretability logic, which are
presented in the following section.


\section{Making \dtfoil~Practical}\label{sec:opt-dt-foil} 

As shown in the previous section, $\dtfoil$ can be evaluated in
polynomial time and can express some natural explainability
properties, but
lacks the ability to express \emph{optimality} properties. To remedy this, in
this section we propose two extension of $\dtfoil$. We start by
proposing $\qdtfoil$, a logic that is defined by allowing
quantification without alternation over $\dtfoil$. As we will show in
this section, $\qdtfoil$ meets all the criteria stated in the
introduction, except for the fact that the computation of an answer
for a $\dtfoil$ formula cannot be done with a polynomial numbers
of calls to an $\np$ oracle.
Based on the findings, we then propose \optdtfoil, a logic that is
defined by introducing a minimality operator over $\dtfoil$. As we
will shown in this section, $\optdtfoil$ meets all the criteria for an
appropriate interpretability logic.

\subsection{The logic \qdtfoil}
\label{sec-qdtfoil}

The logic \qdtfoil~is recursively defined as follows:
(i) each formula in $\dtfoil$ is a $\qdtfoil$ formula;
(ii) Boolean combinations of \qdtfoil~formulas are \qdtfoil~formulas; and
(iii) if $\phi$ is a $\dtfoil$ formula, then $\exists x_1 \cdots \exists x_\ell \, \phi$ and
$\forall x_1 \cdots \forall x_\ell \, \phi$ are \qdtfoil~formulas.
Both predicates $\sr$ and $\sem$ can be expressed as \qdtfoil~formulas, as we shown in the previous section that they can be
expressed as $\dtfoil$ formulas.
More importantly, the form of quantification allowed in $\qdtfoil$ is enough
to express optimality properties. As a first example of this, consider the following simple definition of 
the explainability queries studied in the paper, where $z \lnel y$ is a shorthand for $z \lel y \wedge \neg (y \lel z)$:
\begin{equation*}
	\resizebox{\columnwidth}{!}{$ 
\begin{split}
	\textsc{MinimalSR}(x,y)&=\sr(x,y) \land \forall z \big(z \subset y \to \neg \sr(x,z)\big)\\
	\textsc{MinimumSR}(x,y)&=\sr(x, y) \land \ \forall z \big(z \lnel y \to \neg \sr(x,z)\big)\\
	\textsc{MinimalDFS}(x)&=\sem(x) \land \ \forall y \big(y \subset x \to \neg \sem(y)\big)
\end{split}$}
\end{equation*}
As a second example of the expressiveness of $\qdtfoil$, consider the
notion of minimum change required (MCR) mentioned in the
introduction. Given an instance $\es$ and a decision tree $\T$, MCR aims to
find another instance $\es'$ such that $\T(\es) \neq \T(\es')$ and the
number of features whose values need to be flipped in order to change
the output of the decision tree is minimal, which is the same as saying that
the Hamming distance between $\es$ and $\es'$ is minimal.  It is
possible to express MCR in $\qdtfoil$ as follows. In the supplementary material, we
show that there exists a ternary atomic $\dtfoil$ formula $\led$ such
that for every decision tree $\T$ of dimension $n$ and every sequence of
instances $\es_1, \es_2, \es_3$ of dimension $n$, it holds that:
$\T \models \led(\es_1, \es_2, \es_3)$ if and only if the Hamming
distance between $\es_1$ and $\es_2$ is less or equal than the Hamming
distance between $\es_1$ and $\es_3$. By using $\led$, we can express
in $\qdtfoil$ the notion of minimum change required:
\begin{equation*}
	\resizebox{0.99\columnwidth}{!}{$
\begin{split}
&\mcr(x,y) = 
\full(x) \wedge \full(y) \ \wedge\\
&\hspace{10pt}\neg (\pos(x) \leftrightarrow \pos(y)) \ \wedge \\
&\hspace{10pt}\forall z \big[\big(\full(z) \wedge \neg
(\pos(x) \leftrightarrow \pos(z))\big) \to \led(x,y,z)\big].
\end{split}
$}
\end{equation*}
%
The next necessary step in the study of $\qdtfoil$ is to establish the
complexity of deciding whether a tuple of partial instances is an
answer to a $\qdtfoil$ formula, and the complexity of computing such
answers. To this aim, we first consider the evaluation problem
\logicEvaluation($\phi$) for a fixed \qdtfoil~formula $\phi(x_1,\dots,x_m)$, which is defined exactly as in Section \ref{sec-hc-foil} for the case of $\foil$.
%
%
%
%
Next we provide a precise characterization of the
complexity of the evaluation problem for \qdtfoil. More specifically,
we establish that this problem can always be solved in the {\em
Boolean Hierarchy over $\np$}
\citep{DBLP:conf/fct/Wechsung85,DBLP:journals/siamcomp/CaiGHHSWW88}, i.e., as a
Boolean combination of $\np$ problems.  In this theorem, a level of
the Boolean hierarchy is denoted as $\bh_k$ (the definition of this
hierarchy can be found in the supplementary material).
\begin{theorem}
	\label{thm:eval-qdtfoil}
	(i) For each \qdtfoil~formula $\phi$, there is $k \geq 1$ such
	that \logicEvaluation$(\phi)$ is in $\bh_k$;
	(ii)  For every $k \geq 1$, there is a \qdtfoil~formula
	$\phi_k$ such that \logicEvaluation$(\phi_k)$ is $\bh_k$-hard.   
\end{theorem}
This result tells us that $\qdtfoil$ meets one of the fundamental
criteria for an interpretability logic, namely that it can be verified
whether a tuple is an answer to a $\qdtfoil$ formula with a polynomial
number of calls to an $\np$ oracle. Hence, the next step in the study
of $\qdtfoil$ is to establish the complexity of computing such
answers.
%
For a fixed formula $\phi(x_1,\dots,x_m)$ in \qdtfoil, we define its
corresponding computation problem $\comp(\phi)$ as the problem of
computing, given decision tree $\T$ of dimension $n$, a sequence
partial instances $\es_1$, $\ldots$, $\es_m$ of dimension $n$ such
that $\T \models \phi(\es_1, \ldots, \es_m)$ (and answering no if such
a sequence does not exist). Unfortunately, the following result tells
us that this problem cannot be solved with polynomial number of calls
to an
$\np$ oracle, showing a limitation of $\qdtfoil$ when
computing answers.
\begin{theorem}
	\label{thm:comp-qdtfoil}
	There exists a \qdtfoil~formula $\phi$ such that if $\comp(\phi)$ can be solved
        in \fpnp,
        then the polynomial hierarchy collapses to its second level, $\stp$.
\end{theorem}

\subsection{The logic \optdtfoil}
Given what we have learned in the previous sections, our aim is to
construct the right extension of $\dtfoil$ that meets all the criteria for an interpretability
logic. This logic is $\optdtfoil$, which is studied in this section by
defining its components, studying its expressiveness, and finally
showing that the problem of computing an answer to an $\optdtfoil$
formula can be solved with a polynomial number of calls to an $\np$
oracle.
\vspace{-0.6em}
\paragraph{The definition of the logic.}
An atomic $\dtfoil$ formula $\rho(x, y, v_1, \ldots,
v_\ell)$ represents a strict partial order if for every dimension $n$
and assignment of partial instances of dimension $n$ for the variables
$v_1$, $\ldots$, $v_\ell$, the resulting binary relation over the
variables $x$ and $y$ is a strict partial order over the partial
instances of dimension $n$. Formally, $\rho(x, y, v_1, \ldots,
v_\ell)$ {\em represents a strict partial order} if for every decision
tree $\T$:
\begin{multline*}
	\T \models \ \forall v_1 \cdots \forall v_\ell \, \big[ \forall x \, \neg \rho(x, x, v_1, \ldots, v_\ell) \ \wedge\\
	\forall x \forall y \forall z \,\big(
	(\rho(x, y, v_1, \ldots, v_\ell) \wedge \rho(y, z, v_1, \ldots, v_\ell))\\
	\to \rho(x, z, v_1, \ldots, v_\ell)\big)\big]. 
\end{multline*}
Notice that variables $v_1$, $\ldots$, $v_\ell$ in the formula
$\rho(x, y, v_1, \ldots, v_\ell)$ are considered as parameters that define a
strict partial order. In fact, different assignments for these variables can
give rise to different orders. Hence, we use notation $\rho[v_1, \ldots,
v_\ell](x, y)$ to make explicit the distinction between the parameters $v_1$,
$\ldots$, $v_\ell$ that define the order and the variables $x$, $y$ that are
instantiated with partial instances to be compared. For example, the strict
partial order defined from the subsumption relation is defined by the formula
$\rho_1(x, y) = x \subset y$.
As a second example, consider the case where a certain feature must be
disregarded when defining an order for partial instances.  For
instance, in many cases, it is not desirable to use the feature {\em
gender} when comparing partial instances. Such an order can be defined
as follows. With the appropriate values for variables $v_1$ and $v_2$,
the following formula checks whether feature $x$ has value $\bot$ in
the $i$-th feature:
$\nf(x, v_1, v_2) = \neg(v_1 \subseteq x) \wedge \neg(v_2 \subseteq x)$.
For instance, if we are considering partial instances of dimension $5$ and we
need to check whether feature $x$ has value $\bot$ in the first feature, then
we can use the values $c_1 = (0, \bot, \bot, \bot, \bot)$ and
$c_2 = (1, \bot, \bot, \bot, \bot)$ for the variables $v_1$ and $v_2$,
respectively. Moreover, let
$\pred(x, y) =  x \subset y \wedge \neg \exists z \, (x \subset
z \wedge z \subset y)$
be a formula that check whether $x$ is a predecessor of $y$ under the
order $\subset$. Then, with the appropriate values for the parameters
$v_1$ and $v_2$, the following formula defines a strict partial order based
on $\subset$ but that disregards the $i$-th feature when comparing partial instances:
	\begin{equation}
		\resizebox{0.9\columnwidth}{!}{$ 
\begin{split}
 \rho_2[v_1,v_2](x,y) =
\exists x' \exists y' \,\big[(\nf(x, v_1, v_2) \to x = x') \ \wedge \\
\hspace{15pt}(\neg \nf(x, v_1, v_2) \to (\pred(x', x) \wedge \nf(x', v_1, v_2))) \ \wedge \\
	\label{eq-order-i}	 \hspace{15pt}(\neg \nf(y, v_1, v_2) \to (\pred(y', y) \wedge \nf(y', v_1, v_2))) \ \wedge \\
\hspace{15pt}(\nf(y, v_1, v_2) \to y = y') \ \wedge x' \subset y'\big]
\end{split}$}
\end{equation}
For instance, $\rho_2[c_1, c_2](x, y)$, for the constants $c_1$ and $c_2$
mentioned above, defines a strict partial order that disregards the first
feature when comparing partial instances of dimension~5.

Atomic $\dtfoil$ formulas representing strict partial orders will be
used in the definition of $\optdtfoil$. Hence, it is necessary to have
an algorithm that verifies whether this condition is satisfied in
order to have a decidable syntax for $\optdtfoil$. In what follows, we
prove that such an algorithm exists.
\begin{proposition}\label{prop-strict-po}
The problem of verifying, given an atomic \dtfoil-formula $\rho[v_1, \ldots,
v_\ell](x, y)$, whether $\rho[v_1, \ldots, v_\ell](x, y)$ represents a strict
partial order can be solved in double exponential time.
\end{proposition}
Although the algorithm in Proposition \ref{prop-strict-po} has high
complexity, we are convinced that it can be used in practice, as we
expect formulas representing strict partial orders to be small and to
have a simple structure.

We need one additional piece of notation to define the syntax of $\optdtfoil$.
Given a \dtfoil-formula $\varphi(x, u_1, \ldots, u_k)$, we
use notation $\varphi[u_1, \ldots, u_k](x)$ to indicate that $x$ is a
distinguished variable and $u_1, \ldots, u_k$ are parameters that
define the possible values for $x$. In general, we use this syntax
when $x$ stores an explanation given an assignment for the
variables $u_1$, $\ldots$, $u_\ell$. For example, we use notation
$\varphi[u](x) = \sr(u, x)$ to indicate that $x$ is a sufficient
reason given an assignment for the variable $u$ (that is, $x$ is a
sufficient reason for $u$). We use this terminology to be consistent
with the notation used for the formulas that represent strict partial
orders, so that we have a consistent notation when defining $\optdtfoil$.
Formally, given a \dtfoil-formula $\varphi[u_1, \ldots, u_k](x)$ and an atomic
\dtfoil-formula $\rho[v_1, \ldots, v_\ell](y, z)$ that represents a
strict partial order, an \optdtfoil-formula is an expression
of the following form:
\begin{align*}
	\Psi[u_1, \ldots, u_k,
		& v_1, \ldots, v_\ell](x) = \\
		& \minf[\varphi[u_1, \ldots, u_k](x), \rho[v_1, \ldots, v_\ell](y, z)].
\end{align*}
Notice that $x$, $u_1, \ldots, u_k$, $v_1, \ldots, v_\ell$ are the free
variables of this formula, while the variables $y$, $z$ are quantified
variables in it. In particular, $x$ is a variable used to store a
minimal explanation, $u_1, \ldots, u_k$ are the parameters that define the
notion of explanation, and $v_1, \ldots, v_\ell$ are the parameters that
define the strict partial order over which we are minimizing.
The semantics of $\Psi[u_1, \ldots, u_k, v_1, \ldots, v_\ell](x)$ is defined by
considering the following \qdtfoil-formula:
\begin{multline*}
		\theta_{\minf}(x, 
			 u_1, \ldots, u_k, v_1, \ldots, v_\ell) = 
			\varphi(x, u_1, \ldots, u_k) \ \wedge \\
			\forall y \, \big(\varphi(y, u_1, \ldots, u_k) \to
			\neg \rho(y, x, v_1, \ldots, v_\ell)\big)
\end{multline*}
More precisely, given a decision tree $\T$ of dimension $n$ and partial
instance $\es$, $\es'_1, \ldots, \es'_k$, $\es''_1, \ldots, \es''_\ell$ of
dimension $n$, we have that $\T \models \Psi[\es'_1, \ldots, \es'_k, \es''_1,
\ldots, \es''_\ell](\es)$ if and only if $\T \models \theta_{\minf}(\es,
\es'_1, \ldots, \es'_k, \es''_1, \ldots, \es''_\ell)$.

\paragraph{On the expressiveness of $\optdtfoil$.}
As is customary, a logic $\mathcal{L}_1$ is {\em contained} in a
logic $\mathcal{L}_2$ if for every formula in $\mathcal{L}_1$, there
exists an equivalent formula in $\mathcal{L}_2$. Moreover,
$\mathcal{L}_1$ is {\em properly contained} in $\mathcal{L}_2$ if
$\mathcal{L}_1$ is contained in $\mathcal{L}_2$ if
$\mathcal{L}_2$ is not contained in $\mathcal{L}_1$. The following proposition shows that the
expressive power of $\optdtfoil$ lies between that of $\dtfoil$ and
$\qdtfoil$.
\begin{proposition}\label{prop-ep-optdtfoil}
Assuming that the polynomial hierarchy does not collapse, 
$\dtfoil$ is strictly contained in $\optdtfoil$, and $\optdtfoil$
is strictly contained in $\qdtfoil$.
\end{proposition}
The logic $\optdtfoil$ allows to express in a simple way all notions
of explanation that we study in this paper.  For example, assuming
that $\varphi[u](x) = \sr(u,x)$, the following minimal
\optdtfoil-formulas encode the notions of minimal and minimum sufficient reason:
\begin{eqnarray*}
\minsr[u](x) &=& \minf[\varphi[u](x), y \subset z],\\
\msr[u](x) &=& \minf[\varphi[u](x), y \lel z \wedge \neg (z \lel y)],
\end{eqnarray*}
while the minimal \dtfoil-formula $\minf[\varphi[u](x), \rho_2(y, z)]$
encodes the notion of minimal sufficient reason for the order
$\rho_2(y, z)$ defined in~\cref{eq-order-i} that disregards a
feature. As a second example, consider the notion of minimum change
required and the predicate $\led$ defined in
Section \ref{sec-qdtfoil}.  Then letting $\varphi[u](x)
= \full(u) \wedge \full(x) \wedge \neg
(\pos(u) \leftrightarrow \pos(x))$ and $\rho_3[u](y,z)
= \led(u,y,z) \wedge \neg \led(u,z,y)$, we can express the notion of
minimum change required in $\optdtfoil$ as follows:
\begin{align*}
	\mcr[u](x) = \minf[\varphi[u](x), \rho_3[u](y,z)].
\end{align*}
The logic $\optdtfoil$ can also be used to express notions of
explanation that involve maximality conditions, just by reversing the
order being considered. For example, consider the explainability query 
\textit{maximum change allowed}~\citep{alfano2024evenif} that ask for
the maximum number of changes that can be made to an instance without
changing the output of the classification model.  Considering
$\varphi[u](x) = \full(u) \wedge \full(x) \wedge (\pos(u)
\leftrightarrow \pos(x))$, and defining the reverse order $\rho_4[u](y,z)
= \rho_3[u](z,y)$, we can express the notion of
maximum change allowed in $\optdtfoil$ as follows:
\begin{align*}
	\mca[u](x) = \minf[\varphi[u](x), \rho_4[u](y,z)].
\end{align*}
An important feature of $\optdtfoil$ is that it allows for the combination of notions of explanation.
For example, given
two instances $u_1$ and $u_2$, let $\csr[u_1, u_2](x) = \sr(u_1,x) \wedge
\sr(u_2,x)$, so that this formula checks whether $x$ is a common sufficient
reason for the instances $u_1$ and $u_2$. Then it is possible to prove that the
following \optdtfoil-formula computes a common minimal sufficient reason
for two instances (if such a minimal sufficient reason exists): 
\begin{align*}
\Psi_1[u_1,u_2](x) = \minf[\csr[u_1, u_2](x), y \subset z].
\end{align*}
Finally, another important
feature of $\optdtfoil$ is that it allows for the exploration of the
space of explanations for a given classification. For example, assume
that we already have a minimal sufficient reason $x_1$ for an instance
$u$, which can be computed using the \optdtfoil-formula
$\minsr[u](x)$. Then we can compute a second minimal sufficient reason
$x_2$ for $u$ as follows. Let
\begin{align*}
\nsr[u,x_1](x) = \sr(u, x) \wedge \sr(u, x_1) \wedge \neg (x_1 \subseteq x),
\end{align*}
so this formula checks whether $x$ is a sufficient reason for $u$ which does
not subsume sufficient reason $x_1$. Then a minimal sufficient reason for the
instance $u$ that is different from the minimal sufficient reason $x_1$ can be
computed using the following \optdtfoil~formula:
\begin{align*}
\Psi_1[u,x_1](x) = \minf[\nsr[u, x_1](x), y \subset z].
\end{align*}
We can apply the same idea to other notions
of explanation, such as the $\mcr$ explainability query, in order to
compute multiple explanations for the output of a classification
model.



\vspace{-0.75em}

\paragraph{The computation problem.}
The computation problem for $\optdtfoil$ has to be defined considering
the different roles of the variables in an $\optdtfoil$ formula
$\Psi[u_1, \ldots, u_k, v_1, \ldots, v_\ell](x)$. In particular, the
parameters $u_1, \ldots, u_k, v_1, \ldots, v_\ell$ should be given as
input, while the value of $x$ is the explanation to be computed. The
following definition takes these considerations into account.%
\vspace{-0.35cm}
\begin{center}
\fbox{
                \begin{tabular}{rp{5.58cm}}
                        {\sc Problem:} & \logicComputation$(\Psi)$\\
                        {\sc Input:} & A decision tree $\T$ of dimension $n$ and partial
                        instances  $\es'_1$, $\ldots$, $\es'_k$, $\es''_1$, $\ldots$,
                        $\es''_\ell$ of dimension $n$\\
                        {\sc Output:} & Partial instance
                        $\es$ of dimension $n$ such that $\T \models \Psi[\es'_1, \ldots,                                                                                               
                        \es'_k, \es''_1, \ldots, \es''_\ell](\es)$, and \textsc{No} if no
                        such a partial instance exists
                \end{tabular}
        }
\end{center}
We show that $\optdtfoil$ meets our desiderata by proving that the computation problem for $\optdtfoil$
can be solved with a polynomial number of calls to an NP oracle,
\begin{theorem}\label{theo-comp-optdtfoil}
For every formula~$\Psi$ in~\optdtfoil, the problem~\logicComputation$(\Psi)$
is in \fpnp.
\end{theorem}

\section{Implementation and Challenges}\label{sec:implementation}
\begin{figure*}
	\centering
	\begin{subfigure}{0.19\textwidth}
		\begin{tikzpicture}[scale=0.4]
    \begin{axis}[
    every axis plot/.append style={line width=0.975pt},
    width=2.5\textwidth,
        xlabel={Number of nodes},
        ylabel={Runtime [s]},
        title={Minimal SR},
        legend pos={north west},
        ymajorgrids={true},
        grid style={dashed},
    ]
    \addplot[
        color={red},
        mark={*},
    ]
    coordinates {
        (501, 0.09)
        (1001, 0.158)
        (1501, 0.254)
        (2001, 0.34)
        (2501, 0.405)
        (3001, 0.503)
        (3501, 0.606)
        (4001, 0.676)
    };
    \addplot[
        color={blue},
        mark={+},
    ]
    coordinates {
        (501, 0.157)
        (1001, 0.267)
        (1501, 0.397)
        (2001, 0.541)
        (2501, 0.668)
        (3001, 0.865)
        (3501, 0.927)
        (4001, 1.114)
    };
    \addplot[
        color={green!60!black},
        mark={x},
    ]
    coordinates {
        (501, 0.202)
        (1001, 0.382)
        (1501, 0.565)
        (2001, 0.773)
        (2501, 0.956)
        (3001, 1.149)
        (3501, 1.34)
        (4001, 1.429)
    };
    \addplot[
        color={cyan},
        mark={square*},
    ]
    coordinates {
        (501, 0.242)
        (1001, 0.454)
        (1501, 0.69)
        (2001, 0.932)
        (2501, 1.176)
        (3001, 1.37)
        (3501, 1.558)
        (4001, 1.824)
    };
    \addplot[
        color={orange},
        mark={diamond*},
    ]
    coordinates {
        (501, 0.291)
        (1001, 0.57)
        (1501, 0.793)
        (2001, 1.085)
        (2501, 1.43)
        (3001, 1.69)
        (3501, 2.12)
        (4001, 2.179)
    };
    \addplot[
        color={purple},
        mark={triangle*},
    ]
    coordinates {
        (501, 0.321)
        (1001, 0.651)
        (1501, 0.977)
        (2001, 1.291)
        (2501, 1.68)
        (3001, 1.984)
        (3501, 2.256)
        (4001, 2.656)
    };
    \legend{$d = 50$,$d = 100$,$d = 150$,$d = 200$,$d = 250$,$d = 300$}
    \end{axis}
    \end{tikzpicture}
    
	\end{subfigure}
	\begin{subfigure}{0.19\textwidth}
		\begin{tikzpicture}[scale=0.4]
    \begin{axis}[
    every axis plot/.append style={line width=0.975pt},
    width=2.5\textwidth,
        xlabel={Number of nodes},
        ylabel={Runtime [s]},
        legend pos={north west},
        ymajorgrids={true},
        grid style={dashed},
        title={Maximum CA},
        ytick={0,1,2,3,4},
    ]
    \addplot[
        color={red},
        mark={*},
    ]
    coordinates {
        (501, 0.088)
        (1001, 0.083)
        (1501, 0.104)
        (2001, 0.107)
        (2501, 0.13)
        (3001, 0.146)
        (3501, 0.153)
        (4001, 0.167)
    };
    \addplot[
        color={blue},
        mark={+},
    ]
    coordinates {
        (501, 0.348)
        (1001, 0.265)
        (1501, 0.261)
        (2001, 0.303)
        (2501, 0.297)
        (3001, 0.354)
        (3501, 0.337)
        (4001, 0.317)
    };
    \addplot[
        color={green!60!black},
        mark={x},
    ]
    coordinates {
        (501, 0.643)
        (1001, 0.613)
        (1501, 0.565)
        (2001, 0.552)
        (2501, 0.618)
        (3001, 0.568)
        (3501, 0.571)
        (4001, 0.653)
    };
    \addplot[
        color={cyan},
        mark={square*},
    ]
    coordinates {
        (501, 1.046)
        (1001, 0.92)
        (1501, 1.008)
        (2001, 0.968)
        (2501, 0.905)
        (3001, 0.913)
        (3501, 1.001)
        (4001, 0.956)
    };
    \addplot[
        color={orange},
        mark={diamond*},
    ]
    coordinates {
        (501, 1.404)
        (1001, 1.391)
        (1501, 1.513)
        (2001, 1.511)
        (2501, 1.513)
        (3001, 1.705)
        (3501, 1.49)
        (4001, 1.405)
    };
    \addplot[
        color={purple},
        mark={triangle*},
    ]
    coordinates {
        (501, 2.572)
        (1001, 2.102)
        (1501, 1.91)
        (2001, 2.013)
        (2501, 1.923)
        (3001, 2.128)
        (3501, 2.313)
        (4001, 1.908)
    };
    \legend{$d = 50$,$d = 100$,$d = 150$,$d = 200$,$d = 250$,$d = 300$}
    \end{axis}
    \end{tikzpicture}
    
	\end{subfigure}
	\begin{subfigure}{0.19\textwidth}
		\begin{tikzpicture}[scale=0.4]
    \begin{axis}[
    every axis plot/.append style={line width=0.975pt},
    width=2.5\textwidth,
        xlabel={Number of nodes},
        ylabel={Runtime [s]},
        legend pos={north west},
        ymajorgrids={true},
        grid style={dashed},
        title={Minimum DFS},
    ]
    \addplot[
        color={red},
        mark={*},
    ]
    coordinates {
        (501, 0.075)
        (1001, 0.229)
        (1501, 0.481)
        (2001, 0.847)
        (2501, 1.314)
        (3001, 1.842)
        (3501, 2.514)
        (4001, 3.236)
    };
    \addplot[
        color={blue},
        mark={+},
    ]
    coordinates {
        (501, 0.15)
        (1001, 0.304)
        (1501, 0.564)
        (2001, 0.942)
        (2501, 1.41)
        (3001, 1.98)
        (3501, 2.664)
        (4001, 3.436)
    };
    \addplot[
        color={green!60!black},
        mark={x},
    ]
    coordinates {
        (501, 0.329)
        (1001, 0.419)
        (1501, 0.675)
        (2001, 1.043)
        (2501, 1.524)
        (3001, 2.077)
        (3501, 2.755)
        (4001, 3.515)
    };
    \addplot[
        color={cyan},
        mark={square*},
    ]
    coordinates {
        (501, 0.562)
        (1001, 0.625)
        (1501, 0.836)
        (2001, 1.203)
        (2501, 1.695)
        (3001, 2.258)
        (3501, 2.942)
        (4001, 3.744)
    };
    \addplot[
        color={orange},
        mark={diamond*},
    ]
    coordinates {
        (501, 0.828)
        (1001, 0.967)
        (1501, 1.088)
        (2001, 1.422)
        (2501, 1.901)
        (3001, 2.483)
        (3501, 3.095)
        (4001, 3.803)
    };
    \addplot[
        color={purple},
        mark={triangle*},
    ]
    coordinates {
        (501, 1.492)
        (1001, 1.304)
        (1501, 1.427)
        (2001, 1.643)
        (2501, 2.089)
        (3001, 2.638)
        (3501, 3.295)
        (4001, 4.007)
    };
    \legend{$d = 50$,$d = 100$,$d = 150$,$d = 200$,$d = 250$,$d = 300$}
    \end{axis}
    \end{tikzpicture}
    
	\end{subfigure}
	\begin{subfigure}{0.19\textwidth}
		\begin{tikzpicture}[scale=0.4]
    \begin{axis}[
    every axis plot/.append style={line width=0.975pt},
    width=2.5\textwidth,
        xlabel={Number of nodes},
        ylabel={Runtime [s]},
        title={Minimum SR},
        legend pos={north west},
        ymajorgrids={true},
        grid style={dashed},
    ]
    \addplot[
        color={red},
        mark={*},
    ]
    coordinates {
        (501, 0.188)
        (1001, 0.265)
        (1501, 0.333)
        (2001, 0.413)
        (2501, 0.497)
        (3001, 0.609)
        (3501, 0.687)
        (4001, 0.717)
    };
    \addplot[
        color={blue},
        mark={+},
    ]
    coordinates {
        (501, 0.786)
        (1001, 1.006)
        (1501, 1.107)
        (2001, 1.285)
        (2501, 1.427)
        (3001, 1.563)
        (3501, 1.867)
        (4001, 1.917)
    };
    \addplot[
        color={green!60!black},
        mark={x},
    ]
    coordinates {
        (501, 1.756)
        (1001, 2.245)
        (1501, 2.3)
        (2001, 2.83)
        (2501, 3.062)
        (3001, 3.397)
        (3501, 3.718)
        (4001, 3.792)
    };
    \addplot[
        color={cyan},
        mark={square*},
    ]
    coordinates {
        (501, 3.096)
        (1001, 4.16)
        (1501, 4.419)
        (2001, 4.953)
        (2501, 5.26)
        (3001, 5.853)
        (3501, 6.307)
        (4001, 6.277)
    };
    \addplot[
        color={orange},
        mark={diamond*},
    ]
    coordinates {
        (501, 4.939)
        (1001, 5.788)
        (1501, 6.789)
        (2001, 8.093)
        (2501, 8.057)
        (3001, 9.093)
        (3501, 8.982)
        (4001, 8.927)
    };
    \addplot[
        color={purple},
        mark={triangle*},
    ]
    coordinates {
        (501, 7.306)
        (1001, 8.572)
        (1501, 10.151)
        (2001, 11.213)
        (2501, 11.797)
        (3001, 12.188)
        (3501, 12.192)
        (4001, 12.931)
    };
    \legend{$d = 50$,$d = 100$,$d = 150$,$d = 200$,$d = 250$,$d = 300$}
    \end{axis}
    \end{tikzpicture}
    
	\end{subfigure}
	\begin{subfigure}{0.19\textwidth}
		\begin{tikzpicture}[scale=0.4]
    \begin{axis}[
    every axis plot/.append style={line width=0.975pt},
    width=2.5*\textwidth,
        xlabel={Number of nodes},
        ylabel={Runtime [s]},
        title={Minimum CR},
        legend pos={north west},
        ymajorgrids={true},
        grid style={dashed},
    ]
    \addplot[
        color={red},
        mark={*},
    ]
    coordinates {
        (501, 0.263)
        (1001, 0.316)
        (1501, 0.358)
        (2001, 0.419)
        (2501, 0.444)
        (3001, 0.49)
        (3501, 0.504)
        (4001, 0.585)
    };
    \addplot[
        color={blue},
        mark={+},
    ]
    coordinates {
        (501, 1.257)
        (1001, 1.207)
        (1501, 1.155)
        (2001, 1.067)
        (2501, 1.115)
        (3001, 1.149)
        (3501, 1.177)
        (4001, 1.326)
    };
    \addplot[
        color={green!60!black},
        mark={x},
    ]
    coordinates {
        (501, 2.733)
        (1001, 2.504)
        (1501, 2.543)
        (2001, 2.305)
        (2501, 2.314)
        (3001, 2.381)
        (3501, 2.21)
        (4001, 2.329)
    };
    \addplot[
        color={cyan},
        mark={square*},
    ]
    coordinates {
        (501, 5.338)
        (1001, 4.579)
        (1501, 4.322)
        (2001, 4.297)
        (2501, 4.138)
        (3001, 4.372)
        (3501, 4.315)
        (4001, 4.069)
    };
    \addplot[
        color={orange},
        mark={diamond*},
    ]
    coordinates {
        (501, 7.352)
        (1001, 7.6)
        (1501, 7.046)
        (2001, 6.656)
        (2501, 7.691)
        (3001, 6.766)
        (3501, 7.637)
        (4001, 6.877)
    };
    \addplot[
        color={purple},
        mark={triangle*},
    ]
    coordinates {
        (501, 11.954)
        (1001, 9.805)
        (1501, 11.149)
        (2001, 10.538)
        (2501, 11.243)
        (3001, 11.189)
        (3501, 11.958)
        (4001, 11.763)
    };
    \legend{$d = 50$,$d = 100$,$d = 150$,$d = 200$,$d = 250$,$d = 300$}
    \end{axis}
    \end{tikzpicture}
    
	\end{subfigure}
	\caption{Empirical evaluation of different~\optdtfoil~queries over random synthetic data.}\label{fig:plots}
\end{figure*}
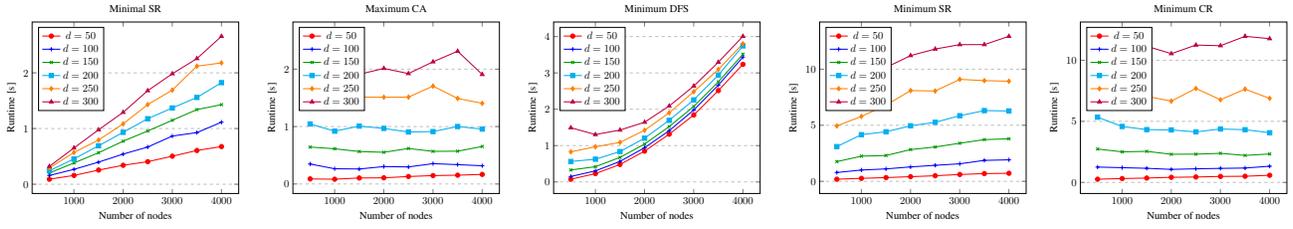

Our implementation consists of three main components:
(i) a prototype \emph{simplifier} for \qdtfoil/\optdtfoil~formulas, (ii) an encoder translating \dtfoil~formulas into CNF~formulas, and (iii) the algorithms to either compute answers for~\optdtfoil~queries, or decide the truth value of~\qdtfoil~formulas. 
We will give a high-level explanation and also present some key experiments--the supplementary material contains additional details.
%
\paragraph{Simplifier.}
Logical connectives can significantly increase the size of our resulting CNF formulas; consider for example the formula:
$\varphi(y) = \exists x \, [ \neg (\neg (\neg (\neg (\neg (x \subseteq y))))) \lor (1, 0, \bot) \subseteq (1, 1, 1)]$. 
	It is clear that double-negations can be safely eliminated, and also that sub-expressions involving only constants can be pre-processed and also eliminated, thus resulting in the simplified formula $\varphi(y) = \exists x \, \neg (x \subseteq y)$. 
\vspace{-2pt}
	\paragraph{Encoder.}
	We use standard encoding techniques for SAT-solving, for which we refer the reader to the \emph{Handbook of Satisfiability}~\citep{Handbook, Handbook2}.
	The basic variables of our propositional encoding are of the form $v_{x, i, s}$,  indicating that the $\dtfoil$ variable $x$ has value $s$ in its $i$-th feature, with $i \in \{1, \ldots, \dim(\T)\}$, for an input decision tree $\T$, and $s \in \{0, 1, \bot\}$. Then, the clauses (and further auxiliary variables)
	are mainly built on two layers of abstraction: a \emph{circuit layer}, and a \emph{first-order layer}. The circuit layer consists of individual ad-hoc encodings for each of the predicates and shorthands that appear frequently in queries, such as $\subseteq$, $\lel$, $\led$, $\sem$,  $\cons$, $\full$, $\allpos$ and $\allneg$. The first-order layer consists of encoding the logical connectives ($\neg, \lor, \land$) as well as the quantifiers (with the corresponding $\node$ and $\posl$ guards when appropriate). 
	
	For two interesting examples of encoding the circuits, let us consider $\lel$ and $\allpos$. For $\allpos$, we use a \emph{reachability} encoding,
        in which we create variables $r_{x, u}$ to represent that a node $u$ of $\T$ is \emph{reachable} by a partial instance $y$ subsuming $x$. We start by enforcing that $r_{x, \text{root}(\T)}$ is set to true, to then propagate the reachability from every node $u$ to its children $u \to 0$ and $u \to 1$ depending on the value of $x[a]$ with $a$ the label of $u$.
        Finally, by adding unit clauses stating that $(\neg r_{x, \ell})$ for every $\false$ leaf $\ell$, we have encoded that no instance $y$ subsuming $x$ reaches a false leaf, and thus is a positive instance.		 	
	For the case of $\lel$, we can see it as a pseudo-boolean constraint and implement it by leveraging the counting variables of~the \emph{sequential encoder} of \citet{Sinz_2005}~(cf.~\citep{Handbook2})    , thus amounting to a total of $O(\dim(\T)^2)$ auxiliary variables and clauses.

	For the first-order layer, we implement the Tseitin transformation~\citep{Tseitin_1968} to more efficiently handle $\neg$ and $\lor$, while treating the guarded-$\forall$ as a conjunction over the $O(|\T|)$ partial instances for which either $\node(\cdot)$ or $\posl(\cdot)$ holds, which can be precomputed from $\T$. An interesting problem that arises when handling negations or disjunctions is that of \emph{consistency constraints}, e.g., for each $i \in \{1, \ldots, \dim(\T)\}$, the clause $(v_{x, i, \bot} \lor v_{x, i, 0} \lor v_{x, i, 1})$ should be true.
        To address this, we partition the clauses of our encoding into two sets: $\textsc{ConsistencyCls}$ (consistency clauses) and $\textsc{SemanticCls}$ (semantic clauses), so that logical connectives operate only over $\textsc{SemanticCls}$, preserving the internal consistency of our variables, both the original and auxiliary ones. 
\vspace{-3pt}
\paragraph{Algorithms for~\optdtfoil~and~\qdtfoil.}

In a nutshell, to compute the answer for an~\optdtfoil~query 
\[
	\Psi(x) = \minf[\varphi[\es'_1, \ldots, \es'_k](x), \rho[\es''_1, \ldots, \es''_\ell](y, z)],
	\]
 we first find, through a single SAT call, a partial instance $\es$ that satisfies $\varphi[\es'_1, \ldots, \es'_k](\es)$. Letting $\es^{(0)} := \es$ be the obtained partial instance, we will iteratively search for a {\em smaller} (according to $\rho$) partial instance $\es^{(i)}$ that satisfies 
 \(
	\varphi[\es'_1, \ldots, \es'_k](\es^{(i)}) \land \rho[\es''_1, \ldots, \es''_\ell](\es^{(i)}, \es^{(i-1)}).
 \)
 Whenever such a partial instance $\es^{(i)}$ is not found, we conclude that $\es^{(i-1)}$ is the answer. The proof of~\Cref{theo-comp-optdtfoil} ensures the number of iterations of this process is polynomial in $\dim(\T)$.

 In the case of a~\qdtfoil~query, which can be in turn a boolean combination of smaller~\qdtfoil~formulas, we recursively evaluate each of the sub-formulas, and then combine the results according to the logical connectives. 
The base case of the recursion corresponds to \dtfoil~formulas, potentially with a single kind of quantifier, whose truth value requires solving a single CNF formula. 
\vspace{-2pt}
\paragraph{Results, challenges and next steps.}
Our implementation is able to handle all queries considered in this paper, and provides an elegant way to specify further queries.
We have tested our implementation on a variety of decision trees, and found that it scales up to thousands of nodes and hundreds of features, as shown in~\Cref{fig:plots}. 
Moreover, a comprehensive suite of tests validates our implementation, with over 2600 unit and integration tests.
Nonetheless, our implementation cannot be considered complete for~\optdtfoil~or~\qdtfoil; a main challenge is that, even though atomic~\dtfoil~queries can be evaluated in polynomial time (cf.~\Cref{theo:ptime-atomic}), we do not know of an efficient propositional encoding for them, as the algorithm provided by the proof of~\Cref{theo:ptime-atomic} is only of theoretical interset.
Naturally, some concrete atomic~\dtfoil~queries can be encoded efficiently, as we have done for, e.g.,~\led,~\cons.

In terms of future work, two roads could make our implementation more suitable for practice: (i) extending our logic to handle multi-class trees with numerical features, and (ii) using incremental SAT-solving techniques~\citep{nadelUltimatelyIncrementalSAT2014}~to speed up the minimization algorithm for~\optdtfoil.

\clearpage
\bibliographystyle{kr}	
\bibliography{main}

\newpage

\appendix
\onecolumn
\newpage


\section{Supplementary Material}
\paragraph{Organization.} 
This additional material contains a more extensive exposition of our implementation, the proofs for each of our theorems, as well as intermediate results and required definitions. Its organization is as follows:
\begin{itemize}
	\item Subsection~\ref{app:experiments} shows the details of our experiments (datasets, training, testing).
	\item Subsection~\ref{app:implementation} details our implementation, the challenges we managed to overcome to create a useful tool.
	\item Subsection~\ref{app:games} presents a brief review of {\em \EF} games.
	\item Subsection~\ref{subsec:bhierarchy} overviews the Boolean Hierarchy complexity class.
	\item Subsection~\ref{app:def-suf} shows how the $\suf$ predicate is expressible in \foil.
	\item Subsection~\ref{app:def-led} shows how the $\led$ predicate is expressible in $\dtfoil$.
	\item Subsection~\ref{app:proof-1} presents the proof of~Theorem~\ref{thm:ne-foil}.
	\item Subsection~\ref{app:eval-folistar} presents the proof of Theorem~\ref{thm:eval-folistar}.
	\item Subsection~\ref{app:subseteq-lel} shows how $\lel$ and $\subseteq$ cannot be defined in terms of each other.
	\item Subsection~\ref{app:ptime-atomic} presents the proof of~Theorem~\ref{theo:ptime-atomic}.
	\item Subsection~\ref{app:eval-folistarpm} presents the proof of Theorem~\ref{thm:eval-qdtfoil}.
	\item Subsection~\ref{app:comp-qdtfoil} presents the proof of Theorem~\ref{thm:comp-qdtfoil}.
	\item Subsection~\ref{app:prop-strict-po} presents the proof of Proposition~\ref{prop-strict-po}.
	\item Subsection~\ref{app:prop-ep-optdtfoil} presents the proof of Proposition~\ref{prop-ep-optdtfoil}.
	\item Subsection~\ref{app:theo-comp-optdtfoil} presents the proof of Theorem~\ref{theo-comp-optdtfoil}.
\end{itemize}
\subsection{Experiments}
\label{app:experiments}

\paragraph{Datasets and Training.} For our experiments we use 2 different
datasets:
\begin{enumerate}
	\item A binarized version of the \emph{MNIST}~\citep{deng2012mnist} dataset,
		where each pixel is binarized by comparing its grayscale value (i.e.,
		from $0$ to $255$) with $70$. As MNIST is a multi-class dataset, we
		define for each digit $d \in \{0, \ldots, 9\}$ the dataset
		$\textrm{MNIST}_d$ by making all images corresponding to digit $d$ of
		the \emph{positive class}, while the rest belong to the \emph{negative
		class}. Our rationale behind using MNIST, considering decision trees
		are not a natural choice for such a classification task, is that it
		presents a large number of features (784), thus establishing a
		challenge for our toolage, while its visual nature allows us for a more
		intuitive presentation of computed explanations (e.g., $\sr$ or
		$\sem$).
	\item A random synthetic datasets which consist simply of uniformly random
		vectors in $\{0, 1\}^d$ for a parameter $d$, whose classification label
		is also chosen uniformly at random. Using such a dataset we can freely
		experiment with the number of nodes and the dimension.
\end{enumerate}

To train the decision trees over these datasets we use the
\textsf{DecisionTreeClassifier} class of the
\texttt{scikit-learn}~\citep{Pedregosa} library. This class receives an upper
bound for the number of leaves, and thus the equation $\# \text{nodes} =
\left(2 \cdot \#\text{leaves} - 1\right)$, which holds for any full binary
tree, provides us control over the upper bound on the total number of nodes
(see, for example, \Cref{fig:plots}).

With respect to~\Cref{fig:plots}, each data point corresponds to an average over $8$ instances and $24$ random datasets with their corresponding trees. 
As it can be observed, the runtime over certain queries such as~$\minsr$ and~$\minsem$ depend strongly on the size of the tree, $|\T|$, whereas others such as~$\mca$ seem to only depend on the dimension. This is consistent with the ways the specific circuits involved in each query are encoded. 

\subsection{Detailed Implementation}
\label{app:implementation}

\paragraph{Software Architecture.} Our code is written and tested in Go 1.21.3
and no external libraries are required for its compilation. Currently the
implementation requires for queries to be defined at compile time. It takes
as input a Decision Tree model encoded in a custom JSON structure and
optionally the path to an executable SAT solver that the binary has access to.
The output of the implemented APIs can be one of the following: (1) the
value of an optimal instance of a~\optdtfoil~formula as a custom extension of
Go's base \texttt{[]int} type computed with the use of the solver passed as
input or (2) the internal CNF representation of~\dtfoil~formulas.
This CNF representation in turn can be written to a file following the
\href{http://www.satcompetition.org/2009/format-benchmarks2009.html}{DIMACS}
format, a standard of common use in official SAT competitions. As a direct
consequence, our implementation is solver-agnostic. Any solver whose input is a
DIMACS CNF can be plugged into the application, including more performant
future SAT solvers.

Regarding the software architecture, our implementation consist of a Go module
separated into 6 main packages:

\begin{enumerate}
	\item \textbf{Base}:
		the definition for the Component interface (the composition of which
		forms the internal encoding of~\dtfoil~formulas), the model partial
		instances and the \texttt{Context} type that exposes all methods
		required to interact with the model-base information as well as the
		storing and accessing \texttt{CNF} variables.
	\item \textbf{CNF}: the \texttt{CNF} type and associated methods
		responsible for the representation, manipulation, conjunction and
		negation of CNF formulas as well as their marshalling.
	\item \textbf{Trees}:
		all the model-based information. Specifically the \texttt{Tree} and
		\texttt{Node} types and the necessary methods for their unmarshalling.
	\item \textbf{Circuits}: includes the logic predicates and a set of
		useful~\dtfoil~formulas as Go types implementing the Component
		interface. For example, the \texttt{Not}, \texttt{And},
		\texttt{varFull} and \texttt{Trivial} types. For types to implement the
		Component interface they must implement the methods \textit{simplified}
		and \textit{encoding} methods, which make up the application's
		Simplifier and Encoder respectively.
	\item \textbf{Operators}: the first-order logical connectives as types that
		implement the Component interface.
	\item \textbf{Computation}: the implementation of the optimization
		algorithm and useful strict partial orders as types implementing the
		Component interface.
	\item \textbf{Experiments}: a set of experiments that uses all the other
		packages to generate metrics and figures.
\end{enumerate}

\paragraph{Simplifier.} We implemented basic optimizations over the user
queries, a very similar concept to the logical plan optimizer in a regular SQL
database management system. As mentioned before, it is coded as a method
required by the Component interface. Here we can distinguish two behaviors
depending on the component type.

In the case of first-order logical connectives, the \textit{simplified} method
propagates the simplification to their children checks if they are already
truth values, and rewrites the plan. For example, if the child of the
\textsf{Not} operator is \textsf{Trivial(false)}; then when the component
simplification is finished, the parent of \textsf{Not} will have a
\textsf{Trivial(true)} child during the calculation of its own simplification.
In the case of circuits, the \textit{simplified} method performs optimizations
over their underlying structure.

When all the variables of a given circuit are constants, then the Simplifier
calculates the truth value of the component. Notice that, this simplification
is possible because it is a polynomial check of a formula with respect to the
size of the model.

\paragraph{Encoder.}The \textit{encoding} method of each component builds the
CNF of the subtree under the node that represents the component in the
composition tree. In other words, a logical connective like \textsf{AND} takes
the CNF encoding of both its children and generates the CNF simply
concatenating them.

We distinguish between constants and variables. Each circuit family exposes a
different type for each representation used to indicate the semantics of its
partial instances. For instance, $x \subseteq y$ has 4 different types:
one for when both $x$ and $y$ are variables, one where both are constants and
two different types when only one of them is a constant.

It is worth noticing that the underlying~\dtfoil~formula of a~\qdtfoil~formula
$\varphi$ starting with a universal quantifier cannot directly be passed to a
SAT-solver. Therefore, if a formula $\varphi$ is of the form $\varphi = \forall
x\, P(x)$, then we apply the following identities:
\[
\varphi := \forall x\,  P(x) = \neg \left(\neg \forall x\, P(x)\right) =
\neg \left(\exists x\, \left(\neg P(x)\right)\right),
 \]
 and thus it is enough to invert the truth value of the formula $\psi = \exists
 x\, \left(\neg P(x)\right)$,  which is obtained with the aid of the
 SAT-solver.

\paragraph{Analysis on encoding size.} In order to keep the size of the
generated CNF formulas restricted, two main techniques were employed. The first
one was to divide the internal representation of CNF formulas into
$\textsc{ConsistencyCls}$ and $\textsc{SemanticCls}$ to reduce the amount of
clauses over which to operate when manipulating them.
The second decision was to introduce new Boolean variables representing the
entire circuits, i.e., a circuit's variable is true if, and only if, the
circuit is true. This representation allowed us to create our own efficient
transform into CNF because logical connectives would now operate over single
variables instead of applying negation and distribution over multiple clauses
and variables.

\paragraph{Testing.} Our application has a high level of robustness regarding
testing. The test we implement range from verifying the semantics of the
vocabulary to checking the correctness of the circuits using handmade decision
trees. The full suite of unit test contains 2696 individual tests.

    \subsection{Games for \FO\ distinguishability}
	\label{app:games}
	The {\em 
		quantifier rank} of an $\FO$ formula $\phi$, denoted by $\qr(\phi)$, 
	is the maximum depth of
	quantifier nesting in it. For $\astruct$, we write $\dom(\astruct)$ to denote its domain. 
	
	Some proofs in
	this paper make use of {\em \EF} (EF) games. This game is played in two
	structures, $\astruct_1$ and $\astruct_2$, 
	of the same schema, by two players, the {\em spoiler} and
	the {\em duplicator}. In round $i$ the spoiler selects a structure,
	say $\astruct_1$, and an element $c_i$ in $\dom(\astruct_1)$; the duplicator responds
	by selecting an element $e_i$ in $\dom(\astruct_2)$. The duplicator {\em wins} in
	$k$ rounds, for $k \geq 0$, if $\{(c_i,e_i) \mid i \leq k\}$ defines a
	partial isomorphism between $\astruct_1$ and $\astruct_2$. If the duplicator wins no
	matter how the spoiler plays, we write $\astruct_1 \equiv_k \astruct_2$. A classical
	result states that $\astruct_1 \equiv_k \astruct_2$ iff $\astruct_1$ and $\astruct_2$ agree on all
	$\FO$ sentences of quantifier rank $\leq k$ (cf. \citep{fmt-book}). 
	
	Also, if $\bar a$ is an $m$-tuple in $\dom(\astruct_1)$ and $\bar b$ is an $m$-tuple in
	$\dom(\astruct_2)$, where $m \geq 0$, 
	we write $(\astruct_1,\bar a) \equiv_k (\astruct_2,\bar b)$ whenever the duplicator
	wins in $k$ rounds no matter how the spoiler plays, but 
	starting from position $(\bar a,\bar b)$. In the same way, $(\astruct_1,\bar a) \equiv_k (\astruct_2,\bar b)$ iff for every $\FO$ formula $\phi(\bar x)$ 
	of quantifier rank $\leq k$, it holds that $\astruct_1 \models \phi(\bar a) \Leftrightarrow \astruct_2
	\models \phi(\bar b)$.   
	
	It is well-known (cf.~\citep{fmt-book}) that  
	there are only finitely many $\FO$ formulae of quantifier rank
	$k$, up to logical equivalence. The {\em rank-$k$ type} 
	of an $m$-tuple $\bar a$ in a structure $\astruct$ is the set of all formulae
	$\phi(\bar x)$ of quantifier rank $\leq k$ such that $\astruct \models
	\phi(\bar a)$. Given the above, there are only finitely many rank-$k$
	types, and each one of them is definable by an $\FO$ formula 
	$\tau_k^{(\astruct,\bar a)}(\bar x)$ of
	quantifier rank $k$. 
	
	\subsection{Boolean Hierarchy overview}
	\label{subsec:bhierarchy}
	
	Boolean operations between complexity classes are defined as follows \citep{DBLP:conf/fct/Wechsung85}:
	\begin{enumerate}
		\item $A \vee B = \{L_A \cup L_B \ | \ L_A \in A \text{ and } L_B \in B \}$
		\item $A \wedge B = \{L_A \cap L_B \ | \ L_A \in A \text{ and } L_B \in B \}$
		\item $\rm{co}A = \{\overline{L} \ | \ L \in A\}$
	\end{enumerate}
	The Boolean Hierarchy $\bh$ is the set of languages such that $\bh = \cup_{k \geq 1} \bh_{k}$ \citep{DBLP:journals/siamcomp/CaiGHHSWW88}, where:
	\begin{enumerate}
		\item $\bh_{1} = \np$
		\item $\bh_{2i} =  \bh_{2i-1} \wedge \conp$
		\item $\bh_{2i + 1} =  \bh_{2i} \vee \np$
	\end{enumerate}
	It is easy to see that $\ptime, \np, \conp$ are all included in $\bh$. In
        the case of $\conp$, the language $L = \Sigma^{*}$ is in
        $\np$. Then, for every $L' \in \conp$, the problem $L \cap L'$
        is in $\bh_{2}$. Thus, $\conp \subseteq \bh_{2}$.
	With this definition, for every Boolean combination of $\np$
        and $\conp$ problems, there exists a number $k$ such that the
        combination is in $\bh_{k}$.
	
	As usual, let us denote by $\sat$ the set of propositional
        formulas that are satisfiable, and by $\unsat$ the set of
        propositional formulas that are not satisfiable. There is a
        family of decision problems known to be complete for every
        level of the Boolean hierarchy. Given $k \geq 1$ and $\varphi_1, ..., \varphi_k$
        propositional formulas, let $\sate(\varphi_1, ..., \varphi_k)$
        be defined as follow:
	\begin{enumerate}
		\item $\sate(\varphi_1) = \varphi_1 \in \sat$
		\item $\sate(\varphi_1, ..., \varphi_{2k}) =  \sate(\varphi_1, ..., \varphi_{2k-1}) \wedge \varphi_{2k} \in \unsat$
		\item $\sate(\varphi_1, ..., \varphi_{2k+1}) =  \sate(\varphi_1, ..., \varphi_{2k}) \vee \varphi_{2k+1} \in \sat$
	\end{enumerate}
	Then define the family of problems $\sat(k) = \{(\varphi_1, ...,
        \varphi_k)\ | \ \sate(\varphi_1, ..., \varphi_k)\ \text{holds}
        \}$. For every $k \geq 1$, $\sat(k)$ is $\bh_k$-complete
        \citep{DBLP:journals/siamcomp/CaiGHHSWW88}.

        \subsection{Definition of formula \suf$(x,y)$ in \foil}
        \label{app:def-suf}
        Recall that $x \subset y$ is a shorthand for the formula $x \subseteq y \wedge \neg (y \subseteq x)$. Then define the
        following auxiliary predicates:
\begin{align*}
L_0(x) \ &= \ \forall y \, (x \subseteq y),\\ L_1(x) \ &= \ \exists
y \, (L_0(y) \wedge y \subset x \wedge \neg \exists z \,
(y \subset z \wedge z \subset x)),\\
\opp(x,y) \ &= \ L_1(x) \wedge L_1(y) \wedge \neg \exists z \, (x \subseteq z \wedge y \subseteq z).
\end{align*}
The interpretation of these predicates is such
that for every model $\M$ of dimension $n$ and every pair of partial
instances $\es, \es'$ of dimension $n$, it holds that: $\M \models L_0(\es)$ if and only if $|\es_\bot| = n$ (the set of partial instances with 0 defined features), $\M \models L_1(\es)$ if and only if $|\es_\bot| = n - 1$ (the set of partial instances with 1 defined feature), $\M \models \opp(\es, \es')$ if and only if $\es$ and $\es'$ have the exact same defined feature and it is flipped. By using these predicates, we define a binary predicate $\suf$ such that $\M \models \suf(\es, \es')$ if and only if $\es_\bot = \es'_\bot$. More precisely, we have that:
\begin{align*}
\suf&(x,y) \ = \ \forall u \forall v \, (\opp(u,v) \ \to\\
&[(u \subseteq x \wedge u \subseteq y \wedge \neg(v \subseteq
x) \wedge \neg(v \subseteq y)) \ \vee\\
&\phantom{[}(v \subseteq x \wedge v \subseteq y \wedge \neg(u \subseteq
x) \wedge \neg(u \subseteq y)) \ \vee\\
&\phantom{[}(u \subseteq x \wedge v \subseteq y \wedge \neg(u \subseteq
y) \wedge \neg(v \subseteq x)) \ \vee\\
&\phantom{[}(u \subseteq y \wedge v \subseteq x \wedge \neg(u \subseteq
x) \wedge \neg(v \subseteq y)) \ \vee\\
&\phantom{[}(\neg(u \subseteq x) \wedge \neg(u \subseteq y) \wedge \neg(v \subseteq
x) \wedge \neg(v \subseteq y))]).
\end{align*}
        
\subsection{Definition of formula \led$(x,y,z)$ in \dtfoil}
\label{app:def-led}

Let $\meet$ (\emph{Greatest Lower Bound}) be the following formula:
\begin{align*}
	\meet(x,y,z) \ = \ z \subseteq x \wedge z \subseteq y \wedge \forall w \, ((w \subseteq x \wedge w \subseteq y) \to w \subseteq z).
\end{align*}
The interpretation of these predicates is such that for every model $\M$ of dimension $n$ and every sequence of instances $\es_1, \es_2, \es_3$ it holds that $\M \models \meet(\es_1, \es_2, \es_3)$ if and only if $\es_3$ is the greatest lower bound contained by $\es_1$ and $\es_2$, i.e., the instance with most defined features subsumed by both. This property allow us to measure the amount of coincidental values between two instances. By using this predicate, let $\led$ be a ternary predicate such that $\M \models \led(\es_1, \es_2, \es_3)$ if and only if the Hamming distance between $\es_1$ and $\es_2$ is less or equal than the Hamming distance between $\es_1$ and $\es_3$. Relation $\led$ can be expressed as follows in \foil:
\begin{align*}
	\led(x,y,z) \ = \ \full(x) \wedge \full(y) \wedge \full(z) \ \wedge \exists w_1 \exists w_2 \, (\meet(x,y,w_1) \wedge
	\meet(x,z,w_2) \wedge w_2 \lel w_1).
\end{align*}
\subsection{Proof of Theorem \ref{thm:ne-foil}}
\label{app:proof-1}
		The proof, which extends techniques from \citep{fmt-book}, requires introducing some terminology and establishing some intermediate lemmas. 
		
		Let $\M$, $\M'$ be models of dimension $n$ and $p$, respectively, 
		and consider the structures $\astruct_\M =  \langle \{0,1,\bot\}^n,\subseteq^{\astruct_\M},\pos^{\astruct_\M} \rangle$ and 
		$\astruct_{\M'} =  \langle \{0,1,\bot\}^p,\subseteq^{\astruct_{\M'}},\pos^{\astruct_{\M'}} \rangle$. 
		We write $\astruct_\M \oplus \astruct_{\M'}$ for the structure over the same vocabulary that satisfies the following:
		\begin{itemize} 
			\item The domain of $\astruct_\M \oplus \astruct_{\M'}$ is $\{0,1,\bot\}^{n+p}$. 
			\item The interpretation of $\subseteq$ on $\astruct_\M \oplus \astruct_{\M'}$ is the usual subsumption relation on $\{0,1,\bot\}^{n+p}$. 
			\item The interpretation of $\pos$ on $\astruct_\M \oplus \astruct_{\M'}$ is the set of instances $\es \in \{0,1\}^{n+p}$ such that 
			$(\es[1],\cdots,\es[n]) \in \pos^{\astruct_\M}$ or $(\es[n+1],\cdots,\es[n+p]) \in \pos^{\astruct_{\M'}}$. 
		\end{itemize} 
		
		We start by establishing the following composition lemma.  
		
		\begin{lemma} 
			\label{lemma:comp} 
			Consider models $\M$, $\M_1$, and $\M_2$ of dimension $n$, $p$, and $q$, respectively, 
			and assume that $(\astruct_{\M_1},\{1\}^p) \equiv_k (\astruct_{\M_2},\{1\}^q)$. Then it is the case that 
			\begin{align*}
				\big(\astruct_\M \oplus \astruct_{\M_1},\{1\}^{n+p},\{\bot\}^n \times \{1\}^p\big) \ \equiv_k \
				\big(\astruct_\M \oplus \astruct_{\M_2},\{1\}^{n+q},\{\bot\}^n \times \{1\}^q\big).
			\end{align*}
		\end{lemma} 
		
		\begin{proof} 
			Let $\es_i$ and $\es'_i$ be the moves played by Spoiler and Duplicator in $\astruct_\M \oplus \astruct_{\M_1}$ and $\astruct_\M \oplus \astruct_{\M_2}$, respectively, for the first $i \leq k$ rounds of the \EF game 
			\begin{align*}
				\big(\astruct_\M \oplus \astruct_{\M_1},\{1\}^{n+p},\{\bot\}^n \times \{1\}^p\big) \ \equiv_k \
				\big(\astruct_\M \oplus \astruct_{\M_2},\{1\}^{n+q},\{\bot\}^n \times \{1\}^q\big).
			\end{align*}
			We write $\es_i = (\es_{i1},\es_{i2})$ to denote that $\es_{i1}$ is the tuple formed by the first $n$ features of $\es_i$ and $\es_{i2}$ is the one formed 
			by the last $p$ features of $\es_i$. Similarly, we write $\es'_i = (\es'_{i1},\es'_{i2})$ to denote that $\es'_{i1}$ is the tuple formed by the first $n$ features of $\es'_i$ and $\es'_{i2}$ is the one formed 
			by the last $q$ features of $\es'_i$. 
			
			The winning strategy for Duplicator is as follows. Suppose $i-1$ rounds have been played, and for round $i$ 
			the Spoiler picks element $\es_i \in \astruct_\M \oplus \astruct_{\M_1}$ (the case when he picks an element in $\astruct_\M \oplus \astruct_{\M_2}$ is symmetric). 
			Assume also that $\es_i = (\es_{i1},\es_{i2})$. The duplicator then considers the position 
			$$\big((\es_{12},\dots,\es_{(i-1)2}),(\es'_{12},\dots,\es'_{(i-1)2})\big)$$
			on the game $(\astruct_{\M_1},\{1\}^p) \equiv_k (\astruct_{\M_2},\{1\}^q)$, and finds his response $\es'_{i2}$ to $\es_{i2}$ in $\astruct_{\M_2}$. 
			The Duplicator then responds to the Spoiler's move $\es_i \in \astruct_\M \oplus \astruct_{\M_1}$ by choosing the element $\es'_i = (\es_{i1},\es'_{i2}) \in \astruct_\M \oplus \astruct_{\M_2}$.
			
			Notice, by definition, that $\es_i = \{1\}^{n+p}$ iff $\es'_i =  \{1\}^{n+q}$. Similarly, $\es_i = \{\bot\}^n \times \{1\}^{p}$ iff $\es'_i =  \{\bot\}^n \times 
			\{1\}^{q}$. Moreover, it is easy to see that playing in this way the Duplicator preserves the subsumption relation. Analogously, the strategy preserves the $\pos$ relation. 
			In fact, $\es_i$ is a positive instance of $\astruct_\M \oplus \astruct_{\M_1}$ iff $\es_{i1}$ is a positive instance of $\astruct_\M$ or $\es_{i2}$ is a positive instance 
			of $\astruct_{\M_1}$. By definition, the latter follows if, and only if,  $\es_{i1}$ is a positive instance of $\astruct_\M$ or $\es'_{i2}$ is a positive instance 
			of $\astruct_{\M_2}$, which in turn is equivalent to $\es'_i$ being a positive instance of $\astruct_\M \oplus \astruct_{\M_2}$. 
			
			We conclude that this is a winning strategy for the Duplicator, and hence that 
			\begin{align*}
				\big(\astruct_\M \oplus \astruct_{\M_1},\{1\}^{n+p},\{\bot\}^n \times \{1\}^p\big) \ \equiv_k \
				\big(\astruct_\M \oplus \astruct_{\M_2},\{1\}^{n+q},\{\bot\}^n \times \{1\}^q\big).
			\end{align*}
			This finishes the proof of the lemma.  
		\end{proof} 
		
		We now consider structures of the form $\astruct_n = \langle \{0,1,\bot\}^n,\subseteq^\astruct \rangle$, where $\subseteq$ is interpreted as the subsumption relation 
		over $\{0,1,\bot\}^n$. For any such structure, we 
		write $\astruct_n^+$ for the structure over the vocabulary $\{\subseteq,\pos\}$ that extends $\astruct$ by adding only the tuple $\{1\}^n$ 
		to the interpretation of $\pos$. The following lemma is crucial for our proof. 
		
		
		\begin{lemma} 
			\label{lemma:pointed}
			If $n,p \geq 3^k$, then 
			$\astruct_n^+ \equiv_k \astruct_p^+$. 
		\end{lemma} 
		
		\begin{proof} 
			We prove something slightly more general. For doing so, we need to define several concepts. 
			Let $U = \{a_i \mid i > 0\}$ be a countably infinite set. 
			We take a disjoint copy $\overline U = \{\overline a_i \mid i > 0\}$ of $U$. 
			For an $X \subseteq U \cup \overline U$, we define
			\begin{align*} 
				X_{U \setminus \overline U} \ & := \ \{a \in U \mid a \in X \text{ and } \overline a \not\in X\} \\ 
				X_{U \cap \overline U} \ & := \ \{a \in U \mid a \in X \text{ and } \overline a \in X\} \\ 
				X_{\overline U \setminus U} \ & := \ \{\overline a \in \overline U \mid \overline a \in X \text{ and } a \not\in X\}
			\end{align*} 
			The {\em $\ell$-type} of $X$, for $\ell \geq 0$, is the tuple 
			$$\big(\min{\{\ell,|X_{U \setminus \overline U}|\}},\, \min{\{\ell,|X_{U \cap \overline U}|\}}, \,  \min{\{\ell,|X_{\overline U \setminus U|}|\}}\big).$$
			We write $X \leftrightarrows_\ell X'$, for $X,X' \subseteq U \cup \overline U$, if $X$ and $X'$ have the same 
			$\ell$-type. 
			If $X \subseteq U \cup \overline U$, then $X$ is {\em well formed} (wf) if for each $i > 0$ at most one element from $\{a_i,\overline a_i\}$ is in $X$.  
			
			We start with the following combinatorial lemma. 
			
			\begin{lemma}  
				\label{lemma:double} 
				Assume that $X \leftrightarrows_{3^{k+1}} Y$, for $X,Y \subseteq U \cup \overline U$ and $k \geq 0$. Then: 
				\begin{itemize} 
					\item 
					For every wf $X_1 \subseteq X$, there exists a wf $Y_1 \subseteq Y$
					such that 
					$$X_1 \leftrightarrows_{3^k} Y_1 \ \ \text{ and } \ \ X\setminus X_1 \leftrightarrows_{3^k} Y \setminus Y_1.$$
					\item 
					For every wf $Y_1 \subseteq Y$, there exists a wf $X_1 \subseteq X$
					such that 
					$$X_1 \leftrightarrows_{3^k} Y_1 \ \ \text{ and } \ \ X\setminus X_1 \leftrightarrows_{3^k} Y \setminus Y_1.$$
				\end{itemize}  
			\end{lemma} 
			
\begin{proof}
Given $Z \subseteq U$, we use $\overline Z$ to denote the set $\{ \overline a \in \overline U \mid a \in Z\}$, and given  $W \subseteq \overline U$, we use $\overline W$ to denote the set $\{ a \in U \mid \overline a \in W\}$. 
  Let $X_1$ be a wf subset of $X$. Then we have that $X_1 = X_{1,1} \cup X_{1,2} \cup X_{1,3} \cup X_{1,4}$, where
  \begin{eqnarray*}
    X_{1,1} & \subseteq & X_{U \setminus \overline U},\\
    X_{1,2} & \subseteq & \{ a \in U \mid a \in X_{U \cap \overline U}\},\\
    X_{1,3} & \subseteq & \{ \overline a \in \overline U \mid a \in X_{U \cap \overline U}\},\\
    X_{1,4} & \subseteq & X_{\overline U \setminus U},
  \end{eqnarray*}
  and $\overline{X_{1,2}} \cap X_{1,3} = \emptyset$ (since $X$ is wf).
  We construct a set $Y_1 = Y_{1,1} \cup Y_{1,2} \cup Y_{1,3} \cup Y_{1,4}$
  by considering the following rules.
  \begin{enumerate}
  \item If $|X_{U \setminus \overline U}| < 3^{k+1}$, then $|Y_{U \setminus \overline U}|
    = |X_{U \setminus \overline U}|$ since $X \leftrightarrows_{3^{k+1}} Y$. In
    this case, we choose $Y_{1,1} \subseteq Y_{U \setminus \overline
      U}$ in such a way that $|Y_{1,1}| = |X_{1,1}|$ and $|Y_{U
      \setminus \overline U} \setminus Y_{1,1}| = |X_{U \setminus
      \overline U} \setminus X_{1,1}|$.

    If $|X_{U \setminus \overline U}| \geq 3^{k+1}$, then $|Y_{U \setminus \overline U}|
    \geq 3^{k+1}$ since $X \leftrightarrows_{3^{k+1}} Y$. In this
    case, we choose $Y_{1,1} \subseteq Y_{U \setminus \overline U}$ in
    the following way. If $|X_{1,1}| < 3^k$, then $|Y_{1,1}| =
    |X_{1,1}|$, and if $|X_{U \setminus \overline U} \setminus
    X_{1,1}| < 3^k$, then $|Y_{U \setminus \overline U} \setminus
    Y_{1,1}| = |X_{U \setminus \overline U} \setminus
    X_{1,1}|$. Finally, if $|X_{1,1}| \geq 3^k$ and $|X_{U \setminus
      \overline U} \setminus X_{1,1}| \geq 3^k$, then $|Y_{1,1}| \geq
    3^k$ and $|Y_{U \setminus \overline U} \setminus Y_{1,1}| \geq
    3^k$. Notice that we can choose such a set $Y_{1,1}$ since $|Y_{U
      \setminus \overline U}| \geq 3^{k+1}$.

  \item If $|X_{U \cap \overline U}| < 3^{k+1}$, then $|Y_{U
    \cap \overline U}| = |X_{U \cap \overline U}|$ since $X
    \leftrightarrows_{3^{k+1}} Y$. In this case, we choose $Y_{1,2}
    \subseteq \{ a \in U \mid a \in Y_{U \cap \overline U}\}$ and
    $Y_{1,3} \subseteq \{ \overline a \in \overline U \mid a \in Y_{U
      \cap \overline U}\}$ in such a way that
    $\overline{Y_{1,2}} \cap Y_{1,3} = \emptyset$,
    $|Y_{1,2}| = |X_{1,2}|$, 
    $|Y_{1,3}| = |X_{1,3}|$ and
    $|Y_{U \cap \overline U} \setminus (Y_{1,2} \cup \overline{Y_{1,3}})| =
    |X_{U \cap \overline U} \setminus (X_{1,2} \cup \overline{X_{1,3}})|$. 
    
    If $|X_{U \cap \overline U}| \geq 3^{k+1}$, then $|Y_{U
    \cap \overline U}| \geq 3^{k+1}$ since $X
    \leftrightarrows_{3^{k+1}} Y$. In this case, we choose $Y_{1,2}
    \subseteq \{ a \in U \mid a \in Y_{U \cap \overline U}\}$ and
    $Y_{1,3} \subseteq \{ \overline a \in \overline U \mid a \in Y_{U
      \cap \overline U}\}$ in the following way.
    \begin{enumerate}
    \item
      If $|X_{1,2}| < 3^k$, $|X_{1,3}| < 3^k$ and $|X_{U \cap
        \overline U} \setminus (X_{1,2} \cup \overline{X_{1,3}})| \geq
      3^k$, then $|Y_{1,2}| = |X_{1,2}|$, $|Y_{1,3}| = |X_{1,3}|$ and
      $\overline{Y_{1,2}} \cap Y_{1,3} = \emptyset$.
      Notice that we can choose such sets $Y_{1,2}$ and $Y_{1,3}$
      since $|Y_{\overline U \setminus U}| \geq 3^{k+1}$.

    \item
            If $|X_{1,2}| < 3^k$, $|X_{1,3}| \geq 3^k$ and $|X_{U \cap
        \overline U} \setminus (X_{1,2} \cup \overline{X_{1,3}})| <
      3^k$, then $|Y_{1,2}| = |X_{1,2}|$,     $|Y_{U \cap \overline U} \setminus (Y_{1,2} \cup \overline{Y_{1,3}})| =
    |X_{U \cap \overline U} \setminus (X_{1,2} \cup \overline{X_{1,3}})|$ and    
      $\overline{Y_{1,2}} \cap Y_{1,3} = \emptyset$.
      Notice that we can choose such sets $Y_{1,2}$ and $Y_{1,3}$
      since $|Y_{\overline U \setminus U}| \geq 3^{k+1}$.

    \item
            If $|X_{1,2}| \geq 3^k$, $|X_{1,3}| < 3^k$ and $|X_{U \cap
        \overline U} \setminus (X_{1,2} \cup \overline{X_{1,3}})| <
      3^k$, then $|Y_{1,3}| = |X_{1,3}|$,     $|Y_{U \cap \overline U} \setminus (Y_{1,2} \cup \overline{Y_{1,3}})| =
    |X_{U \cap \overline U} \setminus (X_{1,2} \cup \overline{X_{1,3}})|$ and    
      $\overline{Y_{1,2}} \cap Y_{1,3} = \emptyset$.
      Notice that we can choose such sets $Y_{1,2}$ and $Y_{1,3}$
      since $|Y_{\overline U \setminus U}| \geq 3^{k+1}$.

    \item
            If $|X_{1,2}| < 3^k$, $|X_{1,3}| \geq 3^k$ and $|X_{U \cap
        \overline U} \setminus (X_{1,2} \cup \overline{X_{1,3}})| \geq
      3^k$, then $|Y_{1,2}| = |X_{1,2}|$, $|Y_{1,3}| \geq 3^k$, $|Y_{U \cap \overline U} \setminus (Y_{1,2} \cup \overline{Y_{1,3}})| \geq 3^k$
      and    
      $\overline{Y_{1,2}} \cap Y_{1,3} = \emptyset$.
      Notice that we can choose such sets $Y_{1,2}$ and $Y_{1,3}$
      since $|Y_{\overline U \setminus U}| \geq 3^{k+1}$.

    \item
            If $|X_{1,2}| \geq 3^k$, $|X_{1,3}| < 3^k$ and $|X_{U \cap
        \overline U} \setminus (X_{1,2} \cup \overline{X_{1,3}})| \geq
      3^k$, then $|Y_{1,3}| = |X_{1,3}|$, $|Y_{1,2}| \geq 3^k$, $|Y_{U \cap \overline U} \setminus (Y_{1,2} \cup \overline{Y_{1,3}})| \geq 3^k$
      and    
      $\overline{Y_{1,2}} \cap Y_{1,3} = \emptyset$.
      Notice that we can choose such sets $Y_{1,2}$ and $Y_{1,3}$
      since $|Y_{\overline U \setminus U}| \geq 3^{k+1}$.

    \item
            If $|X_{1,2}| \geq 3^k$, $|X_{1,3}| \geq 3^k$ and $|X_{U \cap
        \overline U} \setminus (X_{1,2} \cup \overline{X_{1,3}})| <
      3^k$, then $|Y_{U \cap \overline U} \setminus (Y_{1,2} \cup \overline{Y_{1,3}})| =
      |X_{U \cap \overline U} \setminus (X_{1,2} \cup \overline{X_{1,3}})|$, $|Y_{1,2}| \geq 3^k$, $|Y_{1,3}| \geq 3^k$ and
      $\overline{Y_{1,2}} \cap Y_{1,3} = \emptyset$.
      Notice that we can choose such sets $Y_{1,2}$ and $Y_{1,3}$
      since $|Y_{\overline U \setminus U}| \geq 3^{k+1}$.

    \item
            If $|X_{1,2}| \geq 3^k$, $|X_{1,3}| \geq 3^k$ and $|X_{U \cap
        \overline U} \setminus (X_{1,2} \cup \overline{X_{1,3}})| \geq
            3^k$, then $|Y_{1,2}| \geq 3^k$, $|Y_{1,3}| \geq 3^k$, 
            $|Y_{U \cap \overline U} \setminus (Y_{1,2} \cup \overline{Y_{1,3}})| \geq 3^k$ and
      $\overline{Y_{1,2}} \cap Y_{1,3} = \emptyset$.
      Notice that we can choose such sets $Y_{1,2}$ and $Y_{1,3}$
      since $|Y_{\overline U \setminus U}| \geq 3^{k+1}$.

   \end{enumerate}

  \item If $|X_{\overline U \setminus U}| < 3^{k+1}$, then $|Y_{\overline U \setminus U}|
    = |X_{\overline U \setminus U}|$ since $X \leftrightarrows_{3^{k+1}} Y$. In
    this case, we choose $Y_{1,4} \subseteq Y_{\overline U \setminus
      U}$ in such a way that $|Y_{1,4}| = |X_{1,4}|$ and $|Y_{\overline U
      \setminus U} \setminus Y_{1,4}| = |X_{\overline U \setminus
      U} \setminus X_{1,4}|$.

    If $|X_{\overline U \setminus U}| \geq 3^{k+1}$, then $|Y_{\overline U \setminus U}|
    \geq 3^{k+1}$ since $X \leftrightarrows_{3^{k+1}} Y$. In this
    case, we choose $Y_{1,4} \subseteq Y_{\overline U \setminus U}$ in
    the following way. If $|X_{1,4}| < 3^k$, then $|Y_{1,4}| =
    |X_{1,4}|$, and if $|X_{\overline U \setminus U} \setminus
    X_{1,4}| < 3^k$, then $|Y_{\overline U \setminus U} \setminus
    Y_{1,4}| = |X_{\overline U \setminus U} \setminus
    X_{1,4}|$. Finally, if $|X_{1,4}| \geq 3^k$ and $|X_{\overline U \setminus
      U} \setminus X_{1,4}| \geq 3^k$, then $|Y_{1,4}| \geq
    3^k$ and $|Y_{\overline U \setminus U} \setminus Y_{1,4}| \geq
    3^k$. Notice that we can choose such a set $Y_{1,4}$ since $|Y_{\overline U
      \setminus U}| \geq 3^{k+1}$.                          
  \end{enumerate}
  By definition of $Y_{1,1}$, $Y_{1,2}$, $Y_{1,3}$ and $Y_{1,4}$, it is straightforward to conclude that
  $Y_1$ is wf, $X_1 \leftrightarrows_{3^{k}} Y_1$ and  $(X \setminus X_1) \leftrightarrows_{3^{k}} (Y \setminus Y_1)$.

  We have just proved that for every wf $X_1 \subseteq X$, there
  exists a wf $Y_1 \subseteq Y$ such that $X_1 \leftrightarrows_{3^k}
  Y_1$ and $X\setminus X_1 \leftrightarrows_{3^k} Y \setminus Y_1$.
  In the same way, it can be shown that for every wf $Y_1 \subseteq
  Y$, there exists a wf $X_1 \subseteq X$ such that $X_1
  \leftrightarrows_{3^k} Y_1$ and $X\setminus X_1
  \leftrightarrows_{3^k} Y \setminus Y_1$. This concludes the proof of the lemma.

  \end{proof}

			In the rest of the proof we make use of structures of the form 
			$\astruct^* = \langle 2^{X},\subseteq^{\astruct^*} \rangle$, where $X \subseteq U \cup \overline U$ and 
			$\subseteq^{\astruct^*}$ is the relation that contains all pairs $(Y,Z)$, for $Y,Z \subseteq X$, such that 
			$Y \subseteq Z$. Given two structures $\astruct^*_1$ and $\astruct^*_2$ of this form, perhaps with constants, 
			we write $\astruct^*_1 \equiv_k^{{\rm wf}} \astruct^*_2$ to denote that the Duplicator has a winning strategy in 
			the $k$-round \EF game played on structures $\astruct^*_1$ and $\astruct^*_2$, but where Spoiler and Duplicator are forced to play wf subsets of $U \cup \overline U$ only.  
			
			Consider structures $\astruct^*_1 = \langle 2^{X_1},\subseteq^{\astruct^*_1} \rangle$ and 
			$\astruct^*_2 = \langle 2^{X_2},\subseteq^{\astruct^*_2} \rangle$
			of the form described above. We claim that, for every $k \geq 0$, 
			\begin{align*}
			X_1 \leftrightarrows_{3^k} X_2 \quad \Longrightarrow 
			\quad \big(\astruct^*_1,(X_1 \cap U)\big) \ \equiv_k^{{\rm wf}} \  \big(\astruct^*_2,(X_2 \cap U)\big). \qquad \quad (\dagger)
			\end{align*}
			
			Before proving the claim in $(\dagger)$, we explain how it implies Lemma \ref{lemma:pointed}. Take a structure 
			of the form $\astruct_n = \langle \{0,1,\bot\}^n,\subseteq^{\astruct_n}\rangle$, where $\subseteq^{\astruct_n}$ is the subsumption relation 
			over $\{0,1,\bot\}^n$. Take, on the other hand, the structure $\astruct_n^* = \langle 2^X,\subseteq^{\astruct_n^*} \rangle$, 
			where $X = \{a_1,\dots,a_n,\bar a_1,\dots,\bar a_n\}$. It can be seen that there is an isomorphism $f$ between 
			$\astruct_n$ and the substructure of $\astruct_n^*$ induced by the wf subsets of $X$. The isomorphism $f$ takes an instance $\es \in \{0,1,\bot\}^n$ and 
			maps it to $Y \subseteq X$ such that for every $i \in \{1,\dots,n\}$, (a) 
			if $\es[i] = 1$ then $a_i \in Y$, (b) if $\es[i] = 0$ then $\bar a_i \in Y$, and (c) if $\es[i] = \bot$ then neither $a_i$ nor $\bar a_i$ is in $Y$. 
			By definition, the isomorphism $f$ maps the tuple $\{1\}^n$ in $\astruct_n$ to the set 
			$X \cap U = \{a_1,\dots,a_n\}$ in $\astruct^*_n$. 
			
			From the claim in $(\dagger)$, it follows then that if $n,p \geq 3^k$ it is the case that $$(\astruct^*_n,\{a_1,\dots,a_n\}) \ \equiv_k^{{\rm wf}} \ 
			(\astruct^*_p,\{a_1,\dots,a_p\}).$$ 
			From our previous observations, this implies that 
			$$(\astruct_n,\{1\}^n) \ \equiv_k \  (\astruct_p,\{1\}^p).$$
			We conclude, in particular, that $\astruct^+_n \equiv_k \astruct^+_p$, as desired.    
			
			We now prove the claim in $(\dagger)$. 
			We do it by induction on $k \geq 0$. The base 
			cases $k = 0,1$ are immediate. We now move to the induction case for $k+1$. Take structures 
			$\astruct^*_1 = \langle 2^{X_1},\subseteq^{\astruct^*_1} \rangle$ and 
			$\astruct^*_2 = \langle 2^{X_2},\subseteq^{\astruct^*_2} \rangle$
			of the form described above, such that $X_1 \leftrightarrows_{3^{k+1}} X_2$. 
			Assume, without loss of generality, that for the first round the Spoiler picks the well formed element 
			$X'_1 \subseteq X_1$ in the structure $\astruct^*_1$. From Lemma \ref{lemma:double}, there exists 
			$X'_2 \subseteq X_2$ such that 
			$$X'_1 \leftrightarrows_{3^k} X'_2 \ \ \text{ and } \ \ X_1\setminus X'_1 \leftrightarrows_{3^k} X_2 \setminus X'_2.$$
			By induction hypothesis, the following holds: 
			\begin{eqnarray*}
                          \big(\langle 2^{X'_1},\subseteq \rangle,(X'_1 \cap U)\big) & \equiv^{\rm wf}_k & \big(\langle 2^{X'_2},\subseteq \rangle,(X'_2 \cap U)\big)\\
				\big(\langle 2^{X_1 \setminus X'_1},\subseteq \rangle,((X_1 \setminus X'_1) \cap U)\big) & \equiv^{\rm wf}_k &
				\big(\langle 2^{X_2 \setminus X'_2},\subseteq \rangle,((X_2 \setminus X'_2) \cap U)\big).
			\end{eqnarray*}
			A simple composition argument, similar to the one presented in Lemma \ref{lemma:comp}, allows to obtain the following from these two expressions: 
			\begin{equation} 
				\label{eq:final} 
				\big(\langle 2^{X_1},\subseteq \rangle,(X_1 \cap U),X'_1\big) \ \equiv_k^{{\rm wf}} \ \big(\langle 2^{X_2},\subseteq \rangle,
				(X_2 \cap U),X'_2\big).
			\end{equation} 
			In particular, this holds because $X'_1 = (X_1 \cap U)$ iff $X'_2 = (X_2 \cap U)$. In fact, assume without loss of generality that 
			$X'_1 = (X_1 \cap U)$. In particular, $X'_1 \cap \overline U = \emptyset$, but then $X'_2 \cap \overline U = \emptyset$ since $X'_1 \leftrightarrows_{3^k} X'_2$. 
			It follows that $X'_2 = (X_2 \cap U)$. 
			But Equation \eqref{eq:final} is equivalent with the following fact: 
			$$\big(\langle 2^{X_1},\subseteq \rangle,(X_1 \cap U)\big) \ \equiv_{k+1}^{{\rm wf}} \ \big(\langle 2^{X_2},\subseteq \rangle,
			(X_2 \cap U) \big).$$
			This finishes the proof of the lemma. 
		\end{proof} 
		
		We now proceed with the proof of Theorem \ref{thm:ne-foil}. Assume, for the sake of contradiction, that there is in fact a formula 
		$\msr(x,y)$ in $\foil$ such that, for every decision tree $\M$, instance $\es$, and partial instance $\es'$, we have that 
		$\astruct_\M \models \msr(\es,\es')$ iff $\es'$ is a minimum sufficient reason for $\es$ over $\M$. Let $k \geq 0$ be the quantifier rank of this formula. 
		We show that there exist decision trees $\M_1$ and $\M_2$, instances $\es_1$ and $\es_2$ over $\M_1$ and $\M_2$, respectively, 
		and partial instances $\es'_1$ and $\es'_2$ over $\M_1$ and $\M_2$, respectively, for which the following holds: 
		\begin{itemize} 
			\item $(\M_1,\es_1,\es'_1) \equiv_k (\M_2,\es_2,\es'_2)$, and hence 
			\begin{align*}
				\M_1 \models \msr(\es_1,\es'_1) \ \Leftrightarrow \
				\M_2 \models \msr(\es_2,\es'_2).	
			\end{align*}
			\item It is the case that $\es'_1$ is a minimum sufficient reason for $\es_1$ under $\M_1$, but 
			$\es'_2$ is not a minimum sufficient reason for $\es_2$ under $\M_2$. 
		\end{itemize}
		This is our desired contradiction.
		
		Let $\M_{n,p}$ be a decision tree of dimension $n+p$ such that, for every instance $\es \in \{0,1\}^{n+p}$, we have that $\M_{n,p}(\es) = 1$ iff  
		$\es$ is of the form $\{1\}^n \times \{0,1\}^p$, i.e., the first $n$ features of $\es$ are set to 1, or $\es$ is 
		of the form $\{0,1\}^n \times \{1\}^p$, i.e., the last $p$ features of $\es$ are set to 1. 
		Take the instance $\es = \{1\}^{n+p}$. It is easy to see that $\es$ only has two minimal sufficient reasons in $\M_{n,p}$; namely, $\es_1 = \{1\}^n \times \{\bot\}^p$ 
		and $\es_2 = \{\bot\}^n \times \{1\}^p$.  
		
		We define the following: 
		\begin{itemize} 
			\item $\M_1 := \M_{2^{k},2^{k}}$ and $\M_2 := \M_{2^k,2^k+1}$. 
			\item $\es_1 := \{1\}^{2^k + 2^k}$ and $\es_2 := \{1\}^{2^k+2^k+1}$. 
			\item $\es'_1 :=  \{\bot\}^{2^k} \times \{1\}^{2^k}$ and $\es'_2 :=  \{\bot\}^{2^k} \times \{1\}^{2^k+1}$. 
		\end{itemize}
		From our previous observation, $\es'_1$ is a minimal sufficient reason for $\es_1$ over $\M_1$
		and $\es'_2$ is a minimal sufficient reason for $\es_2$ over $\M_2$. 
		
		We show first that $(\M_1,\es_1,\es'_1) \equiv_k (\M_2,\es_2,\es'_2)$.  It can be observed that $\astruct_{\M_1}$ is of the form $\astruct_{N} \oplus \astruct_{N_1}$, where $N$ is a model of dimension $2^k$ that only accepts the tuple $\{1\}^{2^k}$ and the same holds for $N_1$. Analogously, $\astruct_{\M_2}$ is of the form $\astruct_{N} \oplus \astruct_{N_2}$, where $N_2$ is a model of dimension $2^k + 1$ that only accepts the tuple $\{1\}^{2^k + 1}$. 
		From Lemma \ref{lemma:pointed}, we have that $\astruct_{N_1} \equiv_k \astruct_{N_2}$. Notice that any winning strategy for the Duplicator on this game must map 
		the tuples $\{1\}^{2^k}$ in $\astruct_{N_1}$ and $\{1\}^{2^k+1}$ into each other. Therefore, 
		$$(\astruct_{N_1},\{1\}^{2^k}) \ \equiv_k \ 
		(\astruct_{N_2},\{1\}^{2^k+1}).$$
		But then, from Lemma \ref{lemma:comp}, we obtain that
		\begin{align*}
			\big(\astruct_N \oplus \astruct_{N_1},\{1\}^{2^k+2^k},\{\bot\}^{2^k} \times \{1\}^{2^k}\big) \ \equiv_k \
			\big(\astruct_N \oplus \astruct_{N_2},\{1\}^{2^k+2^k+1},\{\bot\}^{2^k} \times \{1\}^{2^k+1}\big).
		\end{align*}
		We can then conclude that $(\M_1,\es_1,\es'_1) \equiv_k (\M_2,\es_2,\es'_2)$, as desired. 
		
		Notice now that $\es'_1$ is a minimum sufficient reason for $\es_1$ over $\M_1$. In fact, by our previous observations, the only other minimal sufficient reason 
		for $\es_1$ over $\M_1$ is $\es''_1 = \{\bot\}^{2^k} \times \{1\}^{2^k}$, which has the same number of undefined features as $\es'_1$. 
		In turn, $\es'_2$ is not a minimum sufficient reason for $\es_2$ over $\M_2$. This is because $\es''_2 = \{\bot\}^{2^k+1} \times \{1\}^{2^k}$ is also a SR for $\es_2$ over $\M_2$, and $\es''_2$ has more undefined features than $\es'_2$. 

%
	
	\subsection{Proof of Theorem \ref{thm:eval-folistar}}
	\label{app:eval-folistar}
	
		For the first item, consider a fixed $\foil$ formula $\phi(x_1,\dots,x_m)$. We assume without loss of generality that $\phi$ is in prenex normal form, i.e., 
		it is of the form $$\exists \bar y_1 \forall \bar y_2 \cdots Q_k \bar y_k \, \psi(x_1,\dots,x_m,\bar y_1,\dots,\bar y_k), \quad \quad (k \geq 0)$$
		where $Q_k = \exists$ if $k$ is odd and $Q = \forall$ otherwise, and $\psi$ is a quantifier-free formula. 
		An FOIL formula of this form is called a $\Sigma_k$-$\foil$ {\em formula}. 
		Consider that $\M$ is a decision tree of dimension $n$, 
		and assume that we want to check whether $\M \models \phi(\es_1,\dots,\es_m)$, for $\es_1,\dots,\es_m$ given partial instances of dimension $n$. Since the formula $\phi$ is fixed, and thus the length of each tuple $\bar y_i$, for $i \leq k$, is constant, we can decide this problem in polynomial time by 
		using a $\Sigma_k$-alternating Turing machine (as the fixed size quantifier-free formula $\psi$ can be evaluated in polynomial time over $\M$). 
		
		\medskip

		We now deal with the second item. 
		We start by studying the complexity of the well-known {\em quantified Boolean formula} (QBF) problem for the case when the underlying formula (or, 
		more precisely, the underlying Boolean function) is defined by a decision tree. More precisely, suppose that $\M$ is a  
		decision tree 
		over 
		instances of dimension $n$. A $\Sigma_k$-QBF over $\M$, for $k > 1$, is an expression
		$$\exists P_1 \forall P_2 \cdots Q_k P_k \, \M,$$
		where $Q_k = \exists$ if $k$ is odd and $Q = \forall$ otherwise, and $P_1,\dots,P_k$ is a partition of $\{1,\dots,n\}$ into $k$ equivalence classes.
		As an example, if $\M$ is of dimension 3 then $\exists \{2,1\} \forall \{3\} \, \M$ is a $\Sigma_2$-QBF over $\M$. The semantics of these expressions is standard. 
		For instance, $\exists \{1,2\} \forall \{3\} \, \M$ holds if there exists a partial instance $(b_1,b_2,\bot) \in \{0,1\} \times \{0,1\} \times \{\bot\}$ such that both 
		$\M(b_1,b_2,0) = 1$ and $\M(b_1,b_2,1) = 1$. 
		
		For a fixed $k > 1$, we introduce then the problem $\Sigma_k$-\textsc{QBF}$(\dt)$. It takes as input a $\Sigma_k$-QBF $\alpha$ over $\M$, for $\M$ a decision tree, and asks whether $\alpha$ holds. 
		We establish the following result, which we believe of independent interest, as (to the best of our knowledge) the complexity of the 
		QBF problem over decision trees has not been studied in the literature. 
		
		\begin{lemma} \label{lemma:qbf} 
			For every even $k > 1$, the problem $\Sigma_k$-\textsc{QBF}$(\dt)$ is $\Sigma_k^{\text{P}}$-complete. 
		\end{lemma}   
		
		\begin{proof} 
			The upper bound is clear. For the lower bound we use a reduction from the following standard $\Sigma_k^{\text{P}}$-hard problem: 
			Given a 3CNF formula $\varphi$ over set $X = \{x_1,\dots,x_m\}$ of propositional variables, is it the case that the expression $\psi = \exists X_1 \forall X_2 \cdots \forall X_k \, \varphi$ holds, where 
			$X_1,\dots,X_k$ is a partition of $X$ in $k$ equivalence classes? 
			From $\psi$ we build in polynomial time a $\Sigma_k$-QBF $\alpha$ over $\M_\varphi$, where $\M_\varphi$ is a decision tree 
			that can be 
			built in polynomial time from $\varphi$, such that 
			\begin{equation} \label{eq:holds}  
				\text{$\psi$ holds} \quad \Longleftrightarrow \quad \text{$\alpha$ holds.}
			\end{equation}  
			
			We now explain how to define $\M_\varphi$ from the CNF formula $\varphi$. Let
			$\varphi = C_1 \wedge \cdots \wedge C_n$ be a propositional formula,
			where each $C_i$ is a disjunction of three literals and does
			not contain repeated or complementary literals. Moreover, assume that
			$\{x_1, \ldots, x_m\}$ is the set of variables occurring in $\varphi$, and the proof will use partial instances of dimension $n+m$. 
			Notice that the last $m$ features of such a partial instance $\es$ naturally define a truth assignment for the propositional formula $\varphi$. More precisely, for every $i \in \{1, \ldots, n\}$, we use notation $\es(C_i) = 1$ to indicate that there is a disjunct $\ell$ of $C_i$ such that $\ell = x_j$ and $\es[n+j] = 1$, or $\ell = \neg x_j$ and $\es[n+j] = 0$, for
			some $j \in \{1,\ldots,m\}$. Furthermore, we write $\es(\varphi) = 1$ if $\es(C_i) = 1$ for every $i \in \{1, \ldots, n\}$.
			
			For each clause $C_i$ ($i \in \{1, \ldots, n\}$), let  $\M_{C_i}$ be
			a decision tree of dimension $n+m$ (but that will only use features $n+1, \ldots, n+m$) such that for every entity $\es$:
			$\M_{C_i}(\es) = 1$ if and only $\es(C_i) = 1$. 
			Notice that  $\M_{C_i}$
			can be constructed in constant time as it only needs to contain at most
			eight paths of depth 3. For example, assuming that $C = (x_1 \vee x_2 \vee x_3)$,
			a possible decision tree $\M_C$ is depicted in the following figure:
			\begin{center}
				\begin{tikzpicture}[every node/.style={font=\footnotesize, scale=.8}]
					\node[circle,draw=black] (n) {$n+1$};
					\node[circle,draw=black,below left = 5mm and 28mm of n] (n0) {$n+2$};
					\node[circle,draw=black,below right = 5mm and 28mm of n] (n1) {$n+2$};
					\node[circle,draw=black,below left = 5mm and 10mm of n0] (n00) {$n+3$};
					\node[circle,draw=black,below right = 5mm and 10mm of n0] (n01) {$n+3$};
					\node[circle,draw=black,below left = 5mm and 10mm of n1] (n10) {$n+3$};
					\node[circle,draw=black,below right = 5mm and 10mm of n1] (n11) {$n+3$};
					\node[circle,draw=black,below left = 5mm and 5mm of n00] (n000) {$\false$};
					\node[circle,draw=black,below right = 5mm and 5mm of n00] (n001) {$\true$};
					\node[circle,draw=black,below left = 5mm and 5mm of n01] (n010) {$\true$};
					\node[circle,draw=black,below right = 5mm and 5mm of n01] (n011) {$\true$};
					\node[circle,draw=black,below left = 5mm and 5mm of n10] (n100) {$\true$};
					\node[circle,draw=black,below right = 5mm and 5mm of n10] (n101) {$\true$};
					\node[circle,draw=black,below left = 5mm and 5mm of n11] (n110) {$\true$};
					\node[circle,draw=black,below right = 5mm and 5mm of n11] (n111) {$\true$};
					
					\path[arrout] (n) edge node[above] {$0$} (n0);
					\path[arrout] (n) edge node[above] {$1$} (n1);
					\path[arrout] (n0) edge node[above] {$0$} (n00);
					\path[arrout] (n0) edge node[above] {$1$} (n01);
					\path[arrout] (n1) edge node[above] {$0$} (n10);
					\path[arrout] (n1) edge node[above] {$1$} (n11);
					\path[arrout] (n00) edge node[above] {$0$} (n000);
					\path[arrout] (n00) edge node[above] {$1$} (n001);
					\path[arrout] (n01) edge node[above] {$0$} (n010);
					\path[arrout] (n01) edge node[above] {$1$} (n011);
					\path[arrout] (n10) edge node[above] {$0$} (n100);
					\path[arrout] (n10) edge node[above] {$1$} (n101);
					\path[arrout] (n11) edge node[above] {$0$} (n110);
					\path[arrout] (n11) edge node[above] {$1$} (n111);
				\end{tikzpicture}
			\end{center} 
			
			Moreover, define $\M_\varphi$ as the following decision tree. 
			\begin{center}
				\begin{tikzpicture} 
					\node[circle,draw=black] (c1) {$1$};
					\node[below left = 6mm and 6mm of c1] (tc1) {$\M_{C_1}$};
					\node[circle,draw=black,below right = 6mm and 6mm of c1] (c2) {$2$};
					\node[below left = 6mm and 6mm of c2] (tc2) {$\M_{C_2}$};
					\node[circle,draw=black,below right = 6mm and 6mm of c2] (c3) {$3$};
					\node[below left = 6mm and 6mm of c3] (tc3) {$\M_{C_3}$};
					\node[below right = 6mm and 6mm of c3] (d) {$\cdots$};
					\node[circle,draw=black,below right = 6mm and 6mm of d] (cn) {$n$};
					\node[below left = 6mm and 6mm of cn] (tcn) {$\M_{C_n}$};
					\node[circle,draw=black,below right = 6mm and 6mm of cn, minimum size=8mm] (o) {$\true$};
					
					\path[arrout] (c1) edge node[above] {$0$} (tc1);
					\path[arrout] (c1) edge node[above] {$1$} (c2);
					\path[arrout] (c2) edge node[above] {$0$} (tc2);
					\path[arrout] (c2) edge node[above] {$1$} (c3);
					\path[arrout] (c3) edge node[above] {$0$} (tc3);
					\path[arrout] (c3) edge node[above] {$1$} (d);
					\path[arrout] (d) edge node[above] {$1$} (cn);
					\path[arrout] (cn) edge node[above] {$0$} (tcn);
					\path[arrout] (cn) edge node[above] {$1$} (o);
				\end{tikzpicture}
			\end{center} 
			
			Recall that the set of features of $\M_\varphi$ is $[1,n+m]$. 
			The formula $\alpha$ is defined as 
			$$\exists P_1 \forall P_2 \cdots \forall (P_k \cup P) \, \M_\varphi,$$
			assuming that $P_i$, for $1 \leq i \leq k$, is the set $\{n+\ell \mid x_\ell \in X_i\}$, and $P = \{1,\dots,n\}$.  
			That is, $P_i$ is the set of features from $\M_\varphi$ that represent the variables in $X_i$ and $P$ is the set of features that are used to encode 
			the clauses of $\varphi$. 
			
			We show next that the equivalence stated in \eqref{eq:holds} holds. For simplicity, we only do it for the case $k = 2$. The proof for $k > 2$ uses exactly the same ideas, only that it is slightly more cumbersome. 
			
			Assume, on the one hand, that $\psi = \exists X_1 \forall X_2 \varphi$ holds. That is, there exists an assignment $\sigma_1 : X_1 \to \{0,1\}$ 
			such that for every assignment $\sigma_2 : X_2 \to \{0,1\}$, it is the case that $\varphi$ holds when variables in $X_1$ are interpreted according to $\sigma_1$ and variables in $X_2$ according to $\sigma_2$. We show next 
			that $\alpha = \exists P_1 \forall (P_2 \cup P) \M_\varphi$ holds, where $P_1$, $P_2$, and $P$ are defined as above. 
			Take the partial instance $\es_{\sigma_1}$ of dimension $n + m$ that naturally ``represents'' the assignment $\sigma_1$; that is: 
			\begin{itemize} 
				\item $\es_{\sigma_1}[i] = \bot$, for each 
				$i \in \{1,\dots,n\}$, 
				\item $\es_{\sigma_1}[n+i] = \sigma_1(x_i)$, for each $i \in \{1,\dots,m\}$ with $x_i \in X_1$, and 
				\item $\es_{\sigma_1}[n+i] = \bot$, 
				for each $i \in \{1,\dots,m\}$ with $x_i \not\in X_1$. 
			\end{itemize} 
			To show that $\alpha$ holds, it suffices to show that $\M_\varphi(\es) = 1$ 
			for every instance $\es$ of dimension $n+m$ that subsumes $\es_{\sigma_1}$. Take an arbitrary such an instance $\es \in \{0,1\}^{n+m}$.
			Notice that if $\es[i] = 1$, for every $i \in \{1,\dots,n\}$, then $\M_\varphi(\es) = 1$ by definition of $\M_\varphi$. Suppose then that there exists a minimum
			value $i \in \{1,\dots,n\}$ such that $\es[i] = 0$. Hence, to show that $\M_\varphi(\es) = 1$ we need to show that $\M_{C_i}(\es) = 1$.   
			But this follows easily from the fact that $\es$ naturally represents an assignment $\sigma$ for $\varphi$ such that the restriction 
			of $\sigma$ to $X_1$ is precisely $\sigma_1$. We know that any such an assignment $\sigma$ satisfies $\varphi$, and therefore it satisfies $C_i$. It follows 
			that $\M_{C_i}(\es) = 1$.  
			
			Assume, on the other hand, that $\alpha = \exists P_1 \forall (P_2 \cup P) \M_\varphi$ holds. Then there exists a partial instance $\es$ such that the following statements hold: 
			\begin{itemize}
				\item 
				$\es[i] \neq \bot$ iff for some $j \in \{1,\dots,m\}$ it is the case that $i = n + j$ and $j \in P_1$, and
				\item for every completion $\es'$ of $\es$ we have that $\M_\varphi(\es) = 1$. 
			\end{itemize} 
			We show next that $\psi = \exists X_1 \forall X_2 \varphi$ holds. 
			Let $\sigma_1 : X_1 \to \{0,1\}$ be the assignment for the variables in $X_1$ that is naturally defined by $\es$. 
			Take an arbitrary assignment $\sigma_2 : X_2 \to \{0,1\}$. It suffices to show that each clause $C_i$ of $\varphi$, for $i \in \{1,\dots,n\}$, 
			is satisfied by the 
			assignment that interprets the variables in $X_1$ according to $\sigma_1$ and the variables in $\sigma_2$ according to $X_2$. 
			Let us define a completion $\es'$ of $\es$ that satisfies the following: 
			\begin{itemize} 
				\item $\es'[i] = 0$, 
				\item $\es'[j] = 1$, for each $j \in \{1,\dots,n\}$ with $i \neq j$,  
				\item $\es'[n+j] = \sigma_1(x_j)$, if $j \in \{1,\dots,m\}$ and $x_j \in P_1$, and
				\item $\es'[n+j] = \sigma_2(x_j)$, if $j \in \{1,\dots,m\}$ and $x_j \in P_2$. 
			\end{itemize}   
			We know that $\M_\varphi(\es') = 1$, which implies that $\M_{C_i}(\es') = 1$ (since $\es'$ takes value 0 for feature $i$). 
			We conclude that $C_i$ is satisfied by the assignment which is naturally defined by $\es'$, which is precisely the one that interprets the variables in $X_1$ according to 
			$\sigma_1$ and the variables in $\sigma_2$ according to $X_2$. 
		\end{proof}

		With the help of Lemma \ref{lemma:qbf} we can now finish the proof of the theorem.
		We reduce from $\Sigma_k$-\textsc{QBF}$(\dt)$.
		The input to $\Sigma_k$-\textsc{QBF}$(\dt)$ is given by an expression $\alpha$ of the form 
		$$\exists P_1 \forall P_2 \cdots \forall_k P_k \, \M,$$
		for $\M$ a decision tree of dimension $n$ and $P_1,\dots,P_k$ a partition of $\{1,\dots,n\}$. 
		We explain next how the formula $\phi_k(x_1,\dots,x_k)$ is defined. 
		
		We start by defining some auxiliary terminology. 
		We use $x[i]$ to denote the $i$-th feature of 
		the partial instance that is assigned to variable $x$. We define the following formulas. 
		\begin{itemize} 
			\item ${\sf Undef}(x) := \neg \exists y (y \subset x)$. That is, ${\sf Undef}$ defines the set that only consists of the partial instance $\{\bot\}^{n}$ in which all components are undefined. 
			
			\item ${\sf Single}(x) := \exists y (y \subset x) \wedge \forall y (y \subset x \, \rightarrow \, {\sf Undef}(y))$. 
			That is, ${\sf Single}$ defines the set that consists precisely of those partial instances in 
			$\{0,1,\bot\}^{n}$ which have at most one component that is defined. 
			
			\item $(x \sqcup y = z) := (x \subseteq z) \wedge (y \subseteq z) \wedge \neg \exists w \big((x \subseteq w) \wedge (y \subseteq w) \wedge (w \subset z)\big)$. That is, $z$, if it exists, is the {\em join} of $x$ and $y$. In other words, $z$ is defined if every feature that is defined over $x$ and $y$ takes the same value in both partial instances, 
			and, in such case, for each $1 \leq i \leq n$ we have that $z[i] = x[i] \sqcup y[i]$, where $\sqcup$ is the commutative and idempotent 
			binary operation that satisfies 
			$\bot \sqcup 0 = 0$ and $\bot \sqcup 1 = 1$.  
			
			As an example, $(1,0,\bot,\bot) \sqcup (1,\bot,\bot,1) = (1,0,\bot,1)$, while $(1,\bot) \sqcup (0,0)$ is undefined. 
			
			\item $(x \sqcap y = z) := (z \subseteq x) \wedge (z \subseteq y) \wedge \neg \exists w \big((w \subseteq x) \wedge (w \subseteq y) \wedge (z \subset w)\big)$. That is, $z$
			is the {\em meet} of $x$ and $y$ (which always exists). In other words, for each $1 \leq i \leq n$ we have that $z[i] = x[i] \sqcap y[i]$,  
			where $\sqcap$ is the commutative and idempotent 
			binary operation that satisfies
			$\bot \sqcap 0 = \bot \sqcap 1 = 0 \sqcap 1 = \bot$.  
			
			As an example, $(1,0,\bot,\bot) \sqcap (1,\bot,\bot,1) = (1,\bot,\bot,\bot)$, while $(1,\bot) \sqcap (0,0) = (\bot,\bot)$. 
			
			\item ${\sf Comp}(x,y) := \exists w \exists z ({\sf Undef}(z) \, \wedge \, x \sqcup y = w \, \wedge \, x \sqcap y = z)$.   
			That is, ${\sf Comp}$ defines the pairs $(\es_1,\es_2)$ 
			of partial instances in $\{0,1,\bot\}^n \times \{0,1,\bot\}^n$ such that no feature that is defined in $\es_1$ is also defined in $\es_2$, and vice versa. In fact, assume 
			for the sake of contradiction that this is not the case. By symmetry, we only have to consider the following two cases.
			\begin{itemize}
				\item There is an $i \leq n$ with $\es_1[i] = 1$ and $\es_2[i] = 0$. Then the join of $\es_1$ and $\es_2$ does not exist. 
				\item There is an $i \leq n$ with $\es_1[i]  = \es_2[i] = 1$. Then the $i$-th component of the meet of $\es_1$ and $\es_2$ takes value 1, 
				and hence $\es_1 \sqcap \es_2 \neq \{\bot\}^n$. 
			\end{itemize} 
			\item ${\sf MaxComp}(x,y) := {\sf Comp}(x,y) \, \wedge \, \neg \exists z \big((y \subset z) \wedge {\sf Comp}(x,z)\big)$.    
			That is, ${\sf MaxComp}$ defines the pairs $(\es_1,\es_2)$ such that the components that are defined in $\es_1$ are precisely the ones that are undefined in $\es_2$, and vice versa.  
			\item ${\sf Rel}(x,y) := \neg \exists z \big((z \subseteq y)  \, \wedge \, {\sf Single}(z) \, \wedge\, {\sf Comp}(x,z)\big)$. 
			That is, ${\sf Rel}$ defines the pairs $(\es_1,\es_2)$ 
			of partial instances in $\{0,1,\bot\}^n \times \{0,1,\bot\}^n$ such that every feature that is defined in $\es_1$ is also defined in $\es_2$.
			\item ${\sf MaxRel}(x,y) := {\sf Rel}(x,y) \, \wedge \, \neg \exists z \big((y \subset z) \wedge {\sf Rel}(x,z)\big)$.    
			That is, ${\sf MaxRel}$ defines the pairs $(\es_1,\es_2)$ such that the features defined in $\es_1$ and in $\es_2$ are the same.  
			%
		\end{itemize} 
		
		We now define the formula $\phi_k(x_1,\dots,x_k)$ as
		\begin{align*}
		\exists y_1 \forall y_2 \cdots \forall y_k \, \bigg( \bigwedge_{i=1}^k {\sf MaxRel}(x_i,y_i) \ \wedge\
		\forall z  \big(z = y_1 \sqcup y_2 \sqcup \cdots \sqcup y_k \, \rightarrow \, \pos(z)\big)\bigg). 
		\end{align*}
		
		For each $i$ with $1 \leq i \leq k$, let $\es_i$ be the partial instance of dimension $n$ such that 
		$$\es_i[j] \ = \ \begin{cases}
			1 \quad \quad & \text{if $j \in P_i$,} \\
			\bot & \text{otherwise.}
		\end{cases}$$   
		That is, $\es_i$ takes value 1 over the features in $P_i$ and it is undefined over all other features. 
		We claim that $\alpha$ holds if, and only if, $\M \models \phi_k(\es_1,\dots,\es_k)$. 
		The result then follows since $\M$ is a decision tree. 
		
		For the sake of presentation we only prove the aforementioned equivalence 
		for the case when $k = 2$, since the extension to $k > 2$ is standard (but cumbersome). That is, we consider the case when $\alpha = \exists P_1 \forall P_2 \M$ and, therefore, 
		\begin{align*}
			\phi_2(x_1,x_2) = \exists y_1 \forall y_2 \, \bigg( {\sf MaxRel}(x_1,y_1) \ \wedge \
			{\sf MaxRel}(x_2,y_2) \wedge \forall z  \big(z = y_1 \sqcup y_2  \, \rightarrow \, \pos(z)\big)\bigg). 
		\end{align*}
		
		\begin{itemize}
			\item[$(\Leftarrow)$] Assume first that $\M \models \phi_2(\es_1,\es_2)$. Hence, there exists a partial instance $\es'_1$ such that 
			\begin{align*}
			\M \, \models \, \forall y_2 \, \bigg( {\sf MaxRel}(\es_1,\es'_1) \ \wedge \
			{\sf MaxRel}(\es_2,y_2) \, \wedge \, \forall z  \big(z = \es'_1 \sqcup y_2  \, \rightarrow \, \pos(z)\big)\bigg). 
			\end{align*} 
			This means that the features defined in $\es_1$ and $\es'_1$ are exactly the same, and hence $\es'_1$ is a partial instance that is defined precisely over the features 
			in $P_1$. We claim that every instance $\es$ that is a completion of $\es'_1$ satisfies $\M(\es) = 1$, thus showing that $\alpha$ holds. In fact, take an arbitrary such a completion $\es$. By definition, $\es$ can be written as $\es'_1 \sqcup \es'_2$, where $\es'_2$ is a partial instance that is defined precisely over those features not in $P_1$, i.e., 
			over the features in $P_2$. Thus, 
			$$\M \, \models \, {\sf MaxRel}(\es_1,\es'_1) \, \wedge \, {\sf MaxRel}(\es_2,\es'_2) 
			\, \wedge \, \es = \es'_1 \sqcup \es'_2,  
			$$ 
			which allows us to conclude that $\M \models \pos(\es)$. This tells us that $\M(\es) = 1$. 
			
			\item[$(\Rightarrow)$] Assume in turn that $\alpha$ holds, and hence that there is a partial instance $\es'_1$ that is defined precisely over the features 
			in $P_1$ such that that every instance $\es$ that is a completion of $\es'_1$ satisfies $\M(\es) = 1$. We claim that 
			\begin{align*}
			\M \, \models \, \forall y_2 \, \bigg( {\sf MaxRel}(\es_1,\es'_1) \, \wedge \, {\sf MaxRel}(\es_2,y_2) \ \wedge \
			\forall z  \big(z = \es'_1 \sqcup y_2  \, \rightarrow \, \pos(z)\big)\bigg),  
			\end{align*}
			which implies that $\M \models \phi_2(\es_1,\es_2)$. In fact, let $\es'_2$ be an arbitrary instance such that ${\sf MaxRel}(\es_2,\es'_2)$ holds. By definition, $\es'_2$ is defined precisely over the features in $P_2$. Let $\es = \es'_1 \sqcup \es'_2$. Notice that $\es$ is well-defined since the sets of features defined in $\es'_1$ and $\es'_2$, respectively, are disjoint. Moreover, $\es$ is a completion of $\es'_1$ as $P_1 \cup P_2 = \{1,\dots,n\}$. We then have that $\M(\es) = 1$ as $\alpha$ holds. This allows us to conclude that 
			$\M \models \pos(\es)$, and hence that $\M \models \phi_2(\es_1,\es_2)$. 
		\end{itemize} 
		This concludes the proof of the theorem. 

\subsection{$\subseteq$ and $\lel$ cannot be defined in terms of each other}
\label{app:subseteq-lel}

First, we show that predicate $\lel$ cannot be defined in terms of predicate $\subseteq$.
    \begin{lemma}
    There is no formula $\varphi(x,y)$ in $\foil$ defined over the vocabulary
    $\{\subseteq\}$ such that, for every decision tree $\T$ and pair
    of partial instances $\es$, $\es'$, we have that
    $$\T \models \varphi(\es,\es') \ \Longleftrightarrow \ |\es_\bot| \geq |\es'_\bot|.$$
    \end{lemma}
    \begin{proof}
      For the sake of contradiction, assume that $\varphi(x,y)$ is
      definable in $\foil$ over the vocabulary $\{\subseteq\}$. Then
      the following are formulas in $\foil$:
      \begin{eqnarray*}
      	\sr(x,y) &=& \full(x) \wedge \, y \subseteq x \wedge \forall z \, \big(y \subseteq z \wedge \full(z) \, \rightarrow \, (\pos(z) \leftrightarrow \pos(\ x))\big),\\
 	  	\msr(x,y) &=& \sr(x,y) \wedge \forall z \, \big(\sr(x,z) \to (\varphi(z,y) \to \varphi(y,z))\big).
   	  \end{eqnarray*}
        But the second formula verifies if a partial instance $x$ is a minimum SR for a given
      instance $y$, which contradicts the inexpressibility result of
      Theorem \ref{thm:ne-foil}, and hence concludes the proof of the lemma.
    \end{proof}
Second, we show that predicate $\subseteq$ cannot be defined in terms of predicate $\lel$.
    \begin{lemma}
    There is no formula $\psi(x,y)$ in $\foil$ defined over the vocabulary
    $\{\lel\}$ such that, for every decision tree $\T$ and pair
    of partial instances $\es$, $\es'$, we have that
    $$\T \models \psi(\es,\es') \ \Longleftrightarrow \ \es \text{ is subsumed by } \es'.$$
    \end{lemma}

    \begin{proof}
For the sake of contradiction, assume that $\psi(x,y)$ is
definable in $\foil$ over the vocabulary $\{\lel\}$, and let $n \geq
3$. Moreover, for every $k \in \{0, \ldots, n\}$, define $L_k$ as the
following set of partial instances:
\begin{eqnarray*}
  L_k &=& \{ \es \in \{0,1,\bot\}^n \mid |\es_\bot| = k\},
\end{eqnarray*}
and let $f_k : L_k \to L_k$ be an arbitrary bijection from $L_k$ to
itself. Finally, let $f : \{0, 1, \bot\}^n \to \{0, 1, \bot\}^n$ be
defined as $f(\es) = f_i(\es)$ if $|\es| = i$. Clearly, $f$ is a
bijection from $\{0, 1, \bot\}^n$ to $\{0, 1, \bot\}^n$.

For a decision tree $\T$ of dimension $n$, define $\mathfrak{A}'_\T$
as the restriction of $\mathfrak{A}_\T$ to the vocabulary
$\{\preceq\}$. Then function $f$ is an automorphism of
$\mathfrak{A}'_\T$ since $f$ is a bijection from $\{0, 1, \bot\}^n$ to
$\{0, 1, \bot\}^n$, and for every pair of partial instances $\es_1,
\es_2$:
\begin{align*}
	\mathfrak{A}'_\T \models \es_1 \lel \es_2 \quad \text{ if and only if } \quad
	\mathfrak{A}'_\T \models f(\es_1) \lel f(\es_2).
\end{align*}
Then given that $\psi(x,y)$ is definable in first-order logic over the
vocabulary $\{\lel\}$, we have that for every pair of partial
instances $\es_1, \es_2$:
\begin{align}\label{eq-subsumed-aut}
	\mathfrak{A}'_\T \models \psi(\es_1, \es_2) \quad \text{ if and only if } \quad
	\mathfrak{A}'_\T \models \psi(f(\es_1), f(\es_2)).
\end{align}
But now assume that $g_k : L_k \to L_k$ is defined as the identity
function for every $k \in \{0, \ldots, n\} \setminus \{1\}$, and
assume that $g_1$ is defined as follows for every partial instance
$\es$:
\begin{eqnarray*}
  g_1(\es) &=&
  \begin{cases}
    (\bot, 0, \ldots, 0) & \text{if } \es = (0, \ldots, 0, \bot)\\
    (0, \ldots, 0, \bot) & \text{if } \es = (\bot, 0, \ldots, 0)\\
    \es & \text{otherwise}
  \end{cases}
\end{eqnarray*}
Clearly, each function $g_i$ is a bijection. Moreover, let $g : \{0,
1, \bot\}^n \to \{0, 1, \bot\}^n$ be defined as $g(\es) = g_i(\es)$ if
$|\es| = i$. Then we have by \eqref{eq-subsumed-aut} that for every
pair of partial instances $\es_1$, $\es_2$:
\begin{align*}
	\mathfrak{A}'_\T \models \psi(\es_1, \es_2) \quad \text{ if and only if } \quad
	\mathfrak{A}'_\T \models \psi(g(\es_1), g(\es_2)).
\end{align*}
Hence, if we consider $\es_1 = (\bot, \bot, 0, \ldots, 0)$ and $\es_2
= (\bot, 0, \ldots, 0)$, given that $g(\es_1) = (\bot, \bot, 0, \ldots, 0)$ and $g(\es_2)
= (0, \ldots, 0, \bot)$, we conclude that:
\begin{align*}
	\begin{gathered}
  		\mathfrak{A}'_\T \models \psi((\bot, \bot, 0, \ldots, 0),\, (\bot, 0, \ldots, 0)) \\
  		\text{ if and only if }\\
  		\mathfrak{A}'_\T \models \psi((\bot, \bot, 0, \ldots, 0),\, (0, \ldots, 0, \bot)).
	\end{gathered}
\end{align*}
But this leads to a contradiction, since $(\bot, \bot, 0, \ldots, 0)$
is subsumed by $(\bot, 0, \ldots, 0)$, but $(\bot, \bot, 0, \ldots,
0)$ is not subsumed by $(0, \ldots, 0, \bot)$. This concludes the
proof of the lemma.
    \end{proof}

    \subsection{Proof of Theorem \ref{theo:ptime-atomic}}
    \label{app:ptime-atomic}

Given $n \geq 0$ and a decision tree $\T$ of dimension $n$, define
$\mathfrak{B}_\T$ as a structure over the vocabulary $\{\subseteq,
\lel\}$ generated from $\mathfrak{A}_\T$ by removing the
interpretation of predicate $\pos$, and adding the interpretation of
predicate $\lel$.
Notice that given two decision trees $\T_1$, $\T_2$ of dimension $n$,
we have that $\mathfrak{B}_{\T_1} =
\mathfrak{B}_{\T_2}$, so we define $\mathfrak{B}_n$ as
$\mathfrak{B}_\T$ for an arbitrary decision tree $\T$ of dimension $n$.

Fix a $\foil$ formula $\phi(x_1, \ldots, x_\ell)$ over the vocabulary
$\{\subseteq, \lel\}$. Then we have that
the input of \logicEvaluation$(\phi)$ is a decision tree $\T$ of
dimension $n$ and partial instances $\es_1$, $\ldots$, $\es_\ell$ of
dimension $n$, and the question to answer is whether
$\mathfrak{B}_\T \models \phi(\es_1, \ldots, \es_\ell)$, which is
equivalent to checking whether
$\mathfrak{B}_n \models \phi(\es_1, \ldots, \es_\ell)$. In what
follows, we provide a polynomial-time algorithm for
\logicEvaluation$(\phi)$, by using \EF\ games as in the proof of
Theorem~\ref{thm:ne-foil} (see Section~\ref{app:proof-1}).

To give an idea of how the polynomial-time algorithm for
\logicEvaluation$(\phi)$ works, assume first that $\phi$ is a sentence
($\phi$ does not have any free variables), and let $k$ be the
quantifier rank of $\phi$. Then as in the proof of
Theorem~\ref{thm:ne-foil}, it is possible to show that there exists a
fixed value $v_0$ such that:
\begin{lemma}\label{lem-sentence-ef}
  For every $n, m \geq v_0$, it holds that $\mathfrak{B}_n \equiv_k
  \mathfrak{B}_m$.
\end{lemma}
Lemma \ref{lem-sentence-ef} gives us a polynomial-time algorithm
for \logicEvaluation$(\phi)$, given that $\phi$ is a fixed formula and
$v_0$ is a fixed value. The algorithm receives as input a decision
tree $\T$. Assume that $\dim(\T) = n$. If $n < v_0$, then the algorithm
materializes $\mathfrak{B}_{n}$ and checks whether $\mathfrak{B}_{n}
\models \phi$. If $n \geq v_0$, then it materializes $\mathfrak{B}_{v_0}$
and checks whether $\mathfrak{B}_{v_0} \models \phi$. Notice that this
algorithm works in polynomial time as $v_0$ is a fixed value. Besides,
it is correct as if $n \geq v_0$, then it holds by Lemma
\ref{lem-sentence-ef} that $\mathfrak{B}_n \equiv_k \mathfrak{B}_{v_0}$,
and, therefore, $\mathfrak{B}_n \models \phi$ if and only if
$\mathfrak{B}_{v_0} \models \phi$ given that the quantifier rank of $\phi$
is $k$.

To show how the previous idea can be extended to a formula with free
variables, assume first that $\phi$ has free variables $x_1$ and
$x_2$. Then given $(a_1,a_2) \in \{0,1,\bot\}^2$ and partial instances
$\es_1$, $\es_2$ of dimension $n$, let $n_{(a_1,a_2)}(\es_1, \es_2)$ be defined as follows:
\begin{align*}
	n_{(a_1,a_2)}(\es_1,\es_2) = \big|\{ i \in \{1, \ldots, n\} \mid
	\es_1[i] = a_1 \text{ and } \es_2[i] = a_2 \}\big|.
\end{align*}
Moreover, let $k$ be the quantifier rank of $\phi$. Then as in the
proof of Theorem~\ref{thm:ne-foil}, it is possible to show that there
exists a fixed value $v_2$ such that:
\begin{lemma}\label{lem-binary-ef}
Let $n, m \geq 1$, $\es_1$, $\es_2$ be partial
instances of dimension $n$, and $\es'_1$, $\es'_2$ be partial
instances of dimension $m$. If $n,m \geq v_2$ and for every $(a_1,a_2) \in \{0,1,\bot\}^2$:
\begin{eqnarray*}
\min\{v_2, n_{(a_1,a_2)}(\es_1,\es_2)\} &= &
\min\{v_2, n_{(a_1,a_2)}(\es'_1,\es'_2)\}, 
\end{eqnarray*}
then it holds that $(\mathfrak{B}_n, \es_1, \es_2) \equiv_k
  (\mathfrak{B}_m, \es'_1, \es'_2)$.
\end{lemma}
Lemma \ref{lem-binary-ef} gives us a polynomial-time algorithm
for \logicEvaluation$(\phi)$, given that $\phi$ is a fixed formula,
the number of free variables of $\phi$ is fixed, and $v_2$ is a fixed
value. The algorithm receives as input a decision tree $\T$ such that
$\dim(\T) = n$, and partial instances $\es_1$, $\es_2$ of dimension
$n$. If $n < v_2$, then the algorithm materializes $\mathfrak{B}_{n}$
and checks whether $\mathfrak{B}_{n} \models \phi(\es_1,\es_2)$. If
$n \geq v_2$, then the algorithm considers the minimum value $m$ such
that $m \geq v_2$ and there exist partial instances $\es'_1$,
$\es'_2$ satisfying that
\begin{eqnarray*}
\min\{v_2, n_{(a_1,a_2)}(\es_1,\es_2)\} &= &
\min\{v_2, n_{(a_1,a_2)}(\es'_1,\es'_2)\},
\end{eqnarray*}
for every $(a_1,a_2) \in \{0,1,\bot\}^2$.
Then the algorithm checks whether $\mathfrak{B}_m \models \phi
(\es'_1, \es'_2)$.  As in the previous case, this algorithm works in
polynomial time since $v_2$ is a fixed value, and it is correct by
Lemma \ref{lem-binary-ef}.

Finally, consider a fixed $\foil$ formula $\phi(x_1, \ldots,
x_\ell)$. Then given $(a_1, \ldots, a_\ell) \in \{0,1,\bot\}^\ell$ and
partial instances $\es_1$, $\ldots$, $\es_\ell$ of dimension $n$, let
$n_{(a_1,\ldots,a_\ell)}(\es_1, \ldots, \es_\ell)$ be defined as
follows:
\begin{align*}
	n_{(a_1,\ldots,a_\ell)}(\es_1,\ldots,\es_\ell) = \big|\{ i \in \{1, \ldots, n\} \mid 
	\es_j[i] = a_j \text{ for every } j \in \{1, \ldots, \ell\}\}\big|.
\end{align*}
Moreover, let $k$ be the quantifier rank of $\phi$. Then as in the
proof of Theorem~\ref{thm:ne-foil}, it is possible to show that there
exists a fixed value $v_\ell$ such that:
\begin{lemma}\label{lem-gen-ef}
Let $n, m \geq 1$, $\es_1$, $\ldots$, $\es_\ell$ be partial instances
of dimension $n$, and $\es'_1$, $\ldots$, $\es'_\ell$ be partial
instances of dimension $m$. If $n,m \geq v_\ell$ and for every
$(a_1, \ldots, a_\ell) \in \{0,1,\bot\}^\ell$:
\begin{align*}
	\min\{v_\ell, n_{(a_1,\ldots,a_\ell)}(\es_1,\ldots,\es_\ell)\} =
	\min\{v_\ell, n_{(a_1,\ldots,a_\ell)}(\es'_1,\ldots,\es'_\ell)\}, 
\end{align*}
then it holds that $(\mathfrak{B}_n, \es_1, \ldots, \es_\ell) \equiv_k
  (\mathfrak{B}_m, \es'_1, \ldots, \es'_\ell)$.
\end{lemma}
As in the previous two cases, 
Lemma \ref{lem-gen-ef} gives us a polynomial-time algorithm
for \logicEvaluation$(\phi)$, given that $\phi$ is a fixed formula,
the number $\ell$ of free variables of $\phi$ is fixed, and $v_\ell$ is a fixed
value.

    \subsection{Proof of Theorem \ref{thm:eval-qdtfoil}}
    \label{app:eval-folistarpm}
For the first item of the theorem, we notice that for a fixed formula
$\phi$ in $\qdtfoil$, the membership of \logicEvaluation$(\phi)$ in
$\bh_k$ for some $k \geq 1$ follows directly from the fact that $\phi$
is defined as a (fixed) Boolean combination of formulas $\phi'$ for
which \logicEvaluation$(\phi')$ is in $\np$ or $\conp$. The membership of every problem in $\np$ and $\conp$, comes from the quantification of properties that can be checked in polynomial time.

Following the definition of $\qdtfoil$ in Section \ref{sec-qdtfoil}; let $\psi$ be a~\dtfoil~formula that represents a property that an explanation must satisfy, then \logicEvaluation$(\psi)$ is in $\ptime$ (see Proposition \ref{prop:ptime-guarded}). Moreover, Boolean combinations of~\dtfoil~formulas are still in $\ptime$ by checking every property independently. Finally, quantified properties than can be checked in polynomial time meet the definitions of $\np$ and $\conp$. By definition (see Section \ref{subsec:bhierarchy}), then exists $k \geq 1$ such that the evaluation of a fixed Boolean combination of such quantified properties is in $\bh_k$ (\logicEvaluation$(\phi) \in \bh_k$).

To prove the second item of the theorem, we show that for every $k \geq 1$, there exists a $\qdtfoil$ formula $\phi_k$ such that $\sat(k)$ (see Section \ref{subsec:bhierarchy}) can be reduced in polynomial time to \logicEvaluation$(\phi_k)$, thus concluding the second item of the
theorem.

Let $L = \{ (\T, \es, \es') \mid \T$ is a decision tree, $\es$, $\es'$
are partial instances and $\T \models \neg\msr(\es, \es')\}$. It
follows from \citep{NEURIPS2020_b1adda14} that $L$ is
$\np$-complete (see Section \ref{sec-lim-foil} for a definition of the
formula $\msr(x,y)$). Hence, for every propositional formula $\psi$, it is
possible to construct in polynomial time a decision tree $\T_\psi$ and
partial instances $\es_{\psi,1}$, $\es_{\psi,2}$ such that:
\begin{itemize}
  \item[(C1)] $\psi \in \sat$ if and only if $\T_\psi \models \neg \msr(\es_{\psi,1},
\es_{\psi,2})$ 
\end{itemize}
Let $k \geq 1$. In what follows, we use language $L$ to construct a
formula $\phi_k$ in $\qdtfoil$ such that, $\sat(k)$ can be reduced in
polynomial time to \logicEvaluation$(\phi_k)$. Let $(\psi_1, \ldots,
\psi_k)$ be a tuple of $k$ propositional formulas, and assume that
each decision tree $\T_{\psi_i}$ has dimension $n_i$ ($1 \leq i \leq
k$). Then a decision tree $\T$ of dimension $d = k + \sum_{\ell=1}^k
n_\ell$ is defined as follows:
    \begin{center}
                                \begin{tikzpicture}
                                        \node[circle,draw=black] (c1) {$1$};
                                        \node[below left = 6mm and 6mm of c1] (tc1) {$\T_{\psi_1}$};
                                        \node[circle,draw=black,below right = 6mm and 6mm of c1] (c2) {$2$};
			                \node[below left = 6mm and 6mm of c2] (tc2) {$\T_{\psi_2}$};
					\node[circle,draw=black,below right = 6mm and 6mm of c2] (c3) {$3$};
			                \node[below left = 6mm and 6mm of c3] (tc3) {$\T_{\psi_3}$};
                                        \node[below right = 6mm and 6mm of c3] (d) {$\cdots$};
			                \node[circle,draw=black,below right = 6mm and 6mm of d] (cn) {$k$};
		                        \node[below left = 6mm and 6mm of cn] (tcn) {$\T_{\psi_k}$};
			                \node[circle,draw=black,below right = 6mm and 6mm of cn, minimum size=8mm] (o) {$\true$};

                                        \path[arrout] (c1) edge node[above] {$1$} (tc1);
              	                        \path[arrout] (c1) edge node[above] {$0$} (c2);
		                        \path[arrout] (c2) edge node[above] {$1$} (tc2);
		                        \path[arrout] (c2) edge node[above] {$0$} (c3);
					\path[arrout] (c3) edge node[above] {$1$} (tc3);
                                        \path[arrout] (c3) edge node[above] {$0$} (d);
                                        \path[arrout] (d) edge node[above] {$0$} (cn);
                                        \path[arrout] (cn) edge node[above] {$1$} (tcn);
                                        \path[arrout] (cn) edge node[above] {$0$} (o);
                                \end{tikzpicture}
                        \end{center}
where each $\T_{\psi_i}$ ($1 \leq i \leq k$) mentions the following
features:
$$ \bigg\{k + \sum_{\ell=1}^{i-1} n_\ell + 1, \ldots, k + \sum_{\ell=1}^{i}
    n_\ell\bigg\}.$$ In this way, we have that $\T_{\psi_i}$ and
    $\T_{\psi_j}$ are defined over disjoint sets of features, for
    every $i,j \in \{1, \ldots, k\}$ with $i \neq j$. Moreover, define the following partial instances of dimension 
    $d$.
    \begin{itemize}
      \item For each $i \in \{1, \ldots, k\}$, the partial instance
        $\es_i$ is defined as: (1) $\es_i[i] = 1$; (2) $\es_i[j] = 0$
        for every $j \in \{1, \ldots, k\} \setminus \{i\}$; and (3)
        $\es_i[j] = \bot$ for every $j \in \{k + 1, \ldots, d\}$.

      \item
        For each $i \in \{1, \ldots, k\}$, the partial instances
        $\es_{i,1}$, $\es_{i,2}$, are defined as: (1) $\es_{i,1}[i] =
        \es_{i,2}[i] = 1$; (2) $\es_{i,1}[j] = \es_{i,1}[j] = 0$ for
        every $j \in \{1, \ldots, k\} \setminus \{i\}$; (3)
        $\es_{i,1}[j] = \es_{\psi_i,1}[j - (k + \sum_{\ell=1}^{i-1}
          n_\ell)]$ for every $j \in \{k + \sum_{\ell=1}^{i-1} n_\ell
        + 1, \ldots, k + \sum_{\ell=1}^{i} n_\ell\}$; (4)
        $\es_{i,2}[j] = \es_{\psi_i,2}[j - (k + \sum_{\ell=1}^{i-1}
          n_\ell)]$ for every $j \in \{k + \sum_{\ell=1}^{i-1} n_\ell
        + 1, \ldots, k + \sum_{\ell=1}^{i} n_\ell\}$; (5)
        $\es_{i,1}[j] = 0$ for every $j \in \{k+1, \ldots, d\}
        \setminus \{k + \sum_{\ell=1}^{i-1} n_\ell + 1, \ldots, k +
          \sum_{\ell=1}^{i} n_\ell\}$; and (6) $\es_{i,2}[j] = \bot$ for
          every $j \in \{k+1, \ldots, d\} \setminus \{k +
            \sum_{\ell=1}^{i-1} n_\ell + 1, \ldots, k +
            \sum_{\ell=1}^{i} n_\ell\}$.
    \end{itemize}
    Moreover, define the following $\qdtfoil$ formulas:
    \begin{eqnarray*}
    	\sr(x,y,w) &=& w \subseteq x \wedge w \subseteq y \wedge \full(x) \wedge y \subseteq x \wedge
    	\quad (\pos(x) \to \allpos(y)) \wedge (\npos(x) \to \allneg(y)) \\
    	\msr(x,y,w) &=& \sr(x,y,w) \wedge \forall z \, \big(w \subseteq z \to
    	(\sr(x,z,w) \to (\lel(z, y) \to \lel(y, z)))\big)
    \end{eqnarray*}
    Notice that the decision tree $\T$ and the partial instances
    $\es_1$, $\ldots$, $\es_k$, $\es_{1,1}$, $\es_{1,2}$, $\ldots$,
    $\es_{k,1}$, $\es_{k,2}$ can be constructed in polynomial time in
    the size of $(\psi_1, \ldots, \psi_k)$. Besides, from the
    definition of these elements, it is possible to conclude that:
    \begin{lemma}\label{lem:bh-foilpm}
    For every $i \in \{1, \ldots, k\}$: $\T_{\psi_i} \models\ $\msr$(\es_{\psi_i,1}, \es_{\psi_i, 2})$ if and only if
    $\T \models\ $\msr$(\es_{i,1}, \es_{i, 2}, \es_i)$.
    \end{lemma}
    Finally, let $\phi_k$ be a $\qdtfoil$ formula obtain by
    constructing the following sequences of formula $\alpha_i$ with $i
    \in \{1, \ldots, k\}$, and then defining $\phi_k(x_{1,1}, x_{1,2},
    x_1, \ldots, x_{k,1}, x_{k,2}, x_k) = \alpha_k(x_{1,1}, x_{1,2},
    x_1, \ldots, x_{k,1}, x_{k,2}, x_k)$:
\begin{eqnarray*}
	    \alpha_1 &=& \neg \msr(x_{1,1}, x_{1, 2}, x_1)\\
	    \alpha_{2 \ell} &=& (\alpha_{2 \ell - 1} \wedge \msr(x_{2\ell,1}, x_{2\ell, 2}, x_{2\ell})) \\
	    \alpha_{2 \ell + 1} &=& (\alpha_{2 \ell} \vee \neg \msr(x_{2\ell+1,1}, x_{2\ell+1, 2}, x_{2\ell+1})
\end{eqnarray*}
For example, we have that:
\begin{eqnarray*}
    \alpha_{2} & = & (\neg \msr(x_{1,1}, x_{1, 2}, x_1) \wedge \msr(x_{2,1}, x_{2, 2}, x_{2})) \\
    \alpha_{3} & = & (\neg \msr(x_{1,1}, x_{1, 2}, x_1) \wedge \msr(x_{2,1}, x_{2, 2}, x_{2})) \vee \neg \msr(x_{3,1}, x_{3, 2}, x_{3})
\end{eqnarray*}
By the definition of $\sate(\psi_1, \ldots, \psi_k)$, the fact that
$\sat(k)$ is $\bh_k$-complete, condition (C1), Lemma \ref{lem:bh-foilpm} and
the definition of $\phi_k(x_{1,1}, x_{1,2}, x_1, \ldots, x_{k,1}, x_{k,2},
x_k)$, we conclude that the expression $\sate(\psi_1, \ldots, \psi_k)$
holds if and only if $\T \models \phi_k(\es_{1,1}, \es_{1,2}, \es_1,
\ldots, \es_{k,1}, \es_{k,2}, \es_k)$. Therefore, given that the
decision tree $\T$ and the partial instances $\es_{1,1}$, $\es_{1,2}$,
$\es_1$, $\ldots$, $\es_{k,1}$, $\es_{k,2}$, $\es_k$ can be
constructed in polynomial time in the size of $(\psi_1, \ldots,
\psi_k)$, we conclude that $\bh_k$ can be reduced in polynomial time
to \logicEvaluation$(\phi_k)$. This completes the proof of the
theorem.

\subsection{Proof of Theorem \ref{thm:comp-qdtfoil}}
\label{app:comp-qdtfoil}
Considering the predicate $\cons$ defined in Section \ref{sec-dtfoil-c-v-1} and predicates $L_0(x)$ and $L_1(x)$ from Section \ref{app:def-suf}, let's define the following auxiliary formula in \dtfoil:
\begin{align*}
	\bound(u, v_0, v_1) &\equiv
		u \lel v_0 \wedge u \lel v_1 \wedge
		\forall w \Big(
			L_1(w) \rightarrow
			\big[
				\big(\neg(w \subseteq v_0) \wedge \neg (w \subseteq v_1)\big)
				\rightarrow \neg(w \subseteq u)
			\big]
		\Big)
\end{align*}
Using them we define the~\qdtfoil~formula
\begin{align*}
	\phi(u, v_0, v_1, v_2, v_3) \equiv
	\bound(u, v_0, v_1)
	\wedge
	\forall w \big(
		(u \subsetneq w \wedge \bound(w, v_2, v_3)
		\rightarrow
		\exists z (\leaf(z) \wedge \cons(w, z) \wedge \pos(z))
	\big)
\end{align*}
and the formula $\varphi(v_0, v_1, v_2, v_3) \equiv \exists u. \ \phi(u, v_0,
v_1, v_2, v_3)$.
\begin{lemma}\label{thm:exist-qdtfoil}
	Let $\alpha(\bar{x}, \bar{y})$ be a propositional formula in~\dnf. For
	every quantified propositional formula $\beta = \exists \bar{x} \forall
	\bar{y} \; \alpha( \bar{x}, \bar{y})$, there exists a decision tree
	$\T_{\alpha}$ and partial instances $c_0, c_1, c_2, c_3$ such that $\beta$
	is true if and only if $\T_{\alpha} \models \varphi(c_0, c_1, c_2, c_3)$.
\end{lemma}
\begin{proof}
	Let $\alpha(\bar{x}, \bar{y})$ be an arbitrary propositional formula
	in~\dnf~with $n$ clauses. We define the decision tree $\T_{\alpha}$ of
	dimension $d = n + |\bar{x}| + |\bar{y}|$ as follows:
	\begin{figure}[H]
		\centering
		\begin{tikzpicture}
			\node[draw=none] (T1) at (0, 0) {$t_1$};
			\node[draw=none] (T2) at (1.3, -1.3) {$t_2$};
			\node[draw=none] (Tdots) at (2.6, -2.6) {$\dots$};
			\node[draw=none] (Tn) at (3.9, -3.9) {$t_n$};
			\node[draw=none] (Zero) at (5.2, -5.2) {false};
			\node[draw, circle] (C1) at (-1.3, -1.3) {$C_1$};
			\node[draw, circle] (C2) at (0, -2.6) {$C_2$};
			\node[draw, circle] (Cn) at (2.6, -5.2) {$C_n$};
			\draw[->] (T1) -- (T2) node[midway, above] {$0$};
			\draw[->] (T1) -- (C1) node[midway, above] {$1$};
			\draw[->] (T2) -- (Tdots) node[midway, above] {$0$};
			\draw[->] (T2) -- (C2) node[midway, above] {$1$};
			\draw[->] (Tdots) -- (Tn) node[midway, above] {$0$};
			\draw[->] (Tn) -- (Cn) node[midway, above] {$1$};
			\draw[->] (Tn) -- (Zero) node[midway, above] {$0$};
		\end{tikzpicture}
	\end{figure}
	The inputs of $\T_{\alpha}$ are of the form $(t_1, \dots, t_n, x_1, \dots,
	x_{|\bar{x}|}, y_1, \dots, y_{|\bar{y}|})$ and $C_i$ corresponds to a
	decision tree representing a circuit equivalent to the $i$-th clause of
	$\alpha(\bar{x}, \bar{y})$ such that:
	\begin{enumerate}
		\item $C_i$ decides only over the features $x_1, \ldots, x_{|\bar{x}|}$
			and $y_1, \ldots, y_{|\bar{y}|}$.
		\item An input reaches a positive leaf of $C_i$ if and only if for
			every feature $f_c$ corresponding to a literal $c$ in $C_i$, it
			holds that $f_c = 0$ if $c$ is negated in $C_i$, and $f_c = 1$
			otherwise.
	\end{enumerate}
	Let
	\begin{align*}
		c_0 &= (\bot, \bot, \dots, \bot, 0, 0, \dots, 0, \bot, \bot, \dots, \bot) \\
		c_1 &= (\bot, \bot, \dots, \bot, 1, 1, \dots, 1, \bot, \bot, \dots, \bot) \\
		c_2 &= (\bot, \bot, \dots, \bot, 0, 0, \dots, 0, 0, 0, \dots, 0) \\
		c_3 &= (\bot, \bot, \dots, \bot, 1, 1, \dots, 1, 1, 1, \dots, 1)
	\end{align*}
	Note that for any partial instance $u = (u_1, u_2, \dots, u_d)$ such that
	$\T_{\alpha} \models \bound(u, c_0, c_1)$ the following two statements are
	always true: 
	\begin{enumerate}
		\item For every $i \in [n+1, n+|\bar{x}|]$: $u_i \in \{0, 1\}$.
		\item For every $j \not\in [n+1, n+|\bar{x}|]$: $u_j = \bot$.
	\end{enumerate}
	The same holds for $\bound(a, c_2, c_3)$ but in that case only
        the first $n$ features are $\bot$ while all the rest must be
        either 0 or 1. Next we prove that the statement of the lemma is correct.
        \begin{itemize}
	\item $\beta \; \text{is true} \Rightarrow \T_{\alpha} \models \varphi(c_0,
	c_1, c_2, c_3)$. Let's assume that there exists $\bar{x} = x_1x_2 \cdots
	x_{|\bar{x}|}$ such that for every $\bar{y} = y_1y_2 \cdots y_{|\bar{y}|}$
	the formula $\alpha(\bar{x}, \bar{y})$ is true. Take the partial instance
	$u = (\bot, \dots, \bot, x_1, x_2, \dots, x_{|\bar{x}|}, \bot, \dots,
	\bot)$ of dimension $d$. Note that $u$ satisfied $\bound(u, c_0, c_1)$.
	Thanks to the construction of $\T_{\alpha}$ we know that for every partial
	completion $w = (\bot, \dots, \bot, x_1, x_2, \dots, x_{|\bar{x}|}, y_1,
	y_2, \dots, y_{|\bar{y}|})$ of $u$ there must be at least one $i \in [1,
	n]$ such that $w$ reaches a positive leaf in $C_i$. Thus there must exist a
	leaf $z = (t_1, t_2, \dots, t_n, x'_1, x'_2, \dots, x'_{|\bar{x}|}, y'_1,
	y'_2, \dots, y'_{\bar{y}|})$ such that $\pos(z)$ and
	\begin{enumerate}
		\item $t_i = 1$.
		\item For every $j \neq i$: $t_j \in \{0, \bot\}$.
		\item For every $j \in [1, |\bar{x}|]$: $x'_j \in \{\bot, x_j\}$.
		\item For every $j \in [1, |\bar{y}|]$: $y'_j \in \{\bot, y_j\}$.
	\end{enumerate}
	Note then that $\cons(z, w)$ must be true because the first $n$ features of
	$w$ are $\bot$, as it satisfies $\bound(w, c_2, c_3)$ and every feature
	$z_i$ of $z$ past index $n$ is either $\bot$ or equal to $w_i$. Therefore
	$\T_{\alpha} \models \varphi(c_0, c_1, c_2, c_3)$.

	\item $\T_{\alpha} \models \varphi(c_0, c_1, c_2, c_3) \Rightarrow \beta
	\; \text{is true}$. Let's assume that there exists a partial instance $u$
	that satisfies $\bound(u, c_0, c_1)$, so that for every partial instance
	$w$ for which it holds that $u \subset w$ and $\bound(w, c_2, c_3)$ there
	exists a partial instance $z$ consistent with $w$ such that $\pos(z)$.
	In other words let's assume there exists a partial instance $u = (\bot,
	\dots, \bot, x_1, x_2, \dots, x_{|\bar{x}|}, \bot, \dots, \bot)$ of
	dimension $d$ so that for every partial completion $w = (\bot, \dots, \bot,
	x_1, x_2, \dots, x_l, y_1, y_2, \dots, y_k)$ of $u$ there exists $i \in [1,
	n]$ so that $w$ reaches a positive leaf in $C_i$. We then know that
	$\bar{x}' = x_1x_2 \cdots x_{|\bar{x}|}$ is such that for every $\bar{y}' =
	y_1y_2 \cdots y_{|\bar{y}|}$ the valuation $(\bar{x}', \bar{y}')$ satisfies
	at least one clause of the formula $\alpha(\bar{x}, \bar{y})$ and thus the
	entire formula. Therefore $\beta$ is true.
        \end{itemize}
	Hence, we conclude that $\beta$ is true if and only if $\T_{\alpha} \models
	\varphi(c_0, c_1, c_2, c_3)$, which was to be shown.
\end{proof}

It is known that the problem of determining if a quantified
propositional formula $\exists\bar{x}\forall\bar{y} \, \alpha(\bar{x},
\bar{y})$, with $\alpha(\bar{x},\bar{y})$ a formula in~\dnf, is
satisfiable is $\stp$-complete.  It is also easy to see that given
a~\dnf~formula $\alpha$, its corresponding decision tree $\T_\alpha$
can be constructed in polynomial time. Thus, we have by Lemma
\ref{thm:exist-qdtfoil} that the problem of determining, for any given
decision tree $\T$, if $\T\models\varphi(c_0, c_1, c_2, c_3)$ is also
$\stp$-complete.

Note that the aforementioned problem is equivalent to determining, for
any given decision tree $\T$, if $\T\models\exists u \, \phi(u, c_0,
c_1, c_2, c_3)$. Therefore, if an algorithm exists capable of
computing $u$ such that $\T\models\phi(u, c_0, c_1, c_2, c_3)$
in~\fpnp, and which outputs No if such an instance $u$ does not exist,
then it would be possible to determine whether $\T \models
\varphi(c_0, c_1, c_2, c_3)$ in $\ptime^\np$. But as the latter
problem is $\stp$-complete, we will have that $\stp = \ptime^\np$,
which implies that the polynomial hierarchy collapses to its second
level. Hence, as $\phi(u, c_0, c_1, c_2, c_3)$ is a~\qdtfoil~formula,
we have proven Theorem \ref{thm:comp-qdtfoil}.

\subsection{Proof of Proposition~\ref{prop-strict-po}}
\label{app:prop-strict-po}

Let $\rho(x, y, v_1, \ldots, v_\ell)$ be an arbitrary atomic formula, that is,
a formula over the vocabulary $\{\subseteq, \lel\}$, and let $m$ be the quantifier rank of $\rho$.  Then, let $\phi$ be the following sentance:
\begin{align*}
	\phi = \ \forall v_1 \cdots \forall v_\ell \, \big[ \forall x \, \neg
	\rho(x, x, v_1, \ldots, v_\ell) \ \wedge \forall x \forall y \forall z
	\,\big( (\rho(x, y, v_1, \ldots, v_\ell) \wedge \rho(y, z, v_1, \ldots,
	v_\ell)) \to \rho(x, z, v_1, \ldots, v_\ell)\big)\big], 
\end{align*}
and note that $\phi$~has quantifier rank at most $w := m + \ell + 3$.
By definition, determining if $\rho$ corresponds to a strict partial order, 
is equivalent to checking if for every decision tree $\T$, it holds that $\T \models \phi$.

As in the proof of~\Cref{theo:ptime-atomic}, we define $\mathfrak{B}_{n}$ for any $n \geq 1$ as the structure over the vocabulary $\{\subseteq, \lel\}$
corresponding to instances of dimension $n$, with $\mathfrak{B}_{n} = \mathfrak{B}_{\T}$ for every tree $\T$ of dimension $n$.

Thanks to the proof of~\Cref{lem-sentence-ef} we know that there exists a value $c \leq 3^{w}$, such that $\mathfrak{B}_n
\equiv_{w} \mathfrak{B}_{c}$ for every $n \geq c$. We can turn this observation into an algorithm for determining
if $\rho$ is a strict partial order. 
Indeed, it implies that it is enough to check if $\mathfrak{B}_{n} \models \phi$ for every $n \leq c$.
For each $n \leq c$ we will perform such a check by instantiating each of the universal and existential quantifiers in $\phi$; in other words, by brute force. Indeed, given that each variable in $\mathfrak{B}_n$ can take at most $|\{0, 1, \bot\}|^n = 3^n$ different values, and there are $w$ of them that must be instantiated at once, the number of possible valuations of the variables of $\phi$ is
 \[
	\left(3^n\right)^{w} = 3^{n\cdot w} \leq 3^{3^{w^2}}.
\]
For any fixed valuation of the variables of $\phi$, it is easy to check in polynomial time (over $|\phi| \in O(|\rho|)$) whether the quantifier-free part of $\phi$ is satisfied. Therefore, we conclude the whole procedure takes time 
\[
	|\phi|^{O(1)} \cdot 3^{3^{w^2}},
\] and as $w \leq O(|\phi|) \leq O(|\rho|)$, this is double exponential in $|\rho|$ as desired.

\subsection{Proof of Proposition \ref{prop-ep-optdtfoil}}
\label{app:prop-ep-optdtfoil}
For the first containment, each formula in $\dtfoil$ is a $\optdtfoil$ formula by definition. It is strict because in $\optdtfoil$ we can express \msr, but in $\dtfoil$ we can not unless $\ptime = \np$ (Proposition \ref{prop:ptime-guarded}).

For the second containment, let $\Psi[u_1, \ldots, u_k, v_1, \ldots, v_\ell](x) = \minf[\varphi[u_1, \ldots, u_k](x), \rho[v_1, \ldots, v_\ell](y, z)]$ be a $\optdtfoil$ formula. As shown in the semantics of $\optdtfoil$, by defining $\theta_{\minf}(x, u_1, \ldots, u_k, v_1, \ldots, v_\ell) = \varphi(x, u_1, \ldots, u_k) \ \wedge \forall y \, \big(\varphi(y, u_1, \ldots, u_k) \to
\neg \rho(y, x, v_1, \ldots, v_\ell)\big)$ we can express any $\optdtfoil$ formula as a $\qdtfoil$ formula. For sake of contradiction, let $\varphi$ be any $\qdtfoil$ and let $\Psi$ be its equivalent expression in $\optdtfoil$. By Theorem \ref{theo-comp-optdtfoil}, for every formula $\Psi$ in~\optdtfoil the problem~\logicComputation$(\Psi)$ is in \fpnp. Thus, \logicComputation$(\varphi)$ can be solved in \fpnp. But, by Theorem \ref{thm:comp-qdtfoil}, this would imply the collapse of the polynomial hierarchy to $\stp$. We conclude that the containment is strict.

\subsection{Proof of Theorem \ref{theo-comp-optdtfoil}}
\label{app:theo-comp-optdtfoil}
The proof of this result relies on the following property of the
strict partial orders defined by atomic $\dtfoil$ formulas. Given an
atomic $\dtfoil$ formula $\rho[v_1, \ldots, v_\ell](x,y)$, we say that
a sequence $(\es_1, \ldots, \es_k)$ of partial instances is a path of
dimension $n$ in $\rho[v_1, \ldots, v_\ell](x,y)$ if $\es_i$ is a
partial instance of dimension $n$ for each $i \in {1, \ldots, k}$, and
there exists a decision tree $\T$ of dimension $n$ and partial
instances $\es'_1$, $\ldots$, $\es'_\ell$ of dimension $n$ such that
$\T \models \rho[\es'_1, \ldots, \es'_\ell](\es_i,\es_{i+1})$ for
every $i \in \{1, \ldots, k-1\}$. The following lemma shows that the
lengths of such paths are polynomially bounded on $n$, assuming that
$\rho[v_1, \ldots, v_\ell](x,y)$ is a fixed formula.
\begin{lemma}\label{lem-fixed-rho}
Fix a formula $\rho[v_1, \ldots, v_\ell](x,y)$ that represents a
strict partial order. Then there exists a fixed polynomial $p$ such
that for every path $(\es_1, \ldots, \es_k)$ of dimension $n$ in
$\rho[v_1, \ldots, v_\ell](x,y)$, it holds that  $k \leq p(n)$.
\end{lemma}

\begin{proof}
First, we consider a formula that represents a strict partial order of
the form $\rho(x,y)$, and then we extend the proof to the case of a
formula $\rho[v_1, \ldots, v_\ell](x,y)$ with parameters $v_1$,
$\ldots$, $v_\ell$.

Fix a dimension $n$ and a decision tree $\T$ of dimension $n$, and assume that
$\mathfrak{B}_\T$ is a structure over the vocabulary $\{\subseteq,
\lel\}$ generated from $\T$ by removing the interpretation of the
predicate $\pos$. Then given a partial instance of dimension $n$,
define $\#_0(\es)$ as the number of occurrences of the symbol $0$ in
$\es$, and likewise for $\#_1(\es)$ and $\#_\bot(\es)$. Moreover, for
every $(p,q,r) \in \mathbb{N}^3$ such that $p + q + r = n$, define
\begin{align*}
  L_{(p,q,r)} \ = \ \{ \es \mid \es \text{ is a partial instance of
    dimension } n \text{ such that } \#_0(\es) = p,\, \#_1(\es) = q
  \text{ and } \#_\bot(\es) = r \}.
\end{align*}
Notice that there are at most $(n+1)^3$ sets $L_{(p,q,r)}$. In what
follows, we shows that if $\es_1$ and $\es_2$ are two distinct partial
instances of dimension $n$ such that $\es_1 \in L_{(p,q,r)}$ and
$\es_2 \in L_{(p, q, r)}$, then $\T \not\models \rho(\es_1,
\es_2)$. From this we conclude that the statement of the lemma holds
for $p(n) = (n+1)^3$, since if $(\es_1, \ldots, \es_k)$ is a path of
dimension $n$ in $\rho(x,y)$, then each $\es_i$ must belong to a
different set $L_{(p,q,r)}$.

For the sake of contradiction, assume that $\es_1 \neq \es_2$, $\es_1
\in L_{(p, q, r)}$, $\es_2 \in L_{(p,q,r)}$, and $\T \models
\rho(\es_1, \es_2)$. Given that $\es_1 \in L_{(p, q, r)}$ and $\es_2
\in L_{(p,q,r,)}$, there exists a permutation $\pi : \{1, \ldots, n\}
\to \{1, \ldots, n\}$ such that $\pi(\es_1) = \es_2$. For every pair
$\es$, $\es'$ of partial instances of dimension $n$, it holds that:
\begin{align*}
  \es \subseteq \es' &\quad\quad\Leftrightarrow\quad\quad   \pi(\es) \subseteq \pi(\es')\\
  \es \lel \es' &\quad\quad\Leftrightarrow\quad\quad   \pi(\es) \lel \pi(\es').
\end{align*}
Hence, $\pi$ is an automorphism for the structure $\mathfrak{B}_\T$,
from which we conclude that $\mathfrak{B}_\T \models \rho(\pi(\es_1),
\pi(\es_2))$ given that $\rho$ is a formula defined over the
vocabulary $\{ \subseteq, \lel\}$ and $\T \models \rho(\es_1 ,
\es_2)$. Thus, we have that $\T \models \rho(\pi(\es_1), \pi(\es_2))$
and, therefore, $\T \models \rho(\es_1, \pi^2(\es_1))$ since $\es_2 =
\pi(\es_1)$ and $\rho(x, y)$ represents a strict partial order. In the
same way, it is possible to conclude that $\T \models \rho(\es_1,
\pi^k(\es_1))$ for every $k \geq 1$. Given that the set of
permutations of $n$ elements with the composition operator forms a
group of order $n!$, we know that $\pi^{n!}$ is the identity
permutation, so that $\pi^{n!}(\es_1) = \es_1$. Therefore, we
conclude that $\T \models \rho(\es_1, \es_1)$, which leads to a
contradiction since $\rho(x,y)$ represents a strict partial order.

Consider now a formula $\rho[v_1, \ldots, v_\ell](x,y)$ that
represents a strict partial order. As in the previous case without
parameters, fix a dimension $n$ and a decision tree $\T$ of dimension
$n$, and assume that $\mathfrak{B}_\T$ is a structure over the
vocabulary $\{\subseteq, \lel\}$ generated from $\T$ by removing
the interpretation of the predicate $\pos$. Moreover, fix a sequence
$\es'_1$, $\ldots$, $\es'_\ell$ of instances of dimension $n$. Then,
for every $(a_1, \ldots, a_\ell) \in \{0,1,\bot\}^\ell$ and every
partial instance $\es$ of dimension $n$, define
\begin{align*}
  P_{(a_1, \ldots, a_\ell)} \ = \ \{i \in \{1, \ldots, n\} \mid \forall j \in \{1, \ldots, \ell\} \,:\, \es'_j[i] = a_i\}
\end{align*}
and
\begin{align*}
  \#_{(0, a_1, \ldots, a_\ell)}(\es) \ = \ |\{ i \in P_{(a_1, \ldots,
    a_\ell)} \mid \es[i]=0 \}|.
\end{align*}
That is, $\#_{(0, a_1, \ldots, a_\ell)}(\es)$ is the number of
occurrences of the symbol $0$ in $\es$ among the positions $P_{(a_1,
  \ldots, a_\ell)}$. The functions $\#_{(1, a_1, \ldots,
  a_\ell)}(\es)$ and $\#_{(\bot, a_1, \ldots, a_\ell)}(\es)$ are
defined in the same way. Finally, consider an arbitrary enumeration
$t_1$, $\ldots$, $t_{3^\ell}$ of the tuples in
$\{0,1,\bot\}^\ell$, and for every $(p_1, q_1, r_1, \ldots,
p_{3^\ell}, q_{3^\ell}, r_{3^\ell}) \in \mathbb{N}^{3^{\ell + 1}}$
define
\begin{multline*}
  L_{(p_1, q_1, r_1, \ldots, p_{3^\ell}, q_{3^\ell},r_{3^\ell})} \ =
  \ \{ \es \mid \es \text{ is a partial instance of dimension } n
  \text{ such that }\\ \#_{(0,t_i)}(\es) = p_i, \#_{(1,t_i)}(\es) =
  q_i, \text{ and } \#_{(\bot,t_i)}(\es) = r_i \text{ for every } i
  \in \{1, \ldots, 3^{\ell} \} \}.
\end{multline*}
Notice that there are at most $(n+1)^{3^{\ell + 1}}$ sets $L_{(p_1,
  q_1, r_1, \ldots, p_{3^\ell}, q_{3^\ell},r_{3^\ell})}$. In what
follows, we shows that if $\es_1$ and $\es_2$ are two distinct partial
instances of dimension $n$ such that $\es_1 \in L_{(p_1, q_1, r_1,
  \ldots, p_{3^\ell}, q_{3^\ell},r_{3^\ell})}$ and $\es_2 \in L_{(p_1,
  q_1, r_1, \ldots, p_{3^\ell}, q_{3^\ell},r_{3^\ell})}$, then $\T
\not\models \rho[\es'_1, \ldots, \es'_\ell](\es_1, \es_2)$. From this
we conclude, as in the previous case without parameters, that the
statement of the lemma holds for $p(n) = (n+1)^{3^{\ell+1}}$, as
$\ell$ is a fixed number since the formula $\rho[v_1, \ldots,
  v_\ell](x,y)$ is assumed to be fixed.

For the sake of contradiction, assume that $\es_1 \neq \es_2$, $\es_1
\in L_{(p_1, q_1, r_1, \ldots, p_{3^\ell}, q_{3^\ell},r_{3^\ell})}$,
$\es_2 \in L_{(p_1, q_1, r_1, \ldots, p_{3^\ell},
  q_{3^\ell},r_{3^\ell})}$, and $\T \models \rho[\es'_1, \ldots,
  \es'_\ell](\es_1, \es_2)$. Given that $\es_1 \in L_{(p_1, q_1, r_1,
  \ldots, p_{3^\ell}, q_{3^\ell},r_{3^\ell})}$ and $\es_2 \in L_{(p_1,
  q_1, r_1, \ldots, p_{3^\ell}, q_{3^\ell},r_{3^\ell})}$, there exist
permutations $\pi_i : P_{t_i} \to P_{t_i}$ for $i \in \{1, \ldots,
3^{\ell}\}$ such that the permutation $\pi : \{1, \ldots, n\} \to \{1,
\ldots, n\}$ obtained by combining them satisfies that $\pi(\es_1) =
\es_2$. Notice that by definition of the permutations $\pi_i$ and the
sets $P_{t_i}$, we have that $\pi(\es'_j) = \es'_j$ for every $j \in
\{1, \ldots, \ell\}$.

As in the previous case without parameters, we have that $\pi$ is an
automorphism for the structure $\mathfrak{B}_\T$. Hence, we have that
$\mathfrak{B}_\T \models \rho[\pi(\es'_1), \ldots,
  \pi(\es'_\ell)](\pi(\es_1), \pi(\es_2))$ given that $\rho$ is a
formula defined over the vocabulary $\{ \subseteq, \lel\}$ and $\T
\models \rho[\es'_1, \ldots, \es'_\ell](\es_1 , \es_2)$. Thus, we have
that $\T \models \rho[\pi(\es'_1), \ldots,
  \pi(\es'_\ell)][(\pi(\es_1), \pi(\es_2))$ and, therefore, $\T
  \models \rho[\es'_1, \ldots, \es'_\ell](\es_1, \pi^2(\es_1))$ since
  $\pi(\es'_j) = \es'_j$ for every $j \in \{1, \ldots, \ell\}$, $\es_2
  = \pi(\es_1)$ and $\rho(x, y)$ represents a strict partial order. In
  the same way, it is possible to conclude that $\T \models
  \rho[\es'_1, \ldots, \es'_\ell](\es_1, \pi^k(\es_1))$ for every $k
  \geq 1$. Given that the set of permutations of $n$ elements with the
  composition operator forms a group of order $n!$, we know that
  $\pi^{n!}$ is the identity permutation, so that $\pi^{n!}(\es_1) =
  \es_1$. Therefore, we conclude that $\T \models \rho[\es'_1, \ldots,
    \es'_\ell](\es_1, \es_1)$, which leads to a contradiction since
  $\rho[v_1, \ldots, v_\ell](x,y)$ represents a strict partial order.
\end{proof}

Lemma \ref{lem-fixed-rho} gives us a simple algorithm to compute a solution for an
$\optdtfoil$ formula
\begin{align*}
\minf[\varphi[u_1, \ldots, u_k](x), \rho[v_1, \ldots, v_\ell](y, z)],
\end{align*}
given as input a decision
tree $\T$ of dimension $n$ and partial instances $\es'_1$, $\ldots$,
$\es'_k$, $\es''_1$, $\ldots$, $\es''_\ell$ of dimension $n$. The
algorithm first uses an $\np$ oracle to verify whether
$\T \models \exists x \, \varphi[\es'_1, \ldots, \es'_k](x)$. If
$\T \not\models \exists x \, \varphi[\es'_1, \ldots, \es'_k](x)$, then
the algorithm answers No. Otherwise, the algorithm uses the $\np$
oracle to construct an initial partial instance $\es_0$ such that
$\T \models \varphi[\es'_1, \ldots, \es'_k](\es_0)$. Then
it uses the $\np$ oracle to verify whether:
\begin{align*}
\T \ \models \ \exists x \, \big(\varphi[\es'_1, \ldots, \es'_k](x) \wedge \rho[\es''_1, \ldots, \es''_\ell](x, \es_0)\big),
\end{align*}
and if this is the case, then it also uses the $\np$ oracle to
construct a partial instance $\es_1$ such that
\begin{align*}
\T \ \models \ \varphi[\es'_1, \ldots, \es'_k](\es_1) \wedge \rho[\es''_1, \ldots, \es''_\ell](\es_1, \es_0).
\end{align*}
The algorithm continues in this way, constructing a sequence of
partial instances $(\es_i, \es_{i-1}, \ldots, \es_0)$ that constitutes
a path of dimension $n$ in $\rho[v_1, \ldots, v_\ell](y, z)$. The
algorithm stops when the condition
\begin{align*}
\T \ \models \ \exists x \, \big(\varphi[\es'_1, \ldots, \es'_k](x) \wedge \rho[\es''_1, \ldots, \es''_\ell](x, \es_i)\big)
\end{align*}
does not hold, which guarantees that $\es_i$ is a minimal
instance. Lemma \ref{lem-fixed-rho} guarantees that $\es_i$ will be
found in a polynomial number of steps, thus showing that the statement
of Theorem \ref{theo-comp-optdtfoil} holds.


	
\end{document}